\pdfoutput=1
\newcommand{\submission}{} %

\ifdefined\camerareadyversion
\documentclass[acmsmall,review]{acmart}
\else
\documentclass[acmsmall,nonacm]{acmart}
\fi
\newboolean{arxiv}
\newboolean{cameraready}
\ifdefined\camerareadyversion
  \setboolean{arxiv}{false}
  \setboolean{cameraready}{true}
\else
  \setboolean{arxiv}{true}
  \setboolean{cameraready}{false}
\fi

\AtBeginDocument{%
  }

\setcopyright{none}
\renewcommand\footnotetextcopyrightpermission[1]{}

\usepackage{CJKutf8} %
\usepackage{ascmac} %
\usepackage{stmaryrd} %

{\begin{screen}\begin{CJK}{UTF8}{ipxm}}%
{\end{CJK}\end{screen}}

\usepackage{xparse}
\usepackage{ifthen}

\usepackage{amsmath}
\usepackage{amsfonts%
}
\usepackage{tabularray} %

\newcommand{\logicdash}{\mathord{\text{-}}}

\usepackage{amsthm}

\usepackage{hyperref}
\usepackage[nosort]{cleveref}
\makeatletter
\def\pdfCref#1{\ifcsname r@#1@cref\endcsname
  \expandafter\expandafter\expandafter\pdfCref@a\csname r@#1@cref\endcsname\@nil
  \else ??\fi}
\def\pdfCref@a#1#2\@nil{\pdfCref@b#1\@nil}
\def\pdfCref@b[#1][#2][#3]#4\@nil{\csname Cref@#1@name\endcsname\space#4}
\makeatother

\newtheorem{theorem}{Theorem}
\newtheorem{lemma}[theorem]{Lemma}
\newtheorem{proposition}[theorem]{Proposition}
\newtheorem{corollary}[theorem]{Corollary}
\newtheorem{definition}[theorem]{Definition}
\newtheorem*{proposition*}{Proposition}
\newtheorem*{theorem*}{Theorem}
\newtheorem*{lemma*}{Lemma}
\newtheorem*{corollary*}{Corollary}

\newtheorem*{claim*}{Claim}

\newtheorem*{hypothesis*}{Hypothesis}

\theoremstyle{remark}
\newtheorem{remark}[theorem]{Remark}
\newtheorem*{remark*}{Remark}

\newtheorem*{problem*}{Problem}
\newtheorem{question}{Question}
\newtheorem*{question*}{Question}
\newtheorem*{example*}{Example}
\newtheorem*{examples*}{Examples}

\crefname{theorem}{Theorem}{Theorems}
\crefname{proposition}{Proposition}{Propositions}
\crefname{apxproposition}{Proposition}{Propositions}
\crefname{lemma}{Lemma}{Lemmas}
\crefname{corollary}{Corollary}{Corollaries}
\crefname{definition}{Definition}{Definitions}
\crefname{section}{\S}{\S\S}
\crefname{figure}{Fig.\@}{Figs.\@}
\crefname{hypothesis}{Hypothesis}{Hypotheses}
\crefname{table}{Table}{Tables}
\crefname{problem}{Problem}{Problems}
\crefname{appendix}{Appendix}{Appendices}
\crefname{footnote}{Footnote}{Footnotes}

\AddToHook{env/theorem/begin}{\crefalias{theorem}{theorem}}
\AddToHook{env/proposition/begin}{\crefalias{theorem}{proposition}}
\AddToHook{env/lemma/begin}{\crefalias{theorem}{lemma}}
\AddToHook{env/corollary/begin}{\crefalias{theorem}{corollary}}
\AddToHook{env/definition/begin}{\crefalias{theorem}{definition}}
\AddToHook{env/remark/begin}{\crefalias{theorem}{remark}}
\AddToHook{env/hypothesis/begin}{\crefalias{theorem}{hypothesis}}
\AddToHook{env/problem/begin}{\crefalias{theorem}{problem}}
\AddToHook{env/table/begin}{\crefalias{theorem}{table}}

\usepackage[table]{xcolor}
\usepackage{subcaption}

\ifdefined\submission
\newcommand{\unote}[1]{}
\newcommand{\rednote}[1]{}
\newcommand{\bluenote}[1]{}
\else
\newcommand{\unote}[1]{{\color[HTML]{D70040} [Uezato: #1]}}
\newcommand{\rednote}[1]{{\color[HTML]{FF2400} #1}}
\newcommand{\bluenote}[1]{{\color[HTML]{3F00FF} #1}}
\fi

\usepackage{diagbox}

\MakeRobust{\ref}%
\usepackage{tikz}
\usetikzlibrary{cd, decorations.pathmorphing, tikzmark, calc, trees, topaths, bending, backgrounds, shapes, shapes.geometric, fit, arrows, arrows.meta, positioning, graphs, quotes}
\usetikzlibrary{arrows.meta,calc,decorations.pathreplacing,positioning}
\pgfdeclarelayer{bg}    %
\pgfsetlayers{bg,main}  %
\tikzstyle{mynode} = [inner sep = 1.5pt, fill= gray!20,  font = \scriptsize, draw, circle]
\tikzstyle{mysmallnode} = [inner sep = 1.pt, fill= gray!20, font = \scriptsize, draw, circle]
\tikzset{earrow/.style={>={{[flex] Latex[length=.08cm, width=2.5pt]}}}}

\usepackage[most]{tcolorbox} %
\usepackage{mathtools} %
\usepackage{environ} %
\ifdefined\submission
\usepackage[xcolor, cleveref, notion, quotation, electronic]{knowledge} %
\knowledgeconfigure {visible anchor points=false}
\else
\usepackage[xcolor, cleveref, notion, quotation]{knowledge}
\fi
\usepackage{mathcommand}
\knowledgeconfigure{protect quotation={graph,tikzpicture,tikzcd}} %
\knowledgeconfigure{diagnose line=true, diagnose bar=true}

\definecolor{Dark Ruby Red}{HTML}{580507}
\definecolor{Dark Blue Sapphire}{HTML}{053641}
\definecolor{Dark Gamboge}{HTML}{be7c00}

\IfKnowledgePaperModeTF{
}{
    \knowledgestyle{intro notion}{color={Dark Ruby Red}, emphasize}
    \knowledgestyle{notion}{color={Dark Blue Sapphire}}
    \IfKnowledgeElectronicModeTF{
    }{
    }
}
\ifdefined\submission
\NewEnviron{hidden}{} %
\else
\NewEnviron{hidden}{\color{gray!60} \BODY} %
\fi

\ifdefined\submission
\usepackage[backgroundcolor=orange!20, textcolor={Dark Ruby Red}, textsize=tiny, disable]{todonotes}
\else
\usepackage[backgroundcolor=orange!20, textcolor={Dark Ruby Red}, textsize=tiny]{todonotes}
\fi

\ifdefined\submission
\usepackage[final,commandnameprefix=always]{changes}
\else
\usepackage[commandnameprefix=always]{changes}
\fi

\NewEnviron{commentenv}{\todo[inline,caption={}, size=\footnotesize]{\BODY}}
\NewEnviron{sidecommentenv}{\todo[caption={}, size=\footnotesize]{\BODY}} 

\definecolor{yoshikieditcolor}{RGB}{100,230,200}
\definecolor{uezatoeditcolor}{RGB}{230,100,200}
\NewEnviron{sideyoshiki}{\todo[caption={},backgroundcolor=yoshikieditcolor, size=\tiny]{YN: \BODY}}
\NewEnviron{yoshiki}{\todo[inline,caption={},backgroundcolor=yoshikieditcolor, size=\footnotesize]{YN: \BODY}} %

\definechangesauthor[name={Yoshiki Nakamura}, color=yoshikieditcolor!80!black]{YN}
\definechangesauthor[name={Yuya Uezato}, color=uezatoeditcolor!80!black]{YU}

\definecolor{themecolor}{named}{gray}
\newcommand{\leftsp}{0.3cm}
\newtcolorbox{myquote}[1][]{
enhanced,
colback=themecolor!10,
coltitle=black,
colbacktitle=themecolor!3,
arc=0pt,
frame empty,
overlay={%
\begin{scope}[shift={(interior.north west)}]
\begin{scope}[shift={(interior.north west)}]
\filldraw [themecolor!60] (\leftsp,0.1cm) -- ++(-0.075cm,-0.15cm) -- ++(0,-0.06cm) -- ++(0.08cm,0) -- ++(0,0.06cm) -- ++(0.055cm,0.155cm) --cycle;
\filldraw [themecolor!60] ($(\leftsp,0.1cm)+(0.2cm,0)$) -- ++(-0.075cm,-0.15cm) -- ++(0,-0.06cm) -- ++(0.08cm,0) -- ++(0,0.06cm) -- ++(0.055cm,0.155cm) --cycle;
\end{scope}
\end{scope}
},
#1
}
\renewcommand{\epsilon}{\varepsilon}
\renewcommand{\phi}{\varphi}

\newcommand{\proofcasebase}[1]{\colorbox{black!10}{\vphantom{()}\hspace{.2em}\textsf{#1}\hspace{.2em}}~}
\newcommand{\proofcase}[1]{\noindent\setlength{\fboxsep}{0pt}\proofcasebase{#1}}
\newcommand{\subproofcase}[1]{\noindent\setlength{\fboxsep}{0pt}\quad\proofcasebase{#1}}

\newcommand{\const}[1]{\mathsf{#1}}

\newcommand{\defeq}{\coloneq}

\newcommand{\mathsc}[1]{\text{\normalfont\textsc{#1}}}

\newcommand{\eps}{\varepsilon}
\newcommand{\asympO}{O}
\newcommand{\asympo}{o}

\newcommand{\asympOmega}{\Omega}
\newcommand{\asympTheta}{\Theta}
\knowledgenewrobustcmd{\poly}{\cmdkl{\operatorname{poly}}}
\knowledgenewrobustcmd{\tower}{\cmdkl{\operatorname{tower}}}
\knowledgenewrobustcmd{\matrixexponent}{\cmdkl{\omega}}
\knowledgenewrobustcmd{\matrixexponentupperbound}{\cmdkl{2.3714}}

\DeclarePairedDelimiter{\set}{\{}{\}}
\knowledgenewrobustcmd{\nat}{\mathbb{\cmdkl{N}}}
\knowledgenewrobustcmd{\card}[1]{\cmdkl{|}#1\cmdkl{|}}

\DeclarePairedDelimiter{\tuple}{(}{)}
\knowledgenewrobustcmd{\range}[1]{\cmdkl{[}#1\cmdkl{]}}
\knowledgenewrobustcmd{\ith}[1]{\cmdkl{[}#1\cmdkl{]}}

\knowledgenewrobustcmd{\GIP}[1]{\cmdkl{\langle}#1\cmdkl{\rangle}} %

\NewDocumentCommand\graph{O{1}}{%
    \ifcase#1 undefined
    \or G
    \or H
    \else undefined \fi
}

\knowledgenewrobustcmd{\sig}{\cmdkl{\const{S}}} %
\knowledgenewrobustcmd{\ordinalomega}{\cmdkl{\infty}} 
\knowledgenewrobustcmd{\arity}[1]{\cmdkl{\operatorname{ar}}_{#1}}
\knowledgenewrobustcmd{\sigarity}[1]{\cmdkl{\operatorname{ar}}_{#1}}

\NewDocumentCommand\struc{O{1}}{%
    \ifcase#1 undefined
    \or \mathbb{A}
    \or \mathbb{B}
    \else undefined \fi
}
\knowledgenewrobustcmd{\univ}[1]{\cmdkl{\mathcal{U}}_{#1}} %

\knowledgenewrobustcmd{\strucn}[1]{\cmdkl{n}_{#1}} %
\knowledgenewrobustcmd{\strucm}[1]{\cmdkl{m}_{#1}} %

\DeclarePairedDelimiter{\sem}{\llbracket}{\rrbracket}
\NewDocumentCommand\val{O{1}}{%
    \ifcase#1 undefined
    \or I
    \else undefined \fi
}

\NewDocumentCommand\fml{O{1}}{%
    \ifcase#1 undefined
    \or \varphi
    \or \psi
    \or \rho
    \else undefined \fi
}
\knowledgenewrobustcmd{\fmltrue}{\cmdkl{\const{T}}}
\knowledgenewrobustcmd{\fmlfalse}{\cmdkl{\const{F}}}
\knowledgenewrobustcmd{\fmlneq}{\mathbin{\cmdkl{\neq}}}

\knowledgenewrobustcmd{\fmlsig}[1]{\cmdkl{\const{S}_{#1}}}
\knowledgenewrobustcmd{\fmlarity}[1]{\cmdkl{\operatorname{ar}}_{#1}}
\knowledgenewrobustcmd{\FV}{\cmdkl{\operatorname{FV}}}
\knowledgenewrobustcmd{\vars}{\cmdkl{\mathbf{Var}}}

\knowledgenewrobustcmd{\len}[1]{\cmdkl{|}#1\cmdkl{|}} %
\knowledgenewrobustcmd{\lendag}[1]{\cmdkl{|}#1\cmdkl{|}_{\mathrm{DAG}}}

\knowledgenewrobustcmd{\qw}{\cmdkl{\operatorname{qw}}}

\knowledgenewrobustcmd{\subw}{\cmdkl{\operatorname{subw}}}

\NewDocumentCommand\fmlclass{O{1}}{%
    \ifcase#1 undefined
    \or \mathcal{F}
    \else undefined \fi
}
\knowledgenewrobustcmd{\sigres}[1]{\cmdkl{[}#1\cmdkl{]}}

\knowledgenewrobustcmd{\FO}{\ensuremath{\mathsf{\cmdkl{FO}}}}
\knowledgenewrobustcmd{\kFO}[1]{\ensuremath{\mathsf{\cmdkl{FO}}^{#1}}}
\knowledgenewrobustcmd{\PNFkFO}[1]{\ensuremath{\cmdkl{\mathsf{PNF} \logicdash \mathsf{FO}}^{#1}}}

\knowledgenewrobustcmd{\kPP}[1]{\ensuremath{\mathsf{\cmdkl{PP}}^{#1}}}
\knowledgenewrobustcmd{\PNFkPP}[1]{\ensuremath{\cmdkl{\mathsf{PNF} \logicdash \mathsf{PP}}^{#1}}}

\knowledgenewrobustcmd{\kEP}[1]{\ensuremath{\mathsf{\cmdkl{EP}}^{#1}}}
\knowledgenewrobustcmd{\PNFkEP}[1]{\ensuremath{\cmdkl{\mathsf{PNF} \logicdash \mathsf{EP}}^{#1}}}

\knowledgenewrobustcmd{\kPFO}[1]{\ensuremath{\mathsf{\cmdkl{PFO}}^{#1}}}
\knowledgenewrobustcmd{\CoR}{\ensuremath{\cmdkl{\mathrm{CoR}}}}
\knowledgenewrobustcmd{\CoRtwo}{\ensuremath{\cmdkl{\mathrm{CoR}^{2}}}}
\knowledgenewrobustcmd{\CoRtwostar}{\ensuremath{\cmdkl{\mathrm{CoR}^{2, \star}}}}
\knowledgenewrobustcmd{\PPCoR}{\ensuremath{\cmdkl{\mathrm{PPCoR}}}}

\NewDocumentCommand\term{O{1}}{%
    \ifcase#1
        undefined
    \or t
    \or s
    \or u
    \else undefined
    \fi
}
\knowledgenewrobustcmd{\union}{\mathbin{\cmdkl{+}}}
\knowledgenewrobustcmd{\intersection}{\mathbin{\cmdkl{\cap}}}
\knowledgenewrobustcmd{\compo}{\mathbin{\cmdkl{;}}}

\knowledgenewrobustcmd{\cortop}{\cmdkl{\top}}
\knowledgenewrobustcmd{\corz}{\cmdkl{\bot}}
\knowledgenewrobustcmd{\corid}{\cmdkl{\mathsf{id}}}

\knowledgenewrobustcmd{\kC}[1]{\ensuremath{\mathsf{\cmdkl{C}}}^{#1}}
\knowledgenewrobustcmd{\BMLP}{\ensuremath{\cmdkl{\mathcal{ML}^{\star}}}}
\knowledgenewrobustcmd{\BMLPwithC}{\ensuremath{\cmdkl{\mathcal{ML}^{\star, \complement}}}}
\knowledgenewrobustcmd{\Atoms}{\cmdkl{\mathit{At}}}
\knowledgenewrobustcmd{\binrep}{\cmdkl{\operatorname{bin}}}
\knowledgenewrobustcmd{\pathfml}{\cmdkl{\theta}}%
\knowledgenewrobustcmd{\pathfmltree}{\cmdkl{\theta}}%
\knowledgenewrobustcmd{\hyperedgefml}{\cmdkl{\mathcal{E}}}
\knowledgenewrobustcmd{\DgOne}{\cmdkl{D_{\mathrm{g}}}}%
\knowledgenewrobustcmd{\DnOne}{\cmdkl{D_{\mathrm{n}}}}
\knowledgenewrobustcmd{\DgTwo}{\cmdkl{D_{\mathrm{g}}}}%
\knowledgenewrobustcmd{\DnTwo}{\cmdkl{D_{\mathrm{n}}}}
\knowledgenewrobustcmd{\HOM}{\cmdkl{\mathrm{HOM}}}
\knowledgenewrobustcmd{\MLLSW}{\ensuremath{\cmdkl{\mathcal{ML}^{\lnot, \cap, \cup, -, \mathit{id}}}}}

\knowledgenewrobustcmd{\dia}[1]{\cmdkl{\langle}#1\cmdkl{\rangle}}
\knowledgenewrobustcmd{\bo}[1]{\cmdkl{[}#1\cmdkl{]}}
\knowledgenewrobustcmd{\windowdia}[1]{\cmdkl{\langle\mkern-1mu\lvert}#1\cmdkl{\rvert\mkern-1mu\rangle}}
\knowledgenewrobustcmd{\windowbo}[1]{\cmdkl{[\mkern-1mu\lvert}#1\cmdkl{\rvert\mkern-1mu]}}
\NewDocumentCommand\mlfml{O{1}}{%
    \ifcase#1 undefined
    \or \theta
    \or \lambda
    \else undefined \fi
}
\knowledgenewrobustcmd{\Gaifman}{\cmdkl{E}}

\knowledgenewrobustcmd{\Varset}{\cmdkl{\mathcal{V}}}
\knowledgenewrobustcmd{\type}{\cmdkl{\operatorname{tp}}}
\knowledgenewrobustcmd{\Realizable}{\cmdkl{\mathrm{Rlz}}}
\knowledgenewrobustcmd{\Nclass}{\cmdkl{N}}

\knowledgenewrobustcmd{\COVER}{\cmdkl{\operatorname{Co}}}
\knowledgenewrobustcmd{\Zeropos}{\cmdkl{\mathcal{Z}}}

\knowledgenewrobustcmd{\Proj}{\cmdkl{\pi}}
\knowledgenewrobustcmd{\IsV}{\cmdkl{\operatorname{IsV}}}
\knowledgenewrobustcmd{\EQ}{\cmdkl{\mathrm{EQ}}}
\knowledgenewrobustcmd{\sigEQ}{\cmdkl{\const{S}'}}
\knowledgenewrobustcmd{\Gdd}{\cmdkl{\operatorname{Gdd}}}

\knowledge{text={i.e.\@}, italic}
  | ie
  | i.e.

\knowledge{text={I.e.\@}, italic}
  | Ie

\knowledge{text={s.t.\@}, italic}
  | st
  | s.t.

\knowledge{text={e.g.\@}, italic}
  | eg
  | e.g.

\knowledge{text={E.g.\@}, italic}
  | Eg
  | E.g.

\knowledge{text={vs.\@}, italic}
  | vs

\knowledge{text={w.r.t.\@}, italic}
  | wrt
  | w.r.t

\knowledge{text={a.k.a.\@}, italic}
  | aka
  | a.k.a

\knowledge{text={w.l.o.g.\@}, italic}
  | wlog
  | w.l.o.g.

\knowledge{text={W.l.o.g.\@}, italic}
  | Wlog
  | W.l.o.g.

\knowledge{text={cf.\@}, italic}
  | cf
  | cf.

\knowledge{text={Cf.\@}, italic}
  | Cf
  | Cf.

\knowledge{text={iff}, italic}
  | iff

\knowledge{text={r.e.\@}, italic}
  | r.e.
  | re

\knowledge{text={resp.\@}, italic}
  | resp

\knowledge{notion}
| word-RAM model

\knowledge{text={\ensuremath{\text{\textsc{PSpace}}}}}
  | PSPACE
  | PSpace

\knowledge{text={\ensuremath{\text{\textsc{NP}}}}}
  | NP

\knowledge{text={\ensuremath{\text{\textsc{NExpTime}}}}}
  | NEXPTIME
  | NExpTime

\knowledge{text={\ensuremath{\text{\textsc{FPT}}}}}
  | FPT

\knowledge{text={\ensuremath{\text{\textsc{AW}[$*$]}}}}
  | AW[*]

\knowledge{notion}
| equivalence relation
| equivalence relations

\knowledge{notion}
| equivalence class
| equivalence classes

\knowledge{notion}
| element
| elements

\knowledge{notion}
| choice function

\knowledge{notion}
| vector
| vectors

\knowledge{notion}
| dimension
| dimensions
| dimensional

\knowledge{notion}
| word
| words

\knowledge{notion}
| graph
| graphs

\knowledge{notion}
| vertex
| vertices

\knowledge{notion}
| edge
| edges

\knowledge{notion}
| nowhere dense

\knowledge{notion}
| treewidth

\knowledge{notion}
| twin-width

\knowledge{notion}
| graph homomorphism problem

\knowledge{notion}
 | graph homomorphism
 | graph homomorphisms
 | homomorphism

\knowledge{notion}
 | colored subgraph isomorphism problem
 | colored (not necessarily induced) subgraph isomorphism problem

\knowledge{notion}
| DAG
| DAGs
| directed acyclic graph
| directed acyclic graphs

\knowledge{notion}
| SETH
| strong exponential time hypothesis
| Strong Exponential Time Hypothesis

\knowledgenewrobustcmd{\CNFSATProblem}[1]{\cmdkl{\textsc{CNF-SAT}}}
\knowledgenewrobustcmd{\CNFSATHypothesis}[1]{\cmdkl{\textsc{CNF-SAT}~hypothesis}}
\knowledgenewrobustcmd{\SparseCNFSATHypothesis}[1]{\cmdkl{\textsc{Sparse CNF-SAT}~hypothesis}}
\knowledgenewrobustcmd{\MAXSATProblem}[1]{\cmdkl{\textsc{MAX-SAT}}}
\knowledgenewrobustcmd{\MAXSATHypothesis}[1]{\cmdkl{\textsc{MAX-SAT}~hypothesis}}
\knowledgenewrobustcmd{\MAXxSATProblem}[1]{\cmdkl{\textsc{MAX}-\ensuremath{#1}-\textsc{SAT}}}
\knowledgenewrobustcmd{\MAXxSATHypothesis}[1]{\cmdkl{\textsc{MAX}-\ensuremath{#1}-\textsc{SAT}~hypothesis}}
\knowledgenewrobustcmd{\SparseMAXxSATProblem}[1]{\cmdkl{\textsc{Sparse MAX}-\ensuremath{#1}-\textsc{SAT}}}
\knowledgenewrobustcmd{\SparseMAXxSATHypothesis}[1]{\cmdkl{\textsc{Sparse MAX}-\ensuremath{#1}-\textsc{SAT}~hypothesis}}

\knowledge{notion}
| CNF
| conjunctive normal form

\knowledge{notion}
| DNF
| disjunctive normal form

\knowledge{notion}
| satisfiable
| satisfiability

\knowledge{notion}
 | propositional formula
 | propositional formulas
 | formula@propositional
 | formulas@propositional

\knowledge{notion}
 | variable@propositional
 | variables@propositional

\knowledge{notion}
 | clause
 | clauses
 | clause@propositional
 | clauses@propositional

\knowledge{notion}
 | literal
 | literals
 | literal@propositional
 | literals@propositional

\knowledgenewrobustcmd{\OVProblem}[1]{\cmdkl{$\mathsc{OV}_{#1}$}}
\knowledgenewrobustcmd{\OVHypothesis}[1]{\cmdkl{$\mathsc{OV}_{#1}$~hypothesis}}
\knowledge{notion}
| generalized-orthogonal
| orthogonal

\knowledge{notion}
| hyperclique problem

\knowledgenewrobustcmd{\HyperCliqueProblem}[1]{\cmdkl{\textsc{HyperClique}\ensuremath{_{#1}}}}
\knowledgenewrobustcmd{\HyperCliqueHypothesis}[1]{\cmdkl{\textsc{HyperClique}\ensuremath{_{#1}}~hypothesis}}

\knowledge{notion}
| hypergraph
| hypergraphs

\knowledge{notion}
| uniform hypergraph
| uniform hypergraphs

\knowledge{notion}
| hyperedge
| hyperedges

\knowledge{notion}
| hyperclique

\knowledge{notion}
| OV-hypergraph
| OV-hypergraphs
| uniform OV-hypergraph
| uniform OV-hypergraphs

\knowledgenewrobustcmd{\OVgraph}[1]{\cmdkl{\graph}_{#1}}

\knowledge{notion}
| OV-hyperclique problem

\knowledge{notion}
| well-formed

\knowledgenewrobustcmd{\OVHyperCliqueProblem}[1]{\cmdkl{\textsc{OV-HyperClique}\ensuremath{_{#1}}}}
\knowledgenewrobustcmd{\OVHyperCliqueHypothesis}[1]{\cmdkl{\textsc{OV-HyperClique}\ensuremath{_{#1}}~hypothesis}}
\knowledgenewrobustcmd{\OVTriangleProblem}[1]{\cmdkl{\textsc{OV}\ensuremath{\Delta_{#1}}}}
\knowledgenewrobustcmd{\OVTriangleHypothesis}[1]{\cmdkl{\textsc{OV}\ensuremath{\Delta_{#1}}~hypothesis}}

\knowledgenewrobustcmd{\ExactIPTriangleProblem}{\cmdkl{\textsc{Exact-IP}\ensuremath{\Delta}}}
\knowledgenewrobustcmd{\ExactIPHyperCliqueProblem}{\cmdkl{\textsc{Exact-IP-HyperClique}}}
\knowledgenewrobustcmd{\PartitionedOVHyperCliqueProblem}[1]{\cmdkl{\textsc{Partitioned-OV-HyperClique}\ensuremath{_{#1}}}}
\knowledgenewrobustcmd{\PartitionedOVHyperCliqueHypothesis}[1]{\cmdkl{\textsc{Partitioned-OV-HyperClique}\ensuremath{_{#1}}~hypothesis}}

\knowledgenewrobustcmd{\PartitionedOVTriangleProblem}[1]{\cmdkl{\textsc{Partitioned-OV}\ensuremath{\Delta_{#1}}}}
\knowledgenewrobustcmd{\PartitionedOVTriangleHypothesis}[1]{\cmdkl{\textsc{Partitioned-OV}\ensuremath{\Delta_{#1}}~hypothesis}}

\knowledgenewrobustcmd{\CliqueHypothesis}[1]{\cmdkl{\textsc{Clique}\ensuremath{_{#1}}~hypothesis}}

\knowledgenewrobustcmd{\MC}[1]{\ensuremath{#1}\logicdash\cmdkl{\text{\normalfont \textsc{MC}}}}
\knowledgenewrobustcmd{\CoRMC}[1]{\ensuremath{#1}\logicdash\cmdkl{\text{\normalfont \textsc{MC}}}}
\knowledgenewrobustcmd{\BMLPMC}[1]{\ensuremath{#1}\logicdash\cmdkl{\text{\normalfont \textsc{MC}}}}

\knowledge{notion}
| model checking
| model-checking problem
| first-order model-checking problem
| first-order model checking

\knowledge{notion}
| database
| databases

\knowledge{notion}
| combined complexity
| Combined complexity

\knowledge{notion}
| signature
| signatures
| relational signature
| relational signatures

\knowledge{notion}
| relation symbol
| relation symbols

\knowledge{notion}
| arity
| Arity
| arities

\knowledge{notion}
| structure
| structures
| relational structure
| relational structures

\knowledge{notion}
| universe

\knowledge{notion}
| relation
| relations

\knowledge{notion}
| valuation
| valuations

\knowledge{notion}
| truth value

\knowledge{notion}
| semantically equivalent
| semantic equivalence

\knowledge{notion}
| expressive power

\knowledge{notion}
 | structure-homomorphism problem

 \knowledge{notion}
 | structure-homomorphism
 | homomorphism@structure

 \knowledge{notion}
 | constraint satisfaction problem
 | CSP

\knowledge{notion}
| query enumeration problem

\knowledge{notion}
| variable
| variables
| variable@fo
| variables@fo

\knowledge{notion}
| first-order
| first-order formula
| first-order formulas
| first-order query
| first-order queries
| FO formula
| FO formulas
| FO query
| FO queries
| formula
| formulas
| query
| queries

\knowledge{notion}
| first-order sentence
| first-order sentences
| sentence
| sentences

\knowledge{notion}
| primitive positive
| primitive positive formulas
| primitive positive formula
| primitive positive queries
| primitive positive query
| PP
| PP queries
| PP query
| PP formulas
| PP formula
| Primitive positive

\knowledge{notion}
| conjunctive queries
| conjunctive query

\knowledge{notion}
| hyperclique query

\knowledge{notion}
| Existential positive

\knowledge{notion}
| self-joins
| self-join
| self-join-free

\knowledge{notion}
| subformula
| subformulas

\knowledge{notion}
| subterm
| subterms

\knowledge{notion}
| size
| tree-size

\knowledge{notion}
| tree-represented

\knowledge{notion}
| DAG-size

\knowledge{notion}
| DAG-represented

\knowledge{notion}
| equality
| equalities

\knowledge{notion}
| inequality
| inequalities

\knowledge{notion}
| prenex normal form
| PNF

\knowledge{notion}
| quantifier@fo
| quantifiers@fo
| quantifier
| quantifiers

\knowledge{notion}
| atom@fo
| atoms@fo
| atom
| atoms

\knowledge{notion}
| extended atom@fo
| extended atoms@fo

\knowledge{notion}
| constant
| constants

\knowledge{notion}
 | quantifier-width

\knowledge{notion}
| quantifier rank
\knowledge{notion}
| Tarski's calculus of relations
| CoR

\knowledge{notion}
| term@CoR
| terms@CoR

\knowledge{notion}
| atomic FO-formula@CoR
| atomic FO-formulas@CoR
| atomic@CoR

\knowledge{notion}
| semantics@CoR

\knowledge{notion}
| model-checking problem@CoR

\knowledge{notion}
| counting quantifiers
| counting quantifier
| Counting quantifiers

\knowledge{notion}
| type
| types

\knowledge{notion}
| realizability
| realizable

\knowledge{notion}
| box

\knowledge{notion}
| window

\knowledge{notion}
| semantics@BMLP

\knowledge{notion}
| term@BMLP
| terms@BMLP

\knowledge{notion}
| formula@BMLP
| formulas@BMLP

\knowledge{notion}
| atomic FO-formula@BMLP
| atomic FO-formulas@BMLP

\knowledge{notion}
| Gaifman graph

\knowledge{notion}
| guarded

\knowledge{notion}
| model-checking problem@BMLP

\knowledge{notion}
| DAG-@BMLP

\knowledge{notion}
| generalized inner product

\knowledge{notion}
| block
| blocks

\knowledge{notion}
| projection
| projections
| Projection

\knowledge{notion}
| guarded-pairing
| Guarded-pairing

\begin{document}

\title[A Fine-Grained Dichotomy for Bounded-Variable Query Evaluation]{%
A Fine-Grained Dichotomy for Bounded-Variable Query Evaluation}
\subtitle{The Calculus of Relations, a Boolean Modal Logic, and One-Variable Counting Logic}

\author{Yoshiki Nakamura}
\affiliation{%
  \institution{Chiba University}
  \city{Chiba}
  \country{Japan}}
\email{nakamura.yoshiki.ny@gmail.com}

\author{Yuya Uezato}
\affiliation{%
  \institution{CyberAgent, Inc.}
  \city{Tokyo}
  \country{Japan}}
\affiliation{%
  \institution{National Institute of Informatics}
  \city{Tokyo}
  \country{Japan}}
\email{uezato_yuya@cyberagent.co.jp}
\email{uezato@nii.ac.jp}

\begin{abstract}
We study the fine-grained complexity of evaluating Boolean bounded-variable first-order queries over sparse relational structures.
For every fixed $k \ge 2$, every relational signature, and every fragment between
$k$-variable primitive positive ($\mathrm{PP}^{k}$) and first-order ($\mathrm{FO}^{k}$) logic, we prove, assuming the Sparse MAX-$3$-SAT hypothesis, a dichotomy theorem for evaluation in $O(m^{k-\varepsilon})$ time,
where $m$ is the number of tuples in the input structure.
The only tractable cases fall into three families:
(1) three-variable fragments,
(2) two-variable fragments, and
(3) fragments over unary signatures.
On the hard side, for every fixed $k \ge 4$ and every $\varepsilon > 0$,
there is a fixed sentence $\varphi_\varepsilon$ in $\mathrm{PP}^{k}$
over a single binary relation symbol,
depending on $\varepsilon$ but not on the input structure,
whose evaluation cannot be performed in $O(m^{k-\varepsilon})$ time.
For the tractable cases, we show that the evaluation problem can be solved in $2^{O(|\varphi|)} \cdot m^{k-\varepsilon}$ time for some $\varepsilon > 0$.
Moreover, every tractable fragment reduces, with a $2^{O(|\varphi|)}$ blowup in formula size under DAG representations, to one of the following query languages:
(1) Tarski's calculus of relations,
(2) a new Boolean modal logic for sparse model checking, and
(3) one-variable counting logic.
 \end{abstract}


\maketitle

\nointro{\poly}
\nointro{constant}

\section{Introduction}\label{sec: introduction}

%
%
%
%
%
%
%
%
%
%
%

%
%
%

%
%
%
%

%
%
%
%
%
%
%

%
%

%

Evaluating "first-order" ($\FO$) queries
is a fundamental problem in database theory and logic.
While unrestricted $\FO$ "model checking" is "PSPACE"-hard~\cite[Theorem 4.3]{stockmeyerWordProblemsRequiring1973}, \cite{vardiComplexityRelationalQuery1982}
and remains "NP"-hard even when restricted to ""conjunctive queries""~\cite{chandraOptimalImplementationConjunctive1977},
it is solvable in polynomial time when the number of "variables" is bounded:
$k$-"variable" $\FO$ ($\kFO{k}$) queries $\fml$ can be evaluated in $\asympO_{k}(\len{\fml} \cdot n^k)$ time%
\footnote{We may write $\asympO_{x_1, x_2, \ldots, x_p}(-)$ instead of $\asympO(-)$ to explicitly denote that parameters $x_1, x_2, \ldots, x_p$ are fixed.}
on a ""database"" with $n$ elements,
by the naive dynamic programming algorithm~\cite{vardiComplexityBoundedvariableQueries1995}
that materializes a Boolean table of size $n^k$ for each "subformula".
A natural question is whether this $n^k$ dependence can be improved.

In general, the $n^k$ dependence cannot be improved even for the $k$-"variable" \reintro*\kl{primitive positive} fragment ($\kPP{k}$).
\reintro*\kPP{k} queries \cite{bulatovClassifyingComplexityConstraints2005,berkholzCompilingExistentialPositive2019} are $\kFO{k}$ queries built using only existential quantifiers ($\exists$) and conjunctions ($\land$).\footnote{%
By "conjunctive queries" we mean "primitive positive" queries in "prenex normal form".
}
Here, nested quantifiers are allowed.
For example, over a directed edge "relation" $E$, the $\kPP{3}$-"formula"
\[
\exists z (E(x,z) \land
\exists x (E(z,x) \land
\exists z (E(x,z) \land
\exists x (E(z,x) \land E(x, y)))))\]
asks whether there is a length-$5$ walk from $x$ to $y$.
Yet even in "prenex normal form" ("PNF"), for instance,
if the maximum arity of relations, denoted by $r$, is fixed at some $k > r \ge 3$,
the following query, called the $(k,r)$-""hyperclique query"",
\[%
{\textstyle \exists x_1 \cdots \exists x_k.~  \bigwedge_{1 \leq i_1 < \cdots < i_r \leq k} R(x_{i_1}, \dots, x_{i_r})}\]
cannot be evaluated in $\asympO(n^{k - \eps})$ time for any $\eps > 0$,
unless the $(k, r)$-\HyperCliqueHypothesis{}~\cite{abboudMoreConsequencesFalsifying2018,lincolnTightHardnessShortest2018}\ifarxiv{} (see \Cref{section: theorem: dichotomy dense})\fi{} fails.
Moreover, even if the maximum arity of relations is $2$, for every $k \ge 4$,
there is already a fixed non-"PNF" $\kPP{k}$-"sentence" over a single binary "relation" and many unary "relations"
which cannot be evaluated in $\asympO(n^{k-\eps})$ time for any $\eps > 0$, assuming the $(qk, q)$-\HyperCliqueHypothesis{} for some $3 \le q < k$, by \cite[Theorem 1.1]{bringmannCurrentAlgorithmsDetecting2021} (see also \Cref{section: bounded treewidth}).

The sparse setting, however, is less well understood:
several fragments of $\kFO{k}$ admit algorithms running in sub-$m^k$ time, where $m$ is the number of tuples in the input "structure".
For example, "conjunctive queries" with $k$ "variables", "ie", $\kPP{k}$ queries in "PNF", can be evaluated by worst-case optimal join algorithms \cite{ngoWorstcaseOptimalJoin2018} within the AGM bound \cite{atseriasSizeBoundsQuery2013}, which for many queries is below the naive $m^k$ bound.
Also, for $k \ge 3$, every $\kFO{k}$ query in "PNF" can be evaluated in $m^{k-1} / 2^{\asympTheta(\sqrt{\log m})}$ time \cite[Theorem 1.4]{gaoCompletenessFirstorderProperties2018}.
Nevertheless, no classification was known for full $\kFO{k}$, or even for $\kPP{k}$, without syntactic restrictions such as "PNF".

\paragraph{\textup{\textbf{Main Question}}}
\emph{Which fragments of $\kFO{k}$ can be evaluated in $f(\len{\fml}) \cdot m^{k-\eps}$ time for some $\eps > 0$ and some computable function $f$?}

In particular, to the best of our knowledge, it was not known even whether full $\kFO{k}$ queries can be evaluated in $f(\len{\fml}) \cdot m^{k-\eps}$ time for some $\eps > 0$ and some computable function $f$.

\subsection{Our contribution}
For each $k \ge 2$ and each "relational signature" $\sig$,
we provide a complete answer for every fragment between
$k$-"variable" "primitive positive" queries ($\kPP{k}\sigres{\sig}$) and "first-order" queries ($\kFO{k}\sigres{\sig}$) over $\sig$,
under standard fine-grained complexity assumptions:

\begin{theorem*}[Main Dichotomy; see \Cref{theorem: dichotomy sparse}]
Fix $k \ge 2$, a "relational signature" $\sig$, and a class $\fmlclass$ such that $\kPP{k}\sigres{\sig} \subseteq \fmlclass \subseteq \kFO{k}\sigres{\sig}$.
Then the following dichotomy holds.
\begin{itemize}
\item If $k \leq 3$, or if $\sig$ is unary, then $\fmlclass$ can be evaluated in $2^{\asympO(\len{\fml})} \cdot m^{k-\eps}$ time for some $\eps > 0$.
\item Otherwise, assuming the \SparseMAXxSATHypothesis{3}, for every $\eps>0$, there is a fixed "sentence" $\fml_{\eps} \in \kPP{k}\sigres{\sig}~(\subseteq \fmlclass)$
whose evaluation cannot be performed in $\asympO(m^{k-\eps})$ time.
\end{itemize}
\end{theorem*}
Using the same reductions,
we also obtain a dense analogue of the dichotomy under certain assumptions (\Cref{sec: dichotomy dense});
\ifthenelse{\boolean{arxiv}}{%
see \Cref{section: theorem: dichotomy dense}}{%
see the extended version~\cite{nakamuraUezatoDichotomyBoundedVariable2026arxiv}} for details.

Interestingly, every tractable bounded-variable fragment in our dichotomies
collapses to one of the three query languages that admit efficient evaluation algorithms,
namely $\CoR$, $\BMLP$, or $\kC{1}$:
\begin{itemize}
    \item $\CoR$: the calculus of relations \cite{tarskiCalculusRelations1941}, which is in $\asympO(\len{\fml} \cdot n^{\matrixexponent})$ time and captures the $k=3$ case,
    \item $\BMLP$: a new Boolean modal logic inspired by the modal logic $\MLLSW$ of Lutz, Sattler, and Wolter \cite{lutzModalLogicTwoVariable2001}, which is in $\asympO(\len{\fml} \cdot m)$ time and captures the $k=2$ case, and
    \item $\kC{1}$: one-variable counting logic, which is in $\asympO(\len{\fml} \cdot n)$ time and captures the unary case.
\end{itemize}
Here, $\intro*\matrixexponent < \intro*\matrixexponentupperbound$ \cite{almanMoreAsymmetryYields2025} is the exponent of matrix multiplication.
Our dichotomy applies to arbitrary $\kFO{k}$-"formulas":
the input "sentence" need not be in "PNF", nor come from any fixed "query" pattern.

The main technical contributions are as follows:
\begin{itemize}
\item
We give a sparsity-preserving reduction for hardness of nested bounded-"variable" queries.
In the dense setting, Bringmann and Slusallek \cite{bringmannCurrentAlgorithmsDetecting2021} already prove hardness with one binary "relation" and enough unary "relations" by reducing from \HyperCliqueProblem{},
but their reduction does not by itself yield the sparse $m^k$ barrier; see \Cref{remark: extending to dense hardness}.
We instead reduce from \SparseMAXxSATProblem{(k-1)} (whose hardness follows from the \SparseMAXxSATHypothesis{3}) via \OVHyperCliqueProblem{}.
Moreover, we encode unary and binary "relations" into a single binary "relation" (\Cref{lemma: reduction monadic and binary predicates to single binary predicate}),
showing that one binary "relation" already suffices for the hard side (\Cref{theorem:sparse-k4,theorem:sparse-k3}).

\item
For the sparse $k = 3$ case, even if the maximum "arity" is unbounded, we can still obtain a $2^{\asympO(\len{\fml})} \cdot m^{\matrixexponent}$-time algorithm (which is tight under the \SparseMAXxSATHypothesis{2}; \Cref{theorem:sparse-k3}).
To this end, we introduce a \emph{"guarded-pairing"} translation (\Cref{lemma: guarded-pairing}), which reduces ternary "relations" to binary "relations" without increasing the number of "variables"
while preserving the sparse input size.
While a similar reduction exists for the "PNF" case by fixing "variables" to "constants" \cite[Lemma 7.1]{gaoCompletenessFirstorderProperties2018}, it does not extend to nested quantifiers.

\item
We design $\BMLP$, a new Boolean modal logic, inspired by \cite{lutzModalLogicTwoVariable2001}, that serves as a normal form for sparse $\kFO{2}$ model checking.
Unlike $\MLLSW$ \cite{lutzModalLogicTwoVariable2001}, $\BMLP$ has no term complement, so it admits a straightforward bottom-up $\asympO(\len{\fml} \cdot m)$-time evaluation algorithm.
A linear-time algorithm was previously known only for $\kFO{2}$ in "PNF" \cite[Lemma 9.1]{gaoCompletenessFirstorderProperties2018};
it is not bottom-up and does not extend to nested quantifiers.
\end{itemize}

\begin{remark}
To obtain algorithmic bounds, $\CoR$, $\BMLP$, and $\kC{1}$ are used as normal forms:
$\CoR$ for $\kFO{3}$, $\BMLP$ for $\kFO{2}$, and $\kC{1}$ for unary $\kFO{k}$.
Once normalized, evaluation is simple and has only linear dependence on the translated "formula" size.
This separation also simplifies the correctness proofs.
For example, a direct approach within $\kFO{2}$ requires careful normalization of sparse/co-sparse cases ("cf" \cite[Lemma 9.1]{gaoCompletenessFirstorderProperties2018}),
whereas $\BMLP$ makes the sparse/co-sparse distinction syntactically explicit.%
\end{remark}

\subsection{Related work}\label{sec:related work}

\subsubsection{$\kPP{k}$ query evaluation and the "structure-homomorphism problem" for bounded "treewidth"}\label{section: bounded treewidth}
The ""structure-homomorphism problem"", equivalently the ""constraint satisfaction problem"" (\reintro*\kl{CSP}) \cite{federComputationalStructureMonotone1998,kolaitisConjunctiveQueryContainmentConstraint2000}, asks whether, given two "structures" $\struc[1]$ and $\struc[2]$ over $\sig$,
there is a "homomorphism" from $\struc[1]$ to $\struc[2]$.
Evaluating a fixed $\kPP{k}$-"sentence" can be viewed as
the "structure-homomorphism problem" from a fixed source "structure" of ""treewidth"" at most $k-1$.
This follows from the standard correspondence between $\kPP{k}$-"sentences" and "structures" of "treewidth" at most $k-1$ \cite[Remark~5.8]{kolaitisConjunctiveQueryContainmentConstraint2000}, \cite[Lemma 1]{dalmauConstraintSatisfactionBounded2002}.

The $\asympO(n^{k - \eps})$-hardness in the dense setting has already been shown even for the ""colored subgraph isomorphism problem"" \cite[Theorem 1.1]{bringmannCurrentAlgorithmsDetecting2021} (more precisely, \cite[Lemma 21 and Theorem 22 (Part 1)]{bringmannCurrentAlgorithmsDetecting2021arXiv}), a special case of the "structure-homomorphism problem" with a single binary "relation" and as many unary "relations" as there are "vertices" in the pattern "graph":
\phantomintro{graph homomorphism}Given a pattern "graph" $\graph[1]$, a "graph" $\graph[2]$, and a coloring function $c \colon V(\graph[2]) \to V(\graph[1])$, the task is to decide whether there is an injective%
\footnote{Observe that any "graph homomorphism" that preserves the coloring is injective, so we can remove the injectivity here.} "graph homomorphism" $f \colon V(\graph[1]) \to V(\graph[2])$ that preserves the coloring: $c(f(v)) = v$ for all $v \in V(\graph[1])$.
Using the correspondence with $\kPP{k}$-"sentences", this can be restated as follows ("cf" \Cref{corollary: dense hardness}):
\begin{proposition}[Cor.\@ of {\cite[Lemma 21 and Theorem 22 (Part 1)]{bringmannCurrentAlgorithmsDetecting2021arXiv}}]\label{proposition: dense hardness}
Fix $k \ge 4$.
Assuming the $(qk, q)$-\HyperCliqueHypothesis{} for some $3 \le q < k$,
there is a fixed $\kPP{k}$-"sentence" $\fml$ with sufficiently many unary "relations" and a single binary "relation" such that for any $\eps > 0$, the "model-checking problem" cannot be solved in $\asympO(n^{k - \eps})$ time.
\end{proposition}
\begin{remark}\label{remark: extending to dense hardness}
The reduction of \cite[Lemma 21]{bringmannCurrentAlgorithmsDetecting2021arXiv}
does not seem to yield a tight lower bound in the sparse setting, "cf" \Cref{theorem:sparse-k4,theorem:sparse-k3}.
Indeed, the "structure" produced by that reduction has $\asympTheta(n^q)$ "elements" and, in the worst case, $\asympTheta(n^{2q - 1})$ "edges".
Equivalently, if \(N=\asympTheta(n^q)\) is the number of "elements", then the number of "edges" can be \(\asympTheta(N^{2-1/q})\), not \(\asympO(N)\).
Thus, this dense lower bound does not yield the sparse $m^{k}$ barrier for bounded "treewidth".
Our reduction preserves sparsity.
\end{remark}

\begin{remark}\label{remark: Gao and Impagliazzo}
\Cref{proposition: dense hardness} highlights a fundamental obstacle to extending
the known algorithms for $\PNFkFO{k}$ to full $\kFO{k}$.
Williams' algorithm in \cite[Corollary 1.3]{williamsFasterDecisionFirstorder2014} exploits the "PNF" structure,
in particular the ability to fix outer "variables" and reduce the remaining formula
to a lower-variable instance.
The lower bound of \Cref{proposition: dense hardness} shows that this strategy cannot yield an
$\asympO(n^{k-3+\matrixexponent})$-time algorithm for full $\kFO{k}$ when
$k \ge 4$, under the $(qk, q)$-\HyperCliqueHypothesis{} for some $3 \le q < k$.
This also explains why the analogous bound claimed for full $\kFO{k}$ in
\cite[Theorem~7]{gaoFineGrainedComplexityStrengthenings2019}, \cite[Theorem~14]{gaoFineGrainedComplexityProblems2019} would require
additional ideas beyond the "PNF"-based approach.
\end{remark}
We slightly refine \Cref{proposition: dense hardness} by removing the need for unary "relations" (\Cref{corollary: dense hardness}), using our encoding into a single binary "relation" (\Cref{lemma: reduction monadic and binary predicates to single binary predicate}).
As a corollary, for every $k \ge 4$, under the same hypothesis, there is a directed uncolored "graph" $\graph[1]$ of "treewidth" at most $k-1$ such that
the fixed-pattern ""graph homomorphism problem"" $\intro*\HOM(\graph[1], -)$,
the problem of deciding, given a directed "graph" $\graph[2]$, whether there is a "homomorphism" from $\graph[1]$ to $\graph[2]$,
cannot be solved in $\asympO(n^{k - \eps})$ time for any $\eps > 0$,
which may be of independent interest as a purely graph-theoretic problem; see \Cref{corollary: dense hardness graph}.

\subsubsection{\texorpdfstring{$\FO$}{FO} model checking on restricted "structure" classes}
While we consider the "model-checking problem" for $\kFO{k}$ on general "structures",
the "model-checking problem" for $\FO$ can be solved more efficiently on restricted classes of "structures", "eg",
in $\asympO_{\fml}(n)$ time on "words" (via finite automata \cite{buchiWeakSecondorderArithmetic1960,elgotDecisionProblemsFinite1961}),
in $\asympO_{\fml,\eps}(n^{1+\eps})$ time for any $\eps > 0$ on ""nowhere dense"" "graphs" \cite{groheDecidingFirstOrderProperties2017}, and
in $\asympO_{\fml}(n)$ time on bounded ""twin-width"" "graphs" (given with a $d$-sequence) \cite{bonnetTwinwidthTractableFO2021}.
These results require a non-elementary dependence on $\len{\fml}$.
In particular, assuming "FPT" $\neq$ "AW[*]",
the $\FO$ "model-checking problem" on "words"
cannot be solved in $f(\len{\fml}) \cdot n^{\asympO(1)}$ time for any elementary function $f$
\cite[Theorem 1 (2)]{frickComplexityFirstorderMonadic2004}.
Even for bounded-"variable" fragments,
assuming the $\ell$-\intro*\CliqueHypothesis{} for every $\ell \ge 3$,
the $\kFO{3}$ "model-checking problem" on "words"
cannot be solved in $f(\len{\fml}) \cdot n^{\matrixexponent - \eps}$ time for any elementary function $f$ and any $\eps > 0$ \cite[Theorem 3 (b)]{nogamiHardnessRegularExpression2026},
via the translation of star-free extended regular expressions into $\kFO{3}$ \cite[Theorem 5.2 (1)]{stockmeyerComplexityDecisionProblems1974}.
\subsection{Organization}
In \Cref{sec: preliminaries},
we recall basic definitions.
In \Cref{section: dichotomy sparse},
we state our dichotomy theorem and give a technical overview.
In \Cref{section: reduction from OV Hyperclique to MCFO},
we prove the hardness results (\Cref{theorem:sparse-k4,theorem:sparse-k3}).
In \Cref{section: fo3,section: fo2},
we prove the upper-bound results (\Cref{theorem: fo3,theorem: fo2});
\Cref{theorem: monadic} is proved in \ifthenelse{\boolean{arxiv}}{\Cref{section: monadic}}{the extended version~\cite{nakamuraUezatoDichotomyBoundedVariable2026arxiv}}.
In \Cref{sec: conclusion}, we conclude this paper with open problems.

\section{Preliminaries}\label{sec: preliminaries}

\AP
We write $\intro*\nat = \set{0, 1, 2, \ldots}$ for the set of non-negative integers.
For $l, r \in \nat$,
let $\intro*\range{l\,..\,r} \defeq \set{i \in \nat \mid l \le i \le r}$.
In particular, $\range{r} \defeq \range{1\,..\,r}$.
Let $X$ be a %
set.
We write $\intro*\card{X}$ for its cardinality.
For $k \in \nat$, we write $\binom{X}{k}$ for the set of $k$-element subsets of $X$:
"ie", $\binom{X}{k} \defeq \set{ A \subseteq X ~\mid~ \card{A} = k }$.
For a ""vector"" ("resp", an ordered set) $v$, we write $v\intro*\ith{i}$ for the $i$-th component ("resp", $i$-th smallest element) of $v$.
\phantomintro{equivalence relation}\phantomintro{equivalence class}\phantomintro{word}%

\AP
A ""relational signature"" (henceforth, \reintro*\kl{signature}) $\intro*\sig$ is a set of ""relation symbols"".
Each $R \in \sig$ has a positive integer ""arity"", denoted by $\intro*\arity{R}$.
We write $\intro*\sigarity{\sig}$ for the supremum of the "arities" in $\sig$,
with the convention $\sigarity{\sig}=0$ when $\sig$ has no "relation symbols".
When the names of the "relation symbols" are irrelevant,
we use a tuple $(C_1, \dots, C_r)$ in $(\nat \cup \set{\ordinalomega})^r$
to denote a "signature" such that,
for every $q \in \range{r}$, $C_q$ is the allowed number of $q$-ary "relation symbols"
and no "relation symbols" of "arity" greater than $r$ are allowed.
$C_q = \intro*\ordinalomega$ means no bound.
This notation follows~\cite[Definition 1.3.1]{borgerClassicalDecisionProblem1997}.

\AP
A ""structure"" $\struc$ over $\sig$ is a tuple $\tuple{\univ{\struc}, \set{R^{\struc}}_{R \in \sig}}$,
where the ""universe"" $\intro*\univ{\struc}$ is a nonempty finite set of ""elements""
and each $R^{\struc} \subseteq \univ{\struc}^{\arity{R}}$ is an $\arity{R}$-ary ""relation"".
We use two size measures for $\struc$,
$\intro*\strucn{\struc}$ in the dense setting and $\intro*\strucm{\struc}$ in the sparse setting.
\begin{align*}
\strucn{\struc} &\defeq \card{\univ{\struc}} \; \text{(vertex-count, dense setting)}&
\strucm{\struc} &\defeq \max(1, \sum_{R\,\in\,\sig} \card{R^{\struc}}) \; \text{(tuple-count, sparse setting)}
\end{align*}
When there is no ambiguity, we simply write $n$ and $m$ for $\strucn{\struc}$ and $\strucm{\struc}$, respectively.
In the sparse setting, the input consists of $\strucn{\struc}$ together with explicit lists of relation tuples.
In the dense setting, each "relation" $R^{\struc}$ is given by its adjacency tensor, "ie", a Boolean array with $\strucn{\struc}^{\arity{R}}$ entries.

Given "structures" $\struc$ and $\struc[2]$ over the same "signature" $\sig$,
a ""structure-homomorphism"" from $\struc$ to $\struc[2]$ is a map $h : \univ{\struc} \to \univ{\struc[2]}$ such that, for every $R \in \sig$ and every tuple $(a_1, \dots, a_{\arity{R}}) \in R^{\struc}$, we have $(h(a_1), \dots, h(a_{\arity{R}})) \in R^{\struc[2]}$.

\AP
We use $x, y, z, \dots$ to denote ""variables"".
We write $\intro*\FO\intro*\sigres{\sig}$ (or simply $\FO$) for the set of ""first-order formulas"" with ""equality"" (henceforth, \reintro*\kl{formulas}) over $\sig$,
given by the following grammar:
\[
\fml[1], \fml[2] ~\Coloneqq~ R(x_1, \dots, x_{\arity{R}}) \mid x = y \mid \fml[1] \land \fml[2] \mid \fml[1] \lor \fml[2] \mid \exists x. \fml \mid \forall x. \fml \mid \lnot \fml[1]
\quad \text{($R \in \sig$)}.
\]
\AP
We write $\intro*\fmltrue$ for the "formula" $\exists x. x = x$,
write $\intro*\fmlfalse$ for $\lnot \fmltrue$, and
write $x \intro*\fmlneq y$ for $\lnot x = y$.
\phantomintro{inequality}%
\AP
\phantomintro{graph}\phantomintro{edge}%
We represent "formulas" as \emph{"directed acyclic graphs"} (""DAGs"") to allow sharing of common ""subformulas"".
The semantics are defined on the unfolded tree of the DAG, "ie", each bound "variable" is captured by a single ""quantifier@@fo"", as usual.
The "vertices" consist of "relation symbols", "variables", and logical connectives.
The ""size"" $\intro*\len{\fml}$ of a "formula" $\fml$ is defined as the number of "edges" in the DAG representation of $\fml$.
We say that $\fml$ is ""DAG-represented"" in general, and ""tree-represented"" if its DAG is a tree.
\begin{remark}\label{remark: DAG and tree representations}
Our complexity results for $\FO$ "model checking"
hold regardless of whether the input "sentences" are "DAG-represented" or "tree-represented"%
\ifarxiv, except for the low-dimensional bound in \Cref{section: theorem: FO model checking hardness one binary}.%
\else.\fi{}
This distinction affects only the "formula" translations used as intermediate steps in our algorithmic results (see also \Cref{remark: DAG representation in algorithmic bounds}).
\end{remark}
\AP
We write $\intro*\FV(\fml)$ for the set of free "variables" in $\fml$.
A ""sentence"" is a "formula" without free "variables".

\AP
For a "formula" $\fml \in \FO\sigres{\sig}$,
a "structure" $\struc$ over $\sig$, and a ""valuation"" $\val$ (a map from "variables" to $\univ{\struc}$), we write 
$\struc, \val \models \fml$ to mean that $\fml$ is true in $\struc$ under $\val$\phantomintro{truth value}\phantomintro{expressive power}.
We simply write $\struc \models \fml$ when $\fml$ is a "sentence".
For a "formula" $\fml(x_1, \dots, x_\ell)$ with pairwise distinct $x_1, \dots, x_\ell$ such that $\FV(\fml) \subseteq \set{x_1, \dots, x_\ell}$ and "elements" $a_1, \dots, a_\ell \in \univ{\struc}$,
we write $\struc \models \fml(a_1, \dots, a_\ell)$ for $\struc, \set{x_i \mapsto a_i \mid i \in \range{\ell}} \models \fml$.

\AP
We write $\intro*\vars(\fml)$ for the set of "variables" (free or bound) occurring in $\fml$.
We denote by $\intro*\kFO{k}$
the set of $\FO$ %
"formulas" using at most $k$ distinct "variables": "ie", $\card{\vars(\fml)} \leq k$.
For example, consider the following two logically equivalent "formulas":
\begin{flalign*}
\fml[1] &= \exists y_1 y_2 y_3. \bigl(E(x, y_1) \land E(y_1, y_2) \land E(y_2, y_3) \bigr), &
\fml[2] &= \exists y. \bigl(E(x, y) \land \exists x. (E(y, x) \land \exists y. E(x, y)) \bigr).
\end{flalign*}
While $\fml[1]$ is in $\kFO{4}$ but not in $\kFO{3}$, $\fml[2]$ is in $\kFO{2}$ by reusing "variables",
though it is not in "PNF".

\AP
An ""atom"" is either $R(x_1, \dots, x_{\arity{R}})$ or $x = y$.
We consider the following syntactic fragments of $\kFO{k}$.
""Primitive positive"" "queries" ($\intro*\kPP{k}$) consist of formulas built from "atoms" (including "equalities") using $\land$ and $\exists$.
""Existential positive"" "queries" ($\intro*\kEP{k}$) additionally allow $\lor$.
\phantomintro*\PNFkFO{}%
\phantomintro*\PNFkPP{}%
\phantomintro*\PNFkEP{}%
$\mathrm{PNF}$-$\fmlclass$ denotes $\fmlclass$-"formulas" in ""prenex normal form"" (\reintro*\kl{PNF}), for each $\fmlclass$.

\AP
For a class $\fmlclass \subseteq \FO\sigres{\sig}$,
the \intro*\kl{model-checking problem} $\intro*\MC{\fmlclass}$ is the following decision problem:
\begin{center}
Given a "sentence" $\fml$ in $\fmlclass$ and a "structure" $\struc$ over $\sig$, decide whether $\struc \models \fml$.
\end{center}
We use Boolean query evaluation and "sentence" "model checking" interchangeably.

We write $\intro*\fmlsig{\fml}$ for the "relation symbols" occurring in $\fml$
and write $\intro*\fmlarity{\fml}$ for its maximum "arity".
For $\MC{\kFO{k}}$,
without loss of generality, we can ensure $n \le \asympO_k(\max(\fmlarity{\fml}, 1) \, m)$.
This follows from the following proposition, which removes \emph{isolated} "elements", "ie", "elements" occurring in no tuple of any "relation";
note that $\strucm{\struc}$ counts tuples rather than isolated "elements" ("cf" \cite[p.~8]{bringmannFineGrainedAnalogueSchaefers2019}).

\ifthenelse{\boolean{arxiv}}{%
\begin{proposition}[{\Cref{section: proposition: isolated}}]%
}{%
\begin{proposition}%
}\label{proposition: isolated}
\gdef\propositionisolated{%
Given an $\kFO{k}$-"sentence" $\fml$ and a "structure" $\struc$,
we can construct a "structure" $\struc[2]$ in $\asympO_{k}(\max(\fmlarity{\fml}, 1) \, \strucm{\struc})$ time such that $\struc \models \fml$ "iff" $\struc[2] \models \fml$.
Moreover, $\strucn{\struc[2]} \le \fmlarity{\fml} \, \strucm{\struc} + k$.
}%
\propositionisolated
\end{proposition}
Each input "formula" $\fml$ is understood over its finite active "signature" $\fmlsig{\fml}$,
consisting of the "relation symbols" that occur in $\fml$.
Accordingly, we restrict the input "structure" to the same "signature" $\fmlsig{\fml}$.
This convention also applies to notation such as $\sigres{(\ordinalomega,\ordinalomega)}$:
although the ambient class allows infinitely many "relation symbols",
each input instance is over the finite "signature" $\fmlsig{\fml}$.%

\AP
Throughout the paper, we work in the ""word-RAM model"" with word size $\asympTheta(\log n)$, following the standard setting in fine-grained complexity.

\section{Dichotomy Theorem: Precise Statement and Technical Overview}\label{section: dichotomy sparse}
\AP
We establish the following dichotomy theorem for $\kFO{k}$ queries with a "signature" $\sig$ (denoted by $\kFO{k}\sigres{\sig}$) in the sparse setting.
The hardness part of the dichotomy relies on the \intro*\SparseMAXxSATHypothesis{q} for $q \ge 3$ \cite[Hypothesis 3 for $q = 3$]{almanOVGraphsAre2020}.
Here \intro*\MAXxSATProblem{q} is the problem of computing the maximum number of "clauses" satisfiable by a single "valuation" in a given $q$-"CNF" formula; we use the convention that each such clause has at most $q$ literals.
The \SparseMAXxSATHypothesis{q} states that, for any $\eps > 0$,
there is some $c > 0$ such that \MAXxSATProblem{q} with $n$ "variables@@propositional" and $c n$ "clauses" is not solvable in $\asympO_{c}(2^{(1 - \eps) n})$ time.
With this convention, the \SparseMAXxSATHypothesis{3} implies the \SparseMAXxSATHypothesis{q} for every $q \geq 4$.
\begin{theorem}[Dichotomy in the sparse setting; see also \Cref{table: sparse dichotomy}]\label{theorem: dichotomy sparse}
Fix%
\footnote{\label{footnote:k=1}The case $k=1$ may be degenerate and is handled separately;
\ifthenelse{\boolean{arxiv}}{see \Cref{section: footnote:k=1}}{see also the extended version~\cite{nakamuraUezatoDichotomyBoundedVariable2026arxiv} for more details}.}
$k \ge 2$,
a "relational signature" $\sig$, and
a class $\fmlclass$ such that $\kPP{k}\sigres{\sig} \subseteq \fmlclass \subseteq \kFO{k}\sigres{\sig}$.
Then the following dichotomy holds:

\noindent\textbf{Tractable cases.}
$\MC{\fmlclass}$, where the input $\fml$ may be any $\fmlclass$-"sentence", is solvable in $2^{\asympO(\len{\fml})} \cdot m^{k - \eps}$ time for some $\eps > 0$ if one of the following holds:
\begin{itemize}
    \item $\sigarity{\sig} \le 1$, in which case the problem reduces to one-variable counting logic $\kC{1}$
    \ifcameraready\AP$\phantomintro*\kC{}$\phantomintro{counting quantifiers}\fi
    (this case is solvable in $\asympO(\len{\fml})$ time if $\sigarity{\sig}=0$ and in $2^{\asympO(\len{\fml})}\cdot m$ time if $\sigarity{\sig}=1$);
    \item $k = 2$ and $\sigarity{\sig} \geq 2$, in which case the problem reduces to $\BMLP$
    (this case is solvable in $2^{\asympO(\len{\fml})} \cdot m$ time); or
    \item $k = 3$ and $\sigarity{\sig} \geq 2$, in which case the problem reduces to $\CoR$
    (this case is solvable in $2^{\asympO(\len{\fml})} \cdot m^{\matrixexponent}$ time).
\end{itemize}

\noindent\textbf{Hard cases.}
In all remaining cases, equivalently when $k \geq 4$ and $\sigarity{\sig} \geq 2$,
for every $\eps>0$,
there exists %
a fixed "sentence" $\fml_{\eps}$ in $\kPP{k}\sigres{\sig}$,
depending on $\eps$ but independent of the input "structure",
whose "model-checking problem" cannot be solved in $\asympO(m^{k-\eps})$ time,
unless the \SparseMAXxSATHypothesis{3} fails.
Consequently, under this hypothesis, $\MC{\fmlclass}$ cannot be solved in $f(\len{\fml}) \cdot m^{k-\eps}$ time for any function $f$ and any $\eps > 0$.
\end{theorem}
\begin{table}[h]
\centering
\caption{Dichotomy for $\kPP{k}\sigres{\sig} \subseteq \fmlclass \subseteq \kFO{k}\sigres{\sig}$, where $k \ge 2$ and $\sig$ is a "relational signature".
The cyan entries correspond to tractable cases
and the orange entries correspond to hard cases in \Cref{theorem: dichotomy sparse}.
In the $k \ge 4$ cases, the \SparseMAXxSATHypothesis{(k-1)} is implied by the \SparseMAXxSATHypothesis{3}.
}
\label{table: sparse dichotomy}
\definecolor{myPink}{HTML}{FFEEDD}
\definecolor{myCyan}{HTML}{DDEEFF}

\resizebox{\textwidth}{!}{
\begin{talltblr}[
  label = none, entry = none]{
  colspec = {Q[c,1.2cm] | Q[c, 1.2cm] | Q[c, 1.4cm] | Q[c,14.cm]  },
  row{1} = {},
  column{1} = {},
  hlines, 
  vlines,
  hline{2} = {1.5pt},
  vline{2} = {1.5pt},
  vline{3} = {2-4}{0pt},
  cell{2}{2} = {r=3}{bg=myCyan},       %
  cell{2}{3} = {r=3}{bg=myCyan},       %
  cell{2}{4} = {c=1}{bg=myCyan},      %
  cell{3}{4} = {c=1}{bg=myCyan},       %
  cell{4}{4} = {c=1}{bg=myPink}, %
}
  \diagbox[width=1.6cm, height = 1.\line]{\raisebox{-.4ex}{$k$}}{\raisebox{.4ex}{$\sigarity{\sig}$}} & 0 & 1 & $\sigarity{\sig} \ge 2$ \\

  2 & {$\asympO(\len{\fml})$} \par via $\kC{1}$ \par (Thm.~\ref{theorem: monadic}) & {$2^{\asympO(\len{\fml})} \cdot m$} \par via $\kC{1}$  \par (Thm.~\ref{theorem: monadic})
    & {$2^{\asympO(\len{\fml})} \cdot m$} via $\BMLP$ (Thm.~\ref{theorem: fo2}) \\

  3 & & & {$2^{\asympO(\len{\fml})} \cdot m^{\matrixexponent}$} via CoR (Thm.~\ref{theorem: fo3}) \\
  $k \ge 4$ & & & Naive $\asympO(m^k)$ is optimal under %
  the \SparseMAXxSATHypothesis{(k-1)} (Thm.~\ref{theorem:sparse-k4}) 

\end{talltblr}
}
 \end{table}
The lower bound is non-uniform in \(\eps\); whether a single fixed $\kPP{k}$-"sentence" works for all \(\eps>0\) is left open (\Cref{sec: fixed sentence}).
Roughly speaking, \Cref{theorem: dichotomy sparse} shows that for sparse "model checking", the only exceptions to
the naive $m^k$ barrier are the unary case, the two-"variable" case, and the three-"variable" case;
from four "variables" on, a single binary "relation" already yields the full $m^k$ hardness.

\subsection{Technical overview of hardness results}\label{section: technical overview}
The hardness result of \Cref{theorem: dichotomy sparse} holds even when restricted to $\kPP{k}\sigres{(0, 1)}$, "ie", $\kPP{k}$-"sentences" with at most one binary "relation"; recall \Cref{sec: preliminaries} for the notation of $(C_1, C_2, \dots, C_r)$.
\begin{theorem}[\Cref{section: reduction from OV Hyperclique to MCFO}]
\label{theorem:sparse-k4}
\gdef\sparseFOhardnesstheorem{
Fix $k \ge 4$.
Assuming the \SparseMAXxSATHypothesis{(k-1)},
for any $\eps > 0$, there exists a $\kPP{k}\sigres{(0, 1)}$-"sentence" $\fml_{\eps}$ such that
\textup{$\MC{\set{\fml_{\eps}}}$} cannot be solved in $\asympO(m^{k - \eps})$ time.
}
\sparseFOhardnesstheorem
\end{theorem}
$\MC{\kPP{k}\sigres{(0, 1)}}$ is reduced to $\MC{\kPP{k}\sigres{\sig}}$ for $\sigarity{\sig} \ge 2$ by replacing each "atom" $E(x, y)$ with $R(x, y, \dots, y)$ (see \Cref{section: proposition: arity monotone}).
Hence, \Cref{theorem:sparse-k4} implies the hardness part of \Cref{theorem: dichotomy sparse}.

For the hardness result, we give a reduction from \SparseMAXxSATProblem{(k-1)} to $\MC{\kPP{k}\sigres{(0, 1)}}$, as in the following chain:
\begin{align*}
&\text{\SparseMAXxSATProblem{(k-1)}}
\leadsto_{\text{by the reduction of \cite[Thm.~23]{almanOVGraphsAre2020} ("cf" \Cref{proposition: alman VVW})}}
\text{$(k,k-1)$-\OVHyperCliqueProblem{}}\\
&\leadsto_{(1)~\Cref{section: to PPk}}~ \text{$\MC{\kPP{k}\sigres{(\ordinalomega{}, k-1)}}$}
~\leadsto_{(2)~\Cref{section: to PPk one binary relation}}~ \text{$\MC{\kPP{k}\sigres{(0, 1)}}$}.
\end{align*}
Here, \OVHyperCliqueProblem{} denotes the variant of \HyperCliqueProblem{} over "OV-hypergraphs" (see \Cref{section: to OVHyperclique} for the precise definition).
(1) is the more technical part:
we reduce $(k, k-1)$-\OVHyperCliqueProblem{} to $\MC{\kPP{k}\sigres{(\ordinalomega{}, k-1)}}$,
by splitting the "dimensions" of the input "vectors" into "blocks", in the spirit of the split-and-list technique \cite{williamsNewAlgorithmOptimal2005} (see \Cref{lemma: new new partitioned binary hyperclique hypothesis,lemma:sparse-k4}).
(2) is based on the idea of encoding unary and binary "relations" into a single binary "relation",
with a slight modification due to the restriction to $\kPP{k}$ (see \Cref{lemma: reduction monadic and binary predicates to single binary predicate}).

\AP
Beyond \Cref{theorem: dichotomy sparse}, the same approach shows that the $k=3$ upper bound is tight.
For this result we use the \SparseMAXxSATHypothesis{2} of Alman and Vassilevska Williams~\cite[Hypothesis~4]{almanOVGraphsAre2020}:
for any $\eps > 0$, there exists a constant $c > 0$ such that \MAXxSATProblem{2} with $n$ "variables@@propositional" and $c n$ "clauses" cannot be solved in $\asympO_{c}(2^{(\matrixexponent/3 - \eps) n})$ time.
The following corollary shows that the $\asympO_{\fml}(m^{\matrixexponent})$ upper bound for $\kFO{3}$ (see \Cref{theorem: fo3}) is tight:
\begin{corollary}[\Cref{section: reduction from OV Hyperclique to MCFO}]\label{theorem:sparse-k3}
\gdef\sparsehardnessbyovtriangle{
Assuming the \SparseMAXxSATHypothesis{2},
for any $\eps > 0$, there exists a $\kPP{3}\sigres{(0, 1)}$-"sentence" $\fml_{\eps}$ such that
\textup{$\MC{\set{\fml_{\eps}}}$} cannot be solved in $\asympO(m^{\matrixexponent - \eps})$ time.
}
\sparsehardnessbyovtriangle
\end{corollary}

\subsection{Technical overview of algorithmic bounds}
\AP
On the upper-bound side,
we establish the following results%
\footnote{\label{footnote: dense fo3} In the dense setting of \Cref{theorem: fo3}, the maximum "arity" must be at most $2$.
Otherwise, a fixed query requires $n^{3-o(1)}$ time by an input-reading lower bound\ifthenelse{\boolean{arxiv}}{\ (see also \Cref{section: trivial lower bound}).%
}{.}} using $\CoR$, $\BMLP$, and $\kC{1}$.
The proof of \Cref{theorem: monadic}
\ifthenelse{\boolean{arxiv}}{is given in \Cref{section: monadic}}{is omitted due to the page limitation (see the extended version~\cite{nakamuraUezatoDichotomyBoundedVariable2026arxiv})};
its baseline idea of replacing "equalities" with "counting quantifiers" is the same as \cite[Theorem 5]{mintsLogicalEquationsMonadic2009}, and we give a bottom-up algorithm using "types".\ifcameraready\AP\phantomintro{type}\fi

\begin{theorem}[\Cref{section: fo3}, via $\CoR$]
\label{theorem: fo3}
\gdef\theoremfothreemodelchecking{%
In the dense setting, $\MC{\kFO{3}\sigres{(\ordinalomega,\ordinalomega)}}$ is solvable in $2^{\asympO(\len{\fml})} \cdot n^{\matrixexponent}$ time.
}
\theoremfothreemodelchecking
Furthermore, in the sparse setting, even if the maximum "arity" is unbounded,
\gdef\theoremfothreemodelcheckingsparse{%
$\MC{\kFO{3}}$ is solvable in $2^{\asympO(\len{\fml})} \cdot m^{\matrixexponent}$ time.
}%
\theoremfothreemodelcheckingsparse
\end{theorem}

\begin{theorem}[{\Cref{section: fo2}}, via $\BMLP$]
\label{theorem: fo2}
\gdef\theoremFOTWOLIN{
$\MC{\kFO{2}}$ is solvable in $2^{\asympO(\len{\fml})} \cdot m$ time.
}
\theoremFOTWOLIN
\end{theorem}

\begin{theorem}[\ifthenelse{\boolean{arxiv}}{\Cref{section: monadic}}{see the extended version~\cite{nakamuraUezatoDichotomyBoundedVariable2026arxiv}}, via $\kC{1}$]\label{theorem: monadic}
Fix $k \ge 1$.
First,
\gdef\theoremmonadicFOEk{%
\textup{$\MC{\kFO{k}\sigres{(\ordinalomega)}}$} is solvable in $2^{\asympO(\len{\fml})} \cdot n$ time (hence, in $2^{\asympO(\len{\fml})} \cdot m$ time by \Cref{proposition: isolated}).
}%
\theoremmonadicFOEk
Second,
\gdef\theoremzeroadicFOEk{%
\textup{$\MC{\kFO{k}\sigres{(0)}}$} is solvable in $\asympO(\len{\fml})$ time (with no dependence on $n$).
}%
\theoremzeroadicFOEk
\end{theorem}

For \Cref{theorem: fo3,theorem: fo2,theorem: monadic},
we extract appropriate normal forms, namely, $\CoR$, $\BMLP$, and $\kC{1}$, respectively,
and we use translations to them while analyzing the sizes of the translations:
\begin{itemize}
    \item For $\CoR$:
    Tarski and Givant~\cite[Section 3.9]{tarskiFormalizationSetTheory1987} presented a translation from $\kFO{3}\sigres{(0, \ordinalomega)}$ (with at most two free "variables") to $\CoR$; see also, "eg", \cite{AndrekaAlgebraicLogic2001,givantCalculusRelationsFoundation2007,nakamuraExpressivePowerSuccinctness2022}.
    Such a translation can be computed in $2^{\asympO(\len{\fml})}$ time
\cite[Theorem 3.1]{nakamuraExpressivePowerSuccinctness2022};
    its implementation is in \cite{brogniTranslatingThreeVariableFirstOrder2025}.

    \item For $\BMLP$:
    Lutz, Sattler, and Wolter
    presented the Boolean modal logic called $\intro*\MLLSW$ in \cite{lutzModalLogicTwoVariable2001},
    which has an exponential-time translation from $\kFO{2}$ \cite[Theorem 1]{lutzModalLogicTwoVariable2001}.
    In this paper, we introduce $\BMLP$, a new Boolean modal logic inspired by $\MLLSW$,
    which has a $2^{\asympO(\len{\fml})}$-time translation from $\kFO{2}$.
    Unlike $\MLLSW$, $\BMLP$ has no term complement (which makes sparse relations dense), and thus admits an $\asympO(\len{\fml} \cdot m)$-time algorithm.

    \item For $\kC{1}$:
    Mints and Hoshi~\cite[Theorem 5]{mintsLogicalEquationsMonadic2009} presented
    a quantifier-elimination procedure for $\FO\sigres{(\ordinalomega)}$ ("ie", monadic FO with "equality"), yielding a normal form expressible in $\kC{1}$, but the complexity was not analyzed.
    In this paper, we present a translation to "DAG-represented" $\kC{1}$-"sentences" in $2^{\asympO(\len{\fml})}$ time by providing an explicit inductive transformation.
\end{itemize}

For unbounded "arity" in the sparse setting (\Cref{theorem: fo3,theorem: fo2}), we reduce the maximum "arity" to at most $2$:
\begin{align*}
\MC{\kFO{3}}
&~~\leadsto_{\text{("projection")}~\Cref{section: projection}} \MC{\kFO{3}\sigres{(\ordinalomega, \ordinalomega, \ordinalomega)}}~~
~~\leadsto_{\text{("guarded-pairing")}~\Cref{section: guarded-pairing}}~~ \MC{\kFO{3}\sigres{(0, \ordinalomega)}},\\
\MC{\kFO{2}}
&~~\leadsto_{\text{("projection")}~\Cref{section: projection}} \MC{\kFO{2}\sigres{(\ordinalomega, \ordinalomega)}}.
\end{align*}
After eliminating repeated "variables" inside "atoms" through the "projection" step (\Cref{lemma:reduction higher order predicates}),
we apply the "guarded-pairing" construction.
In the "guarded-pairing", a single "variable" encodes pairs of "variables" that appear together in "atoms".
For the "PNF" case, it can be done by fixing "variables" to constants \cite[Lemma 7.1]{gaoCompletenessFirstorderProperties2018}, but it does not work for non-"PNF" $\kFO{k}$, which allows nested "quantifiers" and hence reuse of "variables".
Our "guarded-pairing" construction (\Cref{lemma:sparse3:reduction higher order}) avoids "variable" fixing entirely by encoding only "guarded" pairs.

\begin{remark}\label{remark: DAG representation in algorithmic bounds}
In our translation-based upper bounds, "DAG" representations of "formulas" are crucial for keeping the blowup in formula "size" bounded by $2^{\asympO(\len{\fml})}$.
This use of "DAG" representations is needed for \Cref{theorem: monadic,lemma: guarded-pairing,proposition: fo3 to atomic CoR,proposition: fo2 to bml+c}.
\end{remark}

\section{Hardness: Reducing \texorpdfstring{\SparseMAXxSATProblem{(k-1)} to $\kPP{k}\sigres{(0,1)}$}{Sparse MAX-(k-1)-SAT to PPk[(0, 1)]}}
\label{section: reduction from OV Hyperclique to MCFO}
In this section, we prove \Cref{theorem:sparse-k4,theorem:sparse-k3} by
presenting a reduction from \SparseMAXxSATProblem{(k-1)} to $\MC{\kPP{k}\sigres{(0,1)}}$.

\subsection{To the \texorpdfstring{$(k,k-1)$-\OVHyperCliqueProblem{} problem}{(k,k-1)-OV-HyperClique problem}}\label{section: to OVHyperclique}
An $r$-""uniform hypergraph"" is a ""hypergraph"" $\graph = \tuple{V, E}$ where
$V$ is a finite set of ""vertices"" and
$E \subseteq \binom{V}{r}$ is a set of $r$-element undirected ""hyperedges"".
An $r$-uniform $k$-""hyperclique"", denoted by $(k,r)$-"hyperclique", is a $k$-element set $C \subseteq V$ such that $\binom{C}{r} \subseteq E$.
The $(k,r)$-""hyperclique problem"", denoted by $(k,r)$-\intro*\HyperCliqueProblem{}, is the following problem:
given an $r$-"uniform hypergraph" $\graph = \tuple{V, E}$,
decide whether there exists a $(k,r)$-"hyperclique" in $\graph$.
\ifcameraready\AP\phantomintro*\HyperCliqueHypothesis{}\fi%

\AP
For $t \ge 1$ and \kl{vectors} $v_1, \dots, v_t \in \set{0, 1}^{d}$, 
the ""generalized inner product"" is defined by $\intro*\GIP{v_1, \dots, v_t} \defeq \sum^d_{i = 1} \prod_{j = 1}^{t} v_j\ith{i}$,
which coincides with the standard inner product when $t = 2$.
We say that $v_1, \dots, v_t$ are ""generalized-orthogonal"" if $\GIP{v_1, \dots, v_t} = 0$.
Note that $\GIP{v_1, \dots, v_t}$ is invariant under permutations and repetitions of its arguments,
as the $v_j$ are 0/1 vectors.
We consider the ""OV-hyperclique problem"" \cite{almanOVGraphsAre2020}.\ifcameraready\AP\phantomintro{dimension}\fi{}
Given $k > r \ge 1$ and an $n$-element set $\mathcal{V} \subseteq \set{0, 1}^d$ of "vectors",
the $r$-""uniform OV-hypergraph"" $\intro*\OVgraph{\mathcal{V}, r}$ is defined as the $r$-"uniform hypergraph" $\OVgraph{\mathcal{V}, r} = \tuple{
\mathcal{V}, E
}$ where $S \in E$ "iff" $S$ is an $r$-element subset of $\mathcal{V}$ such that $\GIP{ S\ith{1}, \dots, S\ith{r} } = 0$.
Below we consider only \emph{"well-formed"} instances.
We say that an instance $\mathcal{V} \subseteq \set{0, 1}^d$ is ""well-formed"" if 
$\GIP{S \ith{1}, \dots, S \ith{r'}} \neq 0$ for every nonempty subset $S \subseteq \mathcal{V}$ with $r' \defeq \card{S} < r$.
The $(k, r)$-\intro*\OVHyperCliqueProblem{n, d} is the following problem:
given a "well-formed" $n$-element set $\mathcal{V} \subseteq \set{0, 1}^d$,
decide whether $\OVgraph{\mathcal{V}, r}$ has a $(k, r)$-"hyperclique".
In particular, we write \intro*\OVTriangleProblem{n, d} for the $(3, 2)$-\OVHyperCliqueProblem{n, d}.
We write $(k, r)$-\OVHyperCliqueProblem{n, \le c \log n} (and \OVTriangleProblem{n, \le c \log n}) for the restrictions to $d \le c \log n$.
The reduction of Alman and Vassilevska Williams \cite{almanOVGraphsAre2020} implies the corresponding hypotheses below (\Cref{proposition: alman VVW}).
Although our definition of \OVHyperCliqueProblem{} restricts inputs to be "well-formed",
this restriction does not weaken the lower bound.
\Cref{section: proposition: alman VVW} proves this by first
keeping the partite inputs in the Alman--Vassilevska Williams reduction and then
appending constant-length vectors that make no tuple of fewer than $r$ vectors "generalized-orthogonal" while preserving the orthogonality of exactly the $r$-tuples with at most one vector from each part.
\AP
\begin{hypothesis*}[$(k, r)$-\intro*\OVHyperCliqueHypothesis{} {\cite[Theorem~23]{almanOVGraphsAre2020}}%
]
Fix $k > r \ge 3$.
For any $\eps > 0$,
there is some $c$ such that
the $(k, r)$-\OVHyperCliqueProblem{n, \le c \log n} cannot be solved in $\asympO(n^{k - \eps})$ time.
\end{hypothesis*}
\begin{hypothesis*}[\intro*\OVTriangleHypothesis{}, $(3, 2)$-\OVHyperCliqueHypothesis{} {\cite[Theorem~5]{almanOVGraphsAre2020}}]
For any $\eps > 0$,
there is some $c$ such that
the \OVTriangleProblem{n, \le c \log n} cannot be solved in $\asympO(n^{\matrixexponent - \eps})$ time.
\end{hypothesis*}
\begin{proposition}[\Cref{section: proposition: alman VVW}]\label{proposition: alman VVW}
For each fixed $k > r \ge 3$ ("resp", $k = 3$ and $r = 2$),
the \SparseMAXxSATHypothesis{r} implies the $(k, r)$-\OVHyperCliqueHypothesis{} ("resp", \OVTriangleHypothesis{}).
\end{proposition}

\newcommand{\length}{L}
\subsection{To \texorpdfstring{$\kPP{k}\sigres{(\ordinalomega, k-1)}$}{PPk[(∞, k-1)]}}\label{section: to PPk}
We give a reduction from the $(k,r)$-\OVHyperCliqueProblem{n, d} to $\MC{\kPP{k}\sigres{(\ordinalomega, r)}}$, and use it with $r=k-1$ for the present hardness result.
The parameter $B$ below is the number of "blocks" into which we split "dimensions"; it will become the number of unary "relation symbols".
\begin{lemma}\label{lemma: new new partitioned binary hyperclique hypothesis}
Fix $k > r \ge 2$ and $B \ge 1$.
There is a \emph{fixed} $\kPP{k}\sigres{(B, r)}$-"sentence" $\fml$ with the following property:
For every instance $\mathcal{V}$ of the $(k,r)$-\OVHyperCliqueProblem{n, d},
where $d$ is divisible by $B$,
one can construct, in $\asympO_{k, r, B}(r^{d/B + 1} n)$ time, a "structure" $\struc_{\mathcal{V}, r}$ such that
\[\text{$\OVgraph{\mathcal{V}, r}$ has a $(k, r)$-"hyperclique"} \quad\iff\quad \struc_{\mathcal{V}, r} \models \fml.\]
Moreover, $\struc_{\mathcal{V}, r}$ satisfies $\strucn{\struc_{\mathcal{V}, r}} = \asympO_{k, r, B}(r^{d/B} + n)$ and $\strucm{\struc_{\mathcal{V}, r}} = \asympO_{k, r, B}(r^{d/B + 1} n)$.
Additionally, $\len{\fml} = \asympO(k + \binom{k}{r} rB) = \asympO_{k, r, B}(1)$.
\end{lemma}

\AP
For a "vector" $v \in \set{0, 1}^{q}$, we write $\intro*\Zeropos(v) \defeq \set{ i \in \range{q} \mid v\ith{i} = 0 }$.
\phantomintro{choice function}%
The following is clear:
\[\GIP{ y_1, y_2, \ldots, y_r } = 0 \;\Longleftrightarrow\; \text{ there is a "choice function" $P \colon \range{q} \to \range{r}$ "st" $\forall j \in \range{r}.\, P^{-1}(j) \subseteq \Zeropos(y_{j})$}.\]
For example, when $y_1 = 10\,01$, $y_2 = 11\,11$, and $y_3 = 01\,00$,
$\GIP{ y_1, y_2, y_3 } = 0$ is shown by the "choice function" $P$ defined as $1,2,3,4 \mapsto 3,1,1,3$.

We partition each $d$-"dimensional" "vector" $y$ into $B$ ""blocks"",
each of "dimension" $\length \defeq d/B$, as follows:
\def\bit{\lambda}
\[
y ~=~
\underbracket{\,\bit_1 \bit_2 \ldots \bit_{\length}\,}_{y^{(1)}~\text{"block" 1}}\ 
\underbracket{\,\bit_{\length+1} \bit_{\length+2} \ldots \bit_{2\length}\,}_{y^{(2)}~\text{"block" 2}} \cdots
\underbracket{\,\bit_{(b-1)\length+1} \bit_{(b-1)\length+2} \ldots \bit_{b\length}\,}_{y^{(b)} ~\text{"block" $b$}} \cdots
\underbracket{\,\bit_{(B-1)\length+1} \bit_{(B-1)\length+2} \ldots \bit_{B\length}\,}_{y^{(B)}~\text{"block" $B$}}.
\]
For $b \in \range{B}$ and $y \in \set{0, 1}^{d}$, let $y^{(b)}$ denote the $b$-th "block" of $y$, "ie", $y^{(b)}\ith{\ell} = y\ith{(b-1)\cdot\length+\ell}$ for $\ell \in \range{\length}$.
Our reduction is based on the following equivalence:
\begin{align*}
\GIP{ y_1, y_2, \ldots, y_r } = 0
&~\iff~ \forall b \in \range{B}.
\exists P \colon \range{\length} \to \range{r}.\forall j \in \range{r}.\,
P^{-1}(j) \subseteq \Zeropos(y_{j}^{(b)}).
\end{align*}

\begin{proof}[Proof of \Cref{lemma: new new partitioned binary hyperclique hypothesis}]
Let $\fml$ be the $\kPP{k}\sigres{(B, r)}$-"sentence" defined as follows:%
\AP\phantomintro{\COVER}
\begin{align*}
\fml &\defeq \exists x_1 \dots \exists x_k.
\bigwedge_{S \in \binom{\range{k}}{r}} \COVER(S), & \text{where }
\reintro*\COVER(S) &\defeq \bigwedge_{b=1}^B \exists x_{\mu(S)}. \left(U_b(x_{\mu(S)}) \land \bigwedge_{j = 1}^{r} E_{j}(x_{\mu(S)}, x_{S\ith{j}}) \right),
\end{align*}
where $\mu(S) \defeq \min(\range{k} \setminus S)$,
which exists since $k > r$.
By construction, $\fml$ uses only the "variables" $x_1, \dots, x_k$, and hence $\fml \in \kPP{k}\sigres{(B,r)}$.
Moreover,
$\len{\fml} = \asympO(k + \binom{k}{r} \cdot B r) = \asympO_{k, r, B}(1)$, "ie", $\fml$ is fixed.
Intuitively, the outer existential quantifiers of \(\fml\) choose
\(k\) candidate "vectors" \(v_1, \dots, v_k \in \mathcal{V}\).
For each \(r\)-subset \(S = \set{i_1, i_2, \dots, i_r} \subseteq \range{k}\), the subformula \(\COVER(S)\) asserts that
the selected "vectors" \(v_{i_1}, \dots, v_{i_r}\) are generalized-orthogonal:
for every "block" \(b\), there is a "choice function" indicated by \(x_{\mu(S)}\),
certifying that each coordinate of the "block" is zero in at least one of those \(r\) "vectors".

Let $\mathcal{V}$ be an instance of $(k, r)$-\OVHyperCliqueProblem{n, d}.
We define the "structure" $\struc_{\mathcal{V}, r} \defeq \struc \defeq \tuple{\,\univ{\struc}, \set{U_b^{\struc}}_{b \in \range{B}}, \set{E_j^{\struc}}_{j \in \range{r}}}$ as follows:
\begin{align*}
\univ{\struc}
&\defeq \mathcal{V} \sqcup \smashoperator{\bigsqcup_{b \in \range{B}}} U_b^{\struc},&
U_b^{\struc}
&\defeq \set{ b } \times \range{r}^{\range{\length}},&
E_j^{\struc}
&\defeq \smashoperator{\bigcup_{b \in \range{B}}} \{ \tuple{\tuple{b, P},~ v} \in U_b^{\struc} \times \mathcal{V} \mid P^{-1}(j) \subseteq \Zeropos(v^{(b)}) \}.
\end{align*}
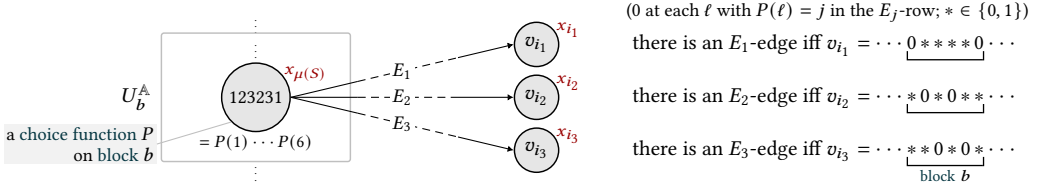
\begin{figure}[t]
    \centering
\begin{tikzpicture}[
  font=\small,
  block/.style={
    draw = gray!50,
    minimum width=2cm,
    minimum height=0.6cm,
    rounded corners=1pt,
    line width=0.55pt,
    align=center
  },
  vertex/.style={
    mynode,
    circle,
    draw,
    inner sep=2pt,
    line width=0.55pt
  },
  arrow/.style={
    -{Latex[length=1.8mm]},
    line width=0.55pt
  },
  form/.style={
    align=center,
    minimum width=3.2cm
  },
  token/.style={
    mynode,
    rectangle,
    rounded corners=3pt,
    inner xsep=5pt,
    inner ysep=3pt,
    line width=0.55pt
  },
  labelbox/.style={
    align=left,
    fill=gray!20,
    inner sep=2pt
  },
  tiny/.style={font=\scriptsize}
]
\node at (-2.6, 1.0) {\tiny \scalebox{.6}{$\vdots$}};
\node[block, minimum width=2.5cm,
    minimum height=1.7cm] (blockb) at (-2.6, 0) {};
\node[left=0pt of blockb] {\footnotesize $U_b^{\struc}$};
\node[token, vertex] (seq) at (blockb.center) {$123231$};
\node at (-2.6,-0.95) {\tiny \scalebox{.6}{$\vdots$}};
\node at (-2.6, 0.67) {\tiny \scalebox{.6}{$\vdots$}};
\node[font=\tiny, inner sep=0pt] at (-2.6,-0.6) {$= P(1) \cdots P(6)$};
\node[anchor=north east, font=\scriptsize, inner sep=1pt, align = right, fill = gray!10] (cfnote) at ($(blockb.west)+(-0.05,-0.35)$) {a \kl{choice function} $P$\\ on \kl{block} $b$};
\path (cfnote.east) edge[draw = gray!50] (seq.south west);
\node[above right = -1.6ex and -.2em of seq, text = red!60!black](varseq) {\tiny $x_{\mu(S)}$};

\node[minimum width=1.2cm] (part1) at (1.3, .7) {};
\node[vertex] (v1) at ($(part1.west)+(0.42,0)$) {\footnotesize $v_{i_1}$};
\node[above right = -1.6ex and -.2em of v1, text = red!60!black]{\tiny $x_{i_1}$};
\node[minimum width=1.2cm] (part2) at (1.3, 0) {};
\node[vertex] (v2) at ($(part2.west)+(0.42,0)$) {\footnotesize $v_{i_2}$};
\node[above right = -1.6ex and -.2em of v2, text = red!60!black]{\tiny $x_{i_2}$};
\node[minimum width=1.2cm] (part3) at (1.3,-.7) {};
\node[vertex] (v3) at ($(part3.west)+(0.42,0)$) {\footnotesize $v_{i_3}$};
\node[above right = -1.6ex and -.2em of v3, text = red!60!black]{\tiny $x_{i_3}$};

\node[anchor=west, align=left, right = 1em of part1.east, yshift = -1ex] (row1) {\footnotesize there is an $E_1$-edge iff $v_{i_1} =$
    $\cdots \underbracket[0.5pt][3pt]{0\,{\ast}\,{\ast}\,{\ast}\,{\ast}\,0}_{\phantom{\text{block}\ b}} \cdots$
};
\node[anchor=west, align=left, right = 1em of part2.east, yshift = -1ex] {\footnotesize there is an $E_2$-edge iff $v_{i_2} =$
    $\cdots \underbracket[0.5pt][3pt]{{\ast}\,0\,{\ast}\,0\,{\ast}\,{\ast}}_{\phantom{\text{block}\ b}} \cdots$
};
\node[anchor=west, align=left, right = 1em of part3.east, yshift = -1ex] {\footnotesize there is an $E_3$-edge iff $v_{i_3} =$
    $\cdots \underbracket[0.5pt][3pt]{{\ast}\,{\ast}\,0\,{\ast}\,0\,{\ast}}_{\text{\kl{block}}\ b} \cdots$
};
\node[anchor=south west, font=\scriptsize, inner sep=1pt] at (row1.north west) {($0$ at each $\ell$ with $P(\ell)=j$ in the $E_j$-row; ${\ast} \in \set{0,1}$)};
\foreach \v/\j in {v1/1, v2/2, v3/3} {
  \draw (seq.east) -- ($(seq.east)!0.30!(\v.west)$);
  \draw[densely dashed] ($(seq.east)!0.30!(\v.west)$) -- ($(seq.east)!0.75!(\v.west)$);
  \draw[earrow, ->] ($(seq.east)!0.75!(\v.west)$) -- (\v.west);
  \node[fill=white, inner sep=1pt] at ($(seq.east)!0.5!(\v.west)$) {\scriptsize $E_{\j}$};
}
\end{tikzpicture}     \vspace{-1.5ex}
    \caption{
    Illustration of a \kl{structure} in \Cref{lemma: new new partitioned binary hyperclique hypothesis}, where $k > r = 3$ and $\length = 6$.
    The \kl{formula} in \Cref{lemma: new new partitioned binary hyperclique hypothesis} 
    is defined so that there are "vectors" $v_1, \dots, v_k$ in $\mathcal{V}$ indicated by $x_1, \dots, x_k$, respectively, such that 
    for each $S = \set{i_1, i_2, i_3} \in \binom{\range{k}}{r}$ with $i_1 < i_2 < i_3$ and each \kl{block} $b$,
    there is a "choice function" $P \colon \range{\length} \to \range{r}$ indicated by $x_{\mu(S)}$ such that $v_{i_{P(\ell)}}^{(b)}\ith{\ell} = 0$ for all $\ell \in \range{\length}$.
    }
    \label{figure: reduction from HClique to PPk}
    \Description{Illustration of the reduction from the (k,r)-OV-hyperclique problem to model checking of PP^k[(omega,r)].}
\end{figure}
See \Cref{figure: reduction from HClique to PPk} for an illustration of $\struc_{\mathcal{V}, r}$.
We have:
$\strucn{\struc}
= \asympO(n + Br^{\length})
= \asympO_{k, r, B}(r^{\length} + n)$ and 
$\strucm{\struc}
= \asympO(B r^{\length} + B r^{\length + 1} n)
= \asympO_{k, r, B}(r^{\length + 1} n)$.
The "relations" can be constructed in $\asympO(B r^{\length + 1} n)$ time,
"eg", by implementing the subset tests $P^{-1}(j) \subseteq \Zeropos(v^{(b)})$ using bitmasks per "block".
Thus, the total construction time is $\asympO((B r^{\length + 1} + d) n) = \asympO_{k, r, B}(r^{\length + 1} n)$.
For correctness,
for every $v_1, \dots, v_k \in \mathcal{V}$ and $S \in \binom{\range{k}}{r}$,
we have:
\begin{align*}
    &\GIP{v_{S\ith{1}}, \ldots, v_{S\ith{r}}} = 0
    \Leftrightarrow \forall b \in \range{B}.\,\exists P \colon \range{\length} \to \range{r}.\forall j \in \range{r}.\,
P^{-1}(j) \subseteq \Zeropos(v_{S\ith{j}}^{(b)}) %
    \\
    &\Leftrightarrow \forall b \in \range{B}.\,\exists P \colon \range{\length} \to \range{r}.~
    \struc, \set{x_i \mapsto v_i \mid i \in S} \cup \set{x_{\mu(S)} \mapsto \tuple{b, P}} \models \bigwedge_{j = 1}^{r} E_{j}(x_{\mu(S)}, x_{S\ith{j}}) \\
    &\Leftrightarrow \struc, \set{x_{i} \mapsto v_i \mid i \in S} \models \COVER(S).
\end{align*}
We show that
$\OVgraph{\mathcal{V}, r}$ has a $(k, r)$-"hyperclique"
"iff"
$\struc \models \fml$.
\proofcase{($\Rightarrow$)}
Straightforward.
\proofcase{($\Leftarrow$)}
Suppose $\struc \models \fml$,
witnessed by $a_1,\ldots,a_k$.
By the definition of $\fml$, for each $S \in \binom{\range{k}}{r}$ and each $b \in \range{B}$, there is some $a_{S, b}$
(the witness for the existential quantifier $\exists x_{\mu(S)}$ in the $b$-th conjunct of $\COVER(S)$; it may depend on $b$) such that
$\struc, \set{x_i \mapsto a_i \mid i \in S} \cup \set{x_{\mu(S)} \mapsto a_{S, b}} \models E_{j}(x_{\mu(S)}, x_{S\ith{j}})$ for each $j \in \range{r}$,
and then $a_{S\ith{j}} \in \mathcal{V}$.
Thus, $a_1, \dots, a_k \in \mathcal{V}$.
By the above equivalence, for each $r$-element subset $S \in \binom{\range{k}}{r}$,
we have $\GIP{a_{S\ith{1}}, \ldots, a_{S\ith{r}}} = 0$.
Because the input $\mathcal{V}$ is "well-formed",
$a_{S\ith{1}}, \ldots, a_{S\ith{r}}$ are pairwise distinct for each $S \in \binom{\range{k}}{r}$, and hence $a_1, \dots, a_k$ are pairwise distinct.
Hence, $\set{a_1,\ldots,a_k}$ is a $(k,r)$-"hyperclique" in $\OVgraph{\mathcal{V}, r}$.
\end{proof}

By \Cref{lemma: new new partitioned binary hyperclique hypothesis},
we have the following.
\begin{lemma}\label{lemma:sparse-k4}
Fix $k > r \geq 3$.
Fix $\eps > 0$.
Assuming the $(k, r)$-\OVHyperCliqueHypothesis{},
there is a $\kPP{k}\sigres{(\ordinalomega, r)}$-"sentence" $\fml_{\eps}$ such that
$\MC{\set{\fml_{\eps}}}$ cannot be solved in $\asympO(m^{k - \eps})$ time.
\end{lemma}
\begin{proof}
We may assume $\eps < 1$ (otherwise, take $\fml_{\eps} \defeq \fml_{1/2}$).
Let $\delta \defeq \eps/2$ and let $c$ be the constant supplied by the
$(k,r)$-\OVHyperCliqueHypothesis{} for the exponent loss $\delta$.

Choose $B$ so large that $\gamma \defeq c/B \cdot \log r < \delta/k$, or equivalently, $k\gamma < \delta$.
Note that $B$ is chosen after $\eps$ and after the constant $c$ (here, $c$ is determined by $\eps$).%
\footnote{This is the only source of the dependence of the hard "sentence" on $\eps$.}
Let $\fml_{\eps}$ be the fixed "sentence" produced by \Cref{lemma: new new partitioned binary hyperclique hypothesis} for the parameters $k$, $r$, and $B$.
Suppose, toward a contradiction, that $\MC{\set{\fml_{\eps}}}$ can be solved in $\asympO(m^{k - \eps})$ time.
Let $\mathcal{V} \subseteq \set{0, 1}^d$ be an arbitrary instance of $(k,r)$-\OVHyperCliqueProblem{n, \le c \log n}.
Without loss of generality, we may assume that $B$ divides $d$ by padding each "vector" with at most $B-1$ zeros.
This padding does not change the "generalized orthogonality@generalized-orthogonal" of any tuple of "vectors", so the instance stays "well-formed" and the answer is unchanged.
By $d \le c \log n + B - 1$, we have:
\[
r^{{d/B}}
~\le~ r^{(\frac{c}{B} \cdot \log n) + 1}
~=~ r \cdot n^{\frac{c}{B} \cdot \log r}
~=~ \asympO(n^{\gamma}).
\]
Let $\struc \defeq \struc_{\mathcal{V}, r}$ be given by \Cref{lemma: new new partitioned binary hyperclique hypothesis}.
By $r^{{d/B}} = \asympO(n^{\gamma})$ and $\gamma < \delta/k < 1$, we have:
\begin{align*}
\strucn{\struc} &= \asympO(n^{\gamma} + n) = \asympO(n), &\strucm{\struc} &= \asympO(n^{1+\gamma}), & \text{and the construction time is } \asympO(n^{1 + \gamma}).
\end{align*}
By applying the assumed $\asympO(m^{k - \eps})$-time algorithm,
$(k,r)$-\OVHyperCliqueProblem{n, \le c \log n} can be solved in time $\asympO(n^{1+\gamma}) + \asympO(n^{(1+\gamma)(k - \eps)}) \le \asympO(n^{k - \delta})$, where
$(1+\gamma)(k-\eps)
< k-\eps + k\gamma
< k-\eps + \delta
= k-\delta$
by $k\,\gamma <\delta$ and $\eps = 2 \delta$.
This contradicts the choice of $c$.
Hence, $\MC{\set{\fml_{\eps}}}$ cannot be solved in $\asympO(m^{k - \eps})$ time.
\end{proof}

\subsection{To \texorpdfstring{$\kPP{k}\sigres{(0, 1)}$}{PPk[(0,1)]}}\label{section: to PPk one binary relation}
Finally, a unary encoding of multiple binary "relations" gives the following reduction:
\begin{lemma}[\Cref{section: lemma: reduction monadic and binary predicates to single binary predicate}\ifcameraready\else{}; see \Cref{section: lemma: reduction monadic and binary predicates to single binary predicate for tree} for the "tree-represented" case\fi]\label{lemma: reduction monadic and binary predicates to single binary predicate}
\gdef\reductionmonadicandbinarypredicatestosinglebinarypredicatelemma{
Fix $k \ge 3$.
Let $\fml$ be a $\kPP{k}\sigres{(C_1,C_2)}$ ("resp", $\kEP{k}\sigres{(C_1,C_2)}$, $\kFO{k}\sigres{(C_1,C_2)}$)-"sentence", and let $\struc$ be a "structure".
Let $C \defeq \card{\fmlsig{\fml}} + 1$, where the extra binary relation is used
to eliminate "equality".
One can construct, in $\asympO(\len{\fml} + C \, \strucn{\struc} + \strucm{\struc})$ time,
a $\kPP{k}\sigres{(0,1)}$ ("resp", $\kEP{k}\sigres{(0,1)}$, $\kFO{k}\sigres{(0,1)}$)-"sentence" $\fml[2]$ without "equality" and a "structure" $\struc[2]$
such that $\struc \models \fml$ "iff" $\struc[2] \models \fml[2]$.
Moreover,
$\len{\fml[2]} = \asympO(\len{\fml})$,
$\strucn{\struc[2]} = \asympO(C \, \strucn{\struc})$, and
$\strucm{\struc[2]} = \asympO(C \, \strucn{\struc} + \strucm{\struc})$.}
\reductionmonadicandbinarypredicatestosinglebinarypredicatelemma
\end{lemma}
Using \Cref{lemma: reduction monadic and binary predicates to single binary predicate},
we can reduce $\kPP{k}\sigres{(\ordinalomega, r)}$ to $\kPP{k}\sigres{(0, 1)}$.
\begin{lemma}\label{lemma:sparse-k4'}
Fix $k > r \geq 3$ and $\eps > 0$.
Assuming the $(k, r)$-\OVHyperCliqueHypothesis{},
there exists a $\kPP{k}\sigres{(0, 1)}$-"sentence" $\fml[2]_{\eps}$ such that
$\MC{\set{\fml[2]_{\eps}}}$ cannot be solved in $\asympO(m^{k - \eps})$ time.
\end{lemma}
\begin{proof}
As before, we may assume $\eps < 1$, so that $k - \eps > 1$.
Let $\fml_{\eps}$ be the $\kPP{k}\sigres{(\ordinalomega, r)}$-"sentence" given by \Cref{lemma:sparse-k4} for the parameters $k$, $r$, and $\eps$.
Let $\fml[2]_{\eps}$ be the "sentence" obtained from $\fml_{\eps}$ by \Cref{lemma: reduction monadic and binary predicates to single binary predicate}; it depends only on $\eps$.
Let $\struc$ be an arbitrary "structure".
Let $\widehat{\struc}$ be the "structure" obtained from $\struc$ by applying \Cref{proposition: isolated}.
Then $\strucn{\widehat{\struc}} = \asympO_{\eps}(\strucm{\struc})$ and $\strucm{\widehat{\struc}} = \asympO(\strucm{\struc})$.
Let $\struc[2]$ be the "structure" obtained from $\widehat{\struc}$ by applying \Cref{lemma: reduction monadic and binary predicates to single binary predicate} to $\fml_{\eps}$ and $\widehat{\struc}$.
Then $\strucn{\struc[2]} = \asympO_{\eps}(\strucn{\widehat{\struc}}) = \asympO_{\eps}(\strucm{\struc})$, and $\strucm{\struc[2]} = \asympO_{\eps}(\strucn{\widehat{\struc}} + \strucm{\widehat{\struc}}) = \asympO_{\eps}(\strucm{\struc})$.
Toward a contradiction, we assume that $\MC{\set{\fml[2]_{\eps}}}$ is solvable in $\asympO(m^{k - \eps})$ time.
Then, by constructing $\struc[2]$ from $\struc$ and running the assumed algorithm on $\struc[2]$, $\MC{\set{\fml_{\eps}}}$ can be solved in time
$\asympO_{\eps}(\strucm{\struc}) + \asympO(\strucm{\struc[2]}^{k-\eps})
\le \asympO_{\eps}(\strucm{\struc}^{k-\eps})$,
which contradicts \Cref{lemma:sparse-k4}.
\end{proof}
Moreover, by the same argument as in the proof of \Cref{lemma:sparse-k4'} with $k = 3$, we have the following.
Note that \Cref{lemma: new new partitioned binary hyperclique hypothesis} works for $k = 3$ as well.
\ifthenelse{\boolean{arxiv}}{}{See the extended version~\cite{nakamuraUezatoDichotomyBoundedVariable2026arxiv} for the proof.}
\ifthenelse{\boolean{arxiv}}{%
\begin{lemma}[\Cref{section:lemma:sparse-k3'}]%
}{%
\begin{lemma}%
}\label{lemma:sparse-k3'}
\gdef\sparsekthreeprimelemma{%
Fix $\eps > 0$.
Assuming the \OVTriangleHypothesis{},
there exists a $\kPP{3}\sigres{(0, 1)}$-"sentence" $\fml[2]_{\eps}$ such that
$\MC{\set{\fml[2]_{\eps}}}$ cannot be solved in $\asympO(m^{\matrixexponent - \eps})$ time.
}%
\sparsekthreeprimelemma
\end{lemma}
\Cref{theorem:sparse-k4} follows from \Cref{lemma:sparse-k4'}
and the implication from the \SparseMAXxSATHypothesis{(k-1)} to the
$(k,k-1)$-\OVHyperCliqueHypothesis{} for $k\ge4$
(\Cref{proposition: alman VVW}; "cf" \cite[Theorem~23]{almanOVGraphsAre2020}).
Similarly, \Cref{theorem:sparse-k3} follows from
\Cref{lemma:sparse-k3'} and the implication from
the \SparseMAXxSATHypothesis{2} to the \OVTriangleHypothesis{}
(\Cref{proposition: alman VVW}; "cf" \cite[Theorem~5]{almanOVGraphsAre2020}).

\section{\texorpdfstring{$\kFO{3}$ Model Checking via Tarski's $\CoR$}{FO3 Model Checking via Tarski's CoR}}\label{section: fo3}
In this section, we show \Cref{theorem: fo3}.
\AP
We first recall ""Tarski's calculus of relations"" ($\intro*\CoR$) \cite{tarskiCalculusRelations1941}.%

\AP
$\CoR$ ""terms@@CoR"" and $\CoR$ ""atomic FO-formulas@@CoR"" are generated by the following grammar,
where $R$ ranges over binary "relation symbols" and $x$ and $y$ range over "variables":
$\phantomintro*{\union}$
$\phantomintro*{\intersection}$
$\phantomintro*{\compo}$
$\phantomintro*{\corid}$
\begin{align*}
    \term[1], \term[2] ~\Coloneqq~
    & R \mid \corid \mid \term[1]^{\smile} \mid \term[1] \intersection \term[2] \mid \term[1] \union \term[2] ~\mid~ \term[1] \compo \term[2] \mid \term[1]^{-}, \tag*{($\CoR$ "terms@@CoR")}\\
    \fml[1], \fml[2] ~\Coloneqq~
    & \term[1](x, y). %
    \tag*{($\CoR$ "atomic FO-formulas@@CoR")}
\end{align*}
Here, $\term[1]$ and $\term[2]$ are meta-variables ranging over $\CoR$ "terms@@CoR", and $\fml[1]$ and $\fml[2]$ are meta-variables ranging over $\CoR$ "atomic FO-formulas@@CoR".
We write $\intro*\cortop \defeq \corid \union \corid^{-}$ for the full "relation"
and write $\intro*\corz \defeq \cortop^{-}$ for the empty "relation".

Given a "structure" $\struc$ over a "signature" consisting only of binary "relation symbols",
the ""semantics@@CoR"" $\sem{\term}^{\struc} \subseteq \univ{\struc}^2$ is defined as follows:
\begin{align*}
    \sem{R}^{\struc} & \defeq R^{\struc},&
    \sem{\corid}^{\struc} &\defeq \set{\tuple{x, x} \mid x \in \univ{\struc}},&
    \sem{\term[1]^{\smile}}^{\struc} &\defeq \set{ \tuple{x, y} \mid \tuple{y, x} \in \sem{\term[1]}^{\struc} },\\
    \sem{\term[1] \intersection \term[2]}^{\struc} &\defeq \sem{\term[1]}^{\struc} \cap \sem{\term[2]}^{\struc},&
    \sem{\term[1] \union \term[2]}^{\struc} &\defeq \sem{\term[1]}^{\struc} \cup \sem{\term[2]}^{\struc},&
    \sem{\term[1]^{-}}^{\struc} &\defeq \univ{\struc}^2 \setminus \sem{\term[1]}^{\struc},\\
    \sem{\term[1] \compo \term[2]}^{\struc} &\defeq \set{\tuple{x, z} \mid \text{$\exists y \in \univ{\struc}$, } \tuple{x, y} \in \sem{\term[1]}^{\struc} \text{ and } \tuple{y, z} \in \sem{\term[2]}^{\struc} }. \span\span\span\span
\end{align*}
\begin{proposition}\label{proposition: CoR term evaluation}
Let $\struc$ be an $n$-element "structure" over a "signature" consisting only of binary "relation symbols".
For a $\CoR$ term $\term$, we can compute $\sem{\term}^{\struc}$ in $\asympO(\len{\term}) \cdot n^{\matrixexponent}$ time.
\end{proposition}
\begin{proof}
We evaluate $\term$ in a bottom-up manner.
Given $\sem{\term[2]}^{\struc}$ and $\sem{\term[3]}^{\struc}$ where $\term[2]$ and $\term[3]$ are ""subterms"" of $\term[1]$,
we can compute
$\sem{\term[2] \intersection \term[3]}^{\struc}$, 
$\sem{\term[2] \union \term[3]}^{\struc}$, 
$\sem{\term[2]^{\smile}}^{\struc}$, and 
$\sem{\term[2]^{-}}^{\struc}$ in $\asympO(n^2)$ time.
For the composition $\sem{\term[2] \compo \term[3]}^{\struc}$, we use an $\asympO(n^{\matrixexponent})$-time Boolean matrix multiplication algorithm.
\end{proof}

For $\CoR$ "atomic FO-formulas@@CoR" $\term(x, y)$,
we define $\struc, I \models \term(x, y)$ if and only if $\tuple{I(x), I(y)} \in \sem{\term}^{\struc}$.
The ""model-checking problem@@CoR"", denoted by $\intro*\CoRMC{\CoR}$, is the following problem:
given a "structure" $\struc$, a "valuation" $I$, and a $\CoR$ "atomic FO-formula@@CoR" $\fml = \term(x, y)$, decide whether $\struc, I \models \fml$.

\begin{proposition}\label{proposition: CoR model checking}
$\CoRMC{\CoR}$ is solvable in $\asympO(\len{\fml}) \cdot n^{\matrixexponent}$ time.
\end{proposition}
\begin{proof}
Let $\fml = \term(x, y)$.
Given $\sem{\term}^{\struc}$, we can check whether $\tuple{I(x), I(y)} \in \sem{\term}^{\struc}$ or not in $\asympO(1)$ time.
Thus, using \Cref{proposition: CoR term evaluation}, we can decide $\struc, I \models \fml$ in $\asympO(\len{\fml}) \cdot n^{\matrixexponent}$ time.
\end{proof}

\subsection{Dense model checking of \texorpdfstring{$\kFO{3}$}{FO3}}
We recall the translation from $\kFO{3}\sigres{(0,\ordinalomega)}$-"formulas" with at most two free "variables" to $\CoR$ "atomic FO-formulas@@CoR" \cite{tarskiFormalizationSetTheory1987}.
Such a translation can be computed in $2^{\asympO(\len{\fml})}$ time \cite[Theorem 3.1]{nakamuraExpressivePowerSuccinctness2022}.%
\footnote{%
Here we introduce only five operators, rather than the richer set of operators used in \cite{nakamuraExpressivePowerSuccinctness2022} to track the correspondence with the quantifier hierarchy.
This does not affect \Cref{proposition: fo3 to atomic CoR}, since each omitted operator is expressible using the operators introduced here with only a constant-factor blowup in the "formula" size.}
We say that two "formulas" ($\FO$-"formulas" or $\CoR$ "atomic FO-formulas@@CoR") $\fml[1]$ and $\fml[2]$ are ""semantically equivalent"" if $\struc, I \models \fml[1]$ and $\struc, I \models \fml[2]$ have the same "truth value" for every "structure" $\struc$ and "valuation" $I$.%
\begin{proposition}[{\cite[Theorem 3.1]{nakamuraExpressivePowerSuccinctness2022}}]
\label{proposition: fo3 to atomic CoR}
Let $\fml$ be a "DAG-represented" $\kFO{3}\sigres{(0,\ordinalomega)}$-"formula" with at most two free "variables".
One can translate $\fml$ into a "semantically equivalent" "DAG-represented" $\CoR$ "atomic FO-formula@@CoR" in $2^{\asympO(\len{\fml})}$ time.
\end{proposition}
\begin{proof}
We use the same bottom-up translation as \cite[Theorem 3.1]{nakamuraExpressivePowerSuccinctness2022}.
For "DAG-represented" inputs, we memoize the translation of each triple
$\tuple{\fml[2], x, y}$, where $\fml[2]$ is a "subformula" and $\FV(\fml[2]) \subseteq \set{x, y}$,
so that a shared "subformula" is translated only once for each possible two-variable interface.
The exponential blowup comes only from the local DNF/CNF-like expansion at each quantifier, whose leaves refer to memoized translations;
thus, the translated $\CoR$ "DAG" has "size" $2^{\asympO(\len{\fml})}$.%
\end{proof}
The key step of the translation is the following quantifier elimination:
\begin{align*}
\fml(x, z) \equiv \exists y. P(x, y) \land Q(y, z)
\quad \leadsto \quad 
\fml[2](x, z) \equiv (P \compo Q)(x, z). \tag{QE} \label{QE}
\end{align*}
When the input is a "sentence", we introduce dummy "variables" for $\CoR$ "atomic FO-formulas@@CoR";
for instance, $\exists x \exists y. R(x, y)$ is translated to $(\cortop \compo R \compo \cortop)(v, v)$, where $v$ is a dummy "variable".
\ifarxiv See also \Cref{section: example: proposition: fo3 to atomic CoR} for more examples including the case when the input is a "sentence".\fi%
Now we can prove \Cref{theorem: fo3} for the dense case.
\begin{proof}[Proof of \Cref{theorem: fo3} in the dense case]
By encoding each unary "relation" $U$
as the diagonal binary "relation" $D_U \defeq \set{\tuple{x, x} \mid x \in U}$
(\Cref{section: proposition: arity monotone}),
it suffices to consider $\kFO{3}\sigres{(0, \ordinalomega)}$-"sentences".
The latter case immediately follows from \Cref{proposition: fo3 to atomic CoR,proposition: CoR model checking}.
\end{proof}

\subsection{Sparse model checking of \texorpdfstring{$\kFO{3}$}{FO3}}

In this subsection, we show \Cref{theorem: fo3} in the sparse case.
We reduce it to the dense case over binary "signatures":
we first reduce the maximum "arity" to three by "projection" (\Cref{section: projection}),
which also eliminates repeated "variables" inside "atoms",
and then reduce the "arity" to two by "guarded-pairing" (\Cref{section: guarded-pairing}).

\subsubsection{"Projection" (reducing the maximum arity to $k$)}
\label{section: projection}
Here, for $\kFO{k}$,
we present a reduction that reduces the maximum "arity" to at most $k$,
using fresh "relations" for encoding ""projections"".

\begin{lemma}
\label{lemma:reduction higher order predicates}
\label{lemma: projection}
\gdef\lemmaTranslatingHigherArityPredicatesToKPredicates{
Fix an integer $k \geq 1$.
Let $\fml$ be an $\kFO{k}$-"sentence" and $\struc$ be a "structure".
One can construct, in $\asympO(\len{\fml} \cdot \strucm{\struc} + \strucn{\struc})$ time, an $\kFO{k}\sigres{(\card{\Atoms_{1}(\fml)}, \dots, \card{\Atoms_{k}(\fml)})}$-"sentence" $\fml[2]$ and a "structure" $\struc[2]$ such that $\struc \models \fml$~"iff"~$\struc[2] \models \fml[2]$,
where $\intro*\Atoms_{\ell}(\fml)$ denotes the set of non-"equality" "atoms" of $\fml$ with exactly $\ell$ distinct "variables" and $\Atoms(\fml) \defeq \bigcup_{\ell = 1}^k \Atoms_{\ell}(\fml)$.
Moreover,
$\strucn{\struc[2]} = \strucn{\struc}$,
$\strucm{\struc[2]} = \asympO(\len{\fml} \cdot \strucm{\struc})$, and
$\len{\fml[2]} = \asympO(\len{\fml})$.
}
\lemmaTranslatingHigherArityPredicatesToKPredicates
\end{lemma}
\begin{proof}
Fix "variables" $x_1, \dots, x_k$ including $\vars(\fml)$.
For each "atom" $\alpha = R(x_{i_1}, \ldots, x_{i_p})$ occurring in $\fml$ (without "equality"),
let $j_1 < \dots < j_\ell$ be the distinct indices among $i_1, \dots, i_p$, and
let $P_\alpha$ be a fresh $\ell$-ary "relation symbol".
We then define $\fml[2]$ and $\struc[2]$ as follows:
\begin{itemize}
    \item $\fml[2]$ is the "formula" $\fml$ in which
    each non-"equality" "atom" $\alpha$ in $\fml$ has been replaced with $P_{\alpha}(x_{j_1}, \dots, x_{j_\ell})$;
    \item $\struc[2]$ is the "structure" such that
    $\univ{\struc[2]} \defeq \univ{\struc}$ and
    for each "atom" $\alpha = R(x_{i_1}, \ldots, x_{i_p})$ occurring in $\fml$,
    $P_\alpha^{\struc[2]} \defeq \set{\tuple{a_{j_1}, \dots, a_{j_\ell}} \in \univ{\struc[2]}^\ell
    \mid \tuple{a_{i_1},\dots,a_{i_p}} \in R^{\struc}}$.
\end{itemize}
Clearly, $\fml[2]$ is an $\kFO{k}$-"sentence" and $\len{\fml[2]} = \asympO(\len{\fml})$.
Also, $\strucn{\struc[2]} = \strucn{\struc}$
and $\strucm{\struc[2]}
= \max(1, \sum_{\alpha \in \Atoms(\fml)} \card{P_\alpha^{\struc[2]}})
\le \max(1, \sum_{\alpha \in \Atoms(\fml)} \card{R^{\struc}})
\le \max(1, \card{\Atoms(\fml)} \cdot \strucm{\struc})
\le \asympO(\len{\fml} \cdot \strucm{\struc})$.
Constructing $\fml[2]$ from $\fml$ takes $\asympO(\len{\fml})$ time by a single scan and replacement.
To construct $\struc[2]$,
for each "atom" $\alpha = R(x_{i_1}, \ldots, x_{i_p}) \in \Atoms(\fml)$,
we scan all tuples in $R^{\struc}$ once.
For each tuple $\tuple{b_1,\dots,b_p}$,
we check that $b_s=b_t$ whenever $i_s=i_t$.
If this condition holds,
we insert the corresponding $\ell$-tuple of the values of $x_{j_1}, \dots, x_{j_\ell}$ (where the value of $x_{i_s}$ is $b_s$) into $P_\alpha^{\struc[2]}$.
For example, if $\alpha=R(x_2,x_2,x_1)$, then a tuple $\tuple{b_1,b_2,b_3}$ passes the check "iff" $b_1=b_2$, and we insert $\tuple{b_3,b_1}$.
Each tuple is processed in time linear in the "arity" of $R$ in the "word-RAM model", and the "arities" of all "atoms" in $\fml$ sum to at most $\len{\fml}$.
Hence, the total running time is $\asympO(\len{\fml} \cdot \strucm{\struc} + \strucn{\struc})$.
For correctness, for all "valuations" $I$ (note that $\univ{\struc} = \univ{\struc[2]}$),
we have $\struc, I \models \alpha$ "iff" $\struc[2], I \models P_{\alpha}(x_{j_1}, \dots, x_{j_\ell})$.
Thus, by induction on "subformulas" of $\fml$,
we have $\struc \models \fml$ "iff" $\struc[2] \models \fml[2]$.
\end{proof}

\subsubsection{"Guarded-pairing" (reducing the maximum arity to $2$)}\label{section: sparse fo3}\label{section: guarded-pairing}
Here, we reduce the maximum arity to at most $2$.
\ifthenelse{\boolean{arxiv}}{}{See the extended version~\cite{nakamuraUezatoDichotomyBoundedVariable2026arxiv} for the full proof.}
\ifthenelse{\boolean{arxiv}}{%
\begin{lemma}[\Cref{section: lemma: guarded-pairing}]%
}{%
\begin{lemma}%
}\label{lemma: guarded-pairing}\label{lemma:sparse3:reduction higher order}
\gdef\guardedpairinglemma{%
Let $\fml$ be an $\kFO{3}\sigres{(C_1, C_2, C_3)}$-"sentence" and $\struc$ be a "structure".
One can construct, in $\asympO(\strucm{\struc} + \len{\fml})$ time, an $\kFO{3}\sigres{(0, 1 + \sum_{i = 1}^{3} i C_i)}$-"sentence" $\fml[2]$ and a "structure" $\struc[2]$
such that $\struc \models \fml$ "iff" $\struc[2] \models \fml[2]$.
Moreover, $\strucn{\struc[2]} = \asympO(\strucm{\struc})$, $\strucm{\struc[2]} = \asympO(\strucm{\struc})$, and $\len{\fml[2]} = \asympO(\len{\fml})$.
}
\guardedpairinglemma
\end{lemma}
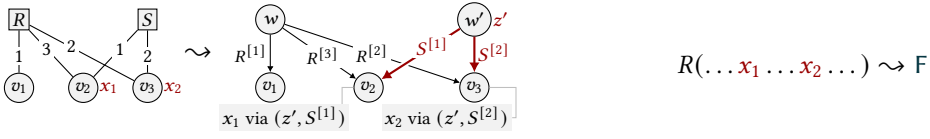
\begin{figure}[t]
\captionsetup[subfigure]{justification=raggedright}
\begin{subfigure}[t]{0.53\textwidth}
\parbox[b][1.85cm][c]{\linewidth}{\centering
\begin{tikzpicture}[baseline = -3.5ex]
    \tikzstyle{vert}=[draw = black, circle, fill = gray!10, inner sep = .5pt, minimum size = 11pt, font = \scriptsize]
    \tikzstyle{myedge}=[draw = black, rectangle, fill = gray!10, inner sep = 1.5pt, minimum size = 5pt, font = \scriptsize]
    \graph[grow right = 0.85cm, branch down = 6ex]{ 
    {s/{$R$}[myedge], s1/{$v_1$}[vert, inner sep = 1pt]} -!-
    {/, s2/{$v_2$}[vert, inner sep = 1pt]} -!-
    {v/{$S$}[myedge], s3/{$v_3$}[vert, inner sep = 1pt]}
    };
    \path (s) edge[earrow] node[fill=white, inner sep=1pt,font= \scriptsize, pos = .6]{$1$} (s1);
    \path (s) edge[earrow] node[fill=white, inner sep=1pt,font= \scriptsize, pos = .4]{$3$} (s2);
    \path (s) edge[earrow] node[fill=white, inner sep=1pt,font= \scriptsize, pos = .4, yshift = 0.1ex]{$2$} (s3);
    \path (v) edge[earrow] node[fill=white, inner sep=1pt,font= \scriptsize, pos = .4]{$1$} (s2);
    \path (v) edge[earrow] node[fill=white, inner sep=1pt,font= \scriptsize, pos = .6]{$2$} (s3);
    \node[right = -0.3em of s2, color = red!60!black]{\scriptsize $x_1$};
    \node[right = -0.3em of s3, color = red!60!black]{\scriptsize $x_2$};
\end{tikzpicture}
\hspace{-.5em}$\leadsto$
\begin{tikzpicture}[baseline = -3.5ex]
    \tikzstyle{vert}=[draw = black, circle, fill = gray!10, inner sep = 1pt, minimum size = 11pt, font = \scriptsize]
    \node[vert] (s) at (0,0) {$w$};
    \node[vert] (s1) at (0,-0.9) {$v_1$};
    \node[vert] (s2) at (1.3,-0.9) {$v_2$};
    \node[vert] (s3) at (2.7,-0.9) {$v_3$};
    \node[vert] (v) at (2.7,0) {$w'$};
    \node[right = -0.3em of v, font = \scriptsize, text = red!60!black]{$z'$};
    \path (s) edge[->, earrow] node[left, inner sep=1pt, font= \tiny, pos = .5]{$R^{[1]}$} (s1);
    \path (s) edge[->, earrow] node[fill=white, inner sep=1pt, font= \tiny, pos = .55]{$R^{[3]}$} (s2);
    \path (s) edge[->, earrow] node[fill=white, inner sep=1pt, font= \tiny, pos = .5]{$R^{[2]}$} (s3);
    \path (v) edge[-{Latex[length=1.3mm, width=1.4mm]}, draw = red!60!black, line width = .8pt, preaction={draw=white, line width=2.4pt}] node[fill=white, inner sep=1pt, font= \scriptsize, pos = .35, text = red!60!black]{$S^{[1]}$} (s2);
    \path (v) edge[-{Latex[length=1.3mm, width=1.4mm]}, draw = red!60!black, line width = .8pt] node[right, inner sep=1pt, font= \scriptsize, pos = .45, text = red!60!black]{$S^{[2]}$} (s3);
    \node[anchor = east, font = \scriptsize, inner sep = 1pt, fill = gray!10] (n1) at ($(s2.center)+(-0.25,-0.4)$) {$x_1$ via $\tuple{z', S^{[1]}}$};
    \node[anchor = east, font = \scriptsize, inner sep = 1pt, fill = gray!10] (n2) at ($(s3.center)+(0.5,-0.4)$) {$x_2$ via $\tuple{z', S^{[2]}}$};
    \draw[gray!50] (s2.180) -| ($(n1.north east)+(-0.1,-0.1)$);
    \draw[gray!50] (s3.0) -- ($(s3.0)+(0.35,0)$) |- (n2.east);
    \node[font = \scriptsize, inner sep = 1pt, fill = gray!10] at (n1.center) {$x_1$ via $\tuple{z', S^{[1]}}$};
    \node[font = \scriptsize, inner sep = 1pt, fill = gray!10] at (n2.center) {$x_2$ via $\tuple{z', S^{[2]}}$};
\end{tikzpicture}}
\caption{$\tuple{x_1, x_2}$ guarded.
The pair is encoded by the "vertex" $w'$ of a tuple containing $I(x_1), I(x_2)$ (here, of $S$) via $S^{[1]}, S^{[2]}$.}
\end{subfigure}%
\hfill
\begin{subfigure}[t]{0.44\textwidth}
\parbox[b][1.85cm][c]{\linewidth}{\centering
{\normalsize $R(\dots \textcolor{red!60!black}{x_1} \dots \textcolor{red!60!black}{x_2} \dots) \leadsto \fmlfalse$}}
\caption{$\tuple{x_1, x_2}$ unguarded.
No tuple contains both $I(x_1)$ and $I(x_2)$; "atoms" with $x_1, x_2$ become $\fmlfalse$.}
\end{subfigure}
 \caption{Idea of the "guarded-pairing".}
\label{figure: guarded-pairing}
\Description{Idea of the guarded-pairing: a pair of variables whose values occur together in a tuple is encoded by the vertex of that tuple, while atoms containing an unguarded pair are replaced by false.}
\end{figure}
Below we present the key idea of \Cref{lemma: guarded-pairing}. %
As in \Cref{figure: guarded-pairing},
a standard binary encoding of higher-arity "relations" puts a new "vertex" (denoted by $w$ and $w'$) for each tuple of each "relation" (e.g., $R$ and $S$)
and adds fresh binary "relation symbols" $R^{[i]}$ to indicate the $i$-th element of the tuple.
In this encoding, for instance, a "formula" $R(x_3, x_2, x_1)$ is translated straightforwardly into the following "formula":
$\exists z.~~ R^{[1]}(z, x_3) \land R^{[2]}(z, x_2) \land R^{[3]}(z, x_1)$.
However, this translation may increase the number of "variables", as an additional "variable" $z$ distinct from $x_1, x_2, x_3$ is introduced.
To avoid this problem, we consider ""guarded-pairing"".\footnote{The notations in the reduction, such as $R^{[i]}$, are inspired by \cite[\S 3]{figueiraGuardedNegationTransitive2026}.}
For a pair $\tuple{x_1, x_2}$ under a "valuation" $I$,
we distinguish two cases:
\begin{itemize}
    \item Case $\tuple{x_1, x_2}$ is \emph{"guarded"}:\ifcameraready\AP\phantomintro{guarded}\fi{}
    Let us consider the case when $I(x_1)$ and $I(x_2)$ occur together in a tuple of some "relation" $S$.
We can then encode the pair as a single "variable" under $I$.
For instance, in \Cref{figure: guarded-pairing}, when $I$ maps $\tuple{x_1, x_2}$ to $\tuple{v_2, v_3}$,
we introduce a "variable"%
\ifthenelse{\boolean{arxiv}}{\footnote{For simplicity, $z'$ is fixed to $x_1$ in the translation given in the proof of \Cref{lemma: guarded-pairing} (\Cref{section: lemma: guarded-pairing}).}}{} $z'$ mapping to $w'$.
This single "variable" $z'$ can encode both $x_1$ and $x_2$ without explicitly using $x_1$ and $x_2$ in the "formula",
by treating $x_1$ and $x_2$ as the unique "elements" such that $S^{[1]}(z', x_1)$ and $S^{[2]}(z', x_2)$ hold, respectively.
Under the mapping above, for instance, we can translate $R(x_3, x_2, x_1)$ into the following "formula" within $3$ "variables" (we can replace $z''$ with $x_3$, so the "formula" uses only $3$ "variables" indeed):
\[\exists z.~~
R^{[1]}(z, x_3) \land
(\exists z''.~ S^{[2]}(z', z'') \land R^{[2]}(z, z'')) \land
(\exists z''.~ S^{[1]}(z', z'') \land R^{[3]}(z, z'')).\]

\item Otherwise, "ie", when $x_1$ and $x_2$ are not "guarded":
We can then replace each "atom" that contains both $x_1$ and $x_2$ with the false "sentence" $\fmlfalse$.
\end{itemize}
In either case, we can translate the "formula" into an equivalent "formula" over the encoded "structure" with only binary "relations", without increasing the number of "variables".

\subsubsection{Combining the reductions}
Finally, combining "projection" (\Cref{lemma: projection}) with "guarded-pairing" (\Cref{lemma: guarded-pairing})
reduces arbitrary sparse $\MC{\kFO{3}}$ to the binary-"signature" case.
\begin{proof}[Proof of \Cref{theorem: fo3} in the sparse case]
(Below we use $\fml[2]$ and $\struc[2]$ to denote the resulting "formula" and "structure" after the reduction, respectively.)
By \Cref{theorem: fo3} in the dense case with $n \le \asympO(m)$ (\Cref{proposition: isolated})
(building the adjacency matrices takes $\asympO(\len{\fml} \cdot n^2 + m)$ time),
$\MC{\kFO{3}\sigres{(\ordinalomega, \ordinalomega)}}$ is in $2^{\asympO(\len{\fml})} \cdot m^{\matrixexponent}$ time.
By \Cref{lemma:sparse3:reduction higher order}
with $\strucn{\struc[2]}, \strucm{\struc[2]} = \asympO(\strucm{\struc[1]})$ and $\len{\fml[2]} = \asympO(\len{\fml[1]})$,
$\MC{\kFO{3}\sigres{(\ordinalomega, \ordinalomega, \ordinalomega)}}$ is also solvable in $2^{\asympO(\len{\fml})} \cdot m^{\matrixexponent}$ time.
After removing isolated "elements" so that $n \le \asympO(\len{\fml} \cdot m)$ (\Cref{proposition: isolated}),
by \Cref{lemma:reduction higher order predicates}
with $\strucm{\struc[2]} = \asympO(\len{\fml} \strucm{\struc[1]})$ and $\len{\fml[2]} = \asympO(\len{\fml[1]})$,
$\MC{\kFO{3}}$ is also solvable in time $2^{\asympO(\len{\fml})} \cdot \asympO(\len{\fml} m)^{\matrixexponent} \le 2^{\asympO(\len{\fml})} \cdot m^{\matrixexponent}$.
\end{proof}
\section{\texorpdfstring{$\kFO{2}$}{FO2} Model Checking via Boolean Modal Logic \texorpdfstring{$\BMLP$}{ML*}}\label{section: fo2}
In this section, we show \Cref{theorem: fo2}.
We introduce a Boolean modal logic $\BMLP$, which has a simple and efficient model-checking algorithm running in $\asympO(\len{\fml} \cdot m)$ time.

\AP
The syntax of $\intro*\BMLPwithC$ ($\BMLP$ with term complement) is defined as follows,
where $R$ ranges over binary "relation symbols",
$U$ ranges over unary "relation symbols", and
$x, y$ range over "variables":\footnote{We use the notation of \Cref{section: fo3}, rather than \cite{lutzModalLogicTwoVariable2001}, for consistency.}
\phantomintro(BMLP){term}%
\phantomintro(BMLP){formula}%
\phantomintro(BMLP){atomic FO-formula}%
\begin{align*}
    \term[1], \term[2] ~\Coloneqq~
    & R \mid \corid \mid \term[1]^{\smile} \mid \term[1] \intersection \term[2] \mid \term[1] \union \term[2] ~\mid~ \term[1] \setminus \term[2] \mid \term^{-}, \tag*{($\BMLPwithC$-"terms@@BMLP")}\\
    \mlfml[1], \mlfml[2] ~\Coloneqq~
    & U \mid \fmltrue \mid \mlfml[1] \land \mlfml[2] \mid \mlfml[1] \lor \mlfml[2] \mid \lnot \mlfml[1] \mid \bo{\term} \mlfml \mid \windowbo{\term} \mlfml,
    \tag*{($\BMLPwithC$-"formulas@@BMLP")}\\
    \fml[1], \fml[2] ~\Coloneqq~
    & %
    \mlfml[1](x).
    \tag*{($\BMLPwithC$ "atomic FO-formulas@@BMLP")}
\end{align*}
The "formula@@BMLP" $\intro*\bo{\term}\mlfml$ denotes the standard ""box"" "formula@@BMLP".
The "formula@@BMLP" $\intro*\windowbo{\term}\mlfml$ denotes the ""window"" "formula@@BMLP" \cite{gargovNoteBooleanModal1990},
which is alternatively defined as $\windowbo{\term} \mlfml \defeq \bo{\term^{-}} \lnot \mlfml$.
We write $\fmlfalse$ for $\lnot \fmltrue$.

\AP
An \intro*$\BMLP$ "term@@BMLP" ("resp", "formula@@BMLP", "atomic FO-formula@@BMLP")
is an $\BMLPwithC$ "term@@BMLP" ("resp", "formula@@BMLP", "atomic FO-formula@@BMLP")
such that the term complement ($\cdot^{-}$) does not occur (but the difference ($\setminus$) may occur).

Given a "structure" $\struc$,
the ""semantics@@BMLP"" $\sem{\mlfml}^{\struc} \subseteq \univ{\struc}$ and $\sem{\term}^{\struc} \subseteq \univ{\struc}^2$
are defined as follows:

\resizebox{.95\textwidth}{!}{
\begin{minipage}{\textwidth}
\begin{align*}
\sem{R}^{\struc} & \defeq R^{\struc},&\hspace{-.5em}
\sem{\term[1] \intersection \term[2]}^{\struc} & \defeq \sem{\term[1]}^{\struc} \cap \sem{\term[2]}^{\struc},&\hspace{-1.em}
\sem{\term[1]^{\smile}}^{\struc} & \defeq \set{ \tuple{x, y} \mid \tuple{y, x} \in \sem{\term[1]}^{\struc} }, \hspace{.5em}
\sem{\term[1]^{-}}^{\struc} \defeq \univ{\struc}^2 \setminus \sem{\term[1]}^{\struc},  \span\span\\
\sem{U}^{\struc} & \defeq U^{\struc},\ &\hspace{-.5em}
\sem{\term[1] \union \term[2]}^{\struc} & \defeq \sem{\term[1]}^{\struc} \cup \sem{\term[2]}^{\struc},&\hspace{-1.em}
\sem{\corid}^{\struc} & \defeq \{\tuple{x, x} \mid x \in \univ{\struc}\},&
\sem{\term[1] \setminus \term[2]}^{\struc} & \defeq \sem{\term[1]}^{\struc} \setminus \sem{\term[2]}^{\struc},\\
\sem{\fmltrue}^{\struc} &\defeq \univ{\struc}, &\hspace{-.5em}
\sem{\mlfml[1] \land \mlfml[2]}^{\struc} & \defeq \sem{\mlfml[1]}^{\struc} \cap \sem{\mlfml[2]}^{\struc},\ &\hspace{-1.em}
\sem{\mlfml[1]  \lor \mlfml[2]}^{\struc} & \defeq \sem{\mlfml[1]}^{\struc} \cup \sem{\mlfml[2]}^{\struc},&
\sem{\lnot \mlfml[1]}^{\struc} & \defeq \univ{\struc} \setminus \sem{\mlfml[1]}^{\struc},\\
   \sem{\bo{\term} \mlfml[1]}^{\struc} & \defeq \set{ x \in \univ{\struc} \mid \forall y \in \univ{\struc}.~ \tuple{x, y} \in \sem{\term}^{\struc} \text{ implies } y \in \sem{\mlfml[1]}^{\struc} },\span\span\span\span\span\span\\
   \sem{\windowbo{\term} \mlfml[1]}^{\struc} & \defeq \set{ x \in \univ{\struc} \mid \forall y \in \univ{\struc}.~ y \in \sem{\mlfml[1]}^{\struc} \text{ implies } \tuple{x, y} \in \sem{\term}^{\struc} }.\span\span\span\span\span\span
\end{align*}
\end{minipage}
}

For $\BMLPwithC$ "atomic FO-formulas@@BMLP" $\mlfml(x)$,
we define $\struc, I \models \mlfml(x)$ if and only if $I(x) \in \sem{\mlfml}^{\struc}$.

Our $\BMLPwithC$ is a syntactic extension of 
Lutz--Sattler--Wolter's $\MLLSW$ \cite{lutzModalLogicTwoVariable2001}.
Here, the extra "window" operator and the difference operator ($\setminus$) are introduced only to remove term complements later.
We use their inductive translation from
$\kFO{2}$-"formulas" with at most one free "variable",
where we store the resulting modal "formula" as a DAG by memoizing the pair of a "subformula" and its free "variable".
From this, the translated $\BMLPwithC$-"formula@@BMLP" has size $2^{\asympO(\len{\fml})}$.%
\footnote{
See
\ifthenelse{\boolean{arxiv}}{\Cref{section: proposition: fo2 to bml+c} for a toy example and}{}
\cite[Theorem 1]{lutzModalLogicTwoVariable2001} for the full translation.
As a "tree-represented" "formula@@BMLP", its "size" is in $2^{\asympO(\len{\fml}^2)}$ if the nesting depth of quantifiers is not fixed.}
We say that two "formulas" ($\FO$-"formulas" or $\BMLPwithC$ "atomic FO-formulas@@BMLP") $\fml[1]$ and $\fml[2]$ are \reintro*\kl{semantically equivalent} if $\struc, I \models \fml[1]$ and $\struc, I \models \fml[2]$ have the same "truth value" for every "structure" $\struc$ and "valuation" $I$.%
\begin{proposition}[{Cor.~of~\cite[Theorem 1]{lutzModalLogicTwoVariable2001}}]
\label{proposition: fo2 to bml+c}
Let $\fml$ be a ("DAG-represented")
$\kFO{2}\sigres{(\ordinalomega, \ordinalomega)}$-"formula" with at most one free "variable" $x$.
One can construct, in $2^{\asympO(\len{\fml})}$ time,
a "DAG-represented" $\BMLPwithC$-"formula@@BMLP" $\mlfml$,
such that $\fml$ and $\mlfml(x)$ are "semantically equivalent".
\end{proposition}
Here, when the input is a "sentence", we use a dummy "variable" in the resulting $\BMLPwithC$-"formula@@BMLP".

We next eliminate the term complement ($\cdot^{-}$) from $\BMLPwithC$.
First, every "term@@BMLP" is rewritten as either a complement-free "term@@BMLP" $\term$ or the complement $\term^{-}$ of such a "term@@BMLP",
using the difference operator $\setminus$.
The remaining outermost term complement is then eliminated at the modal level using the "window" modality.
For instance, we can normalize the "formula@@BMLP" $\bo{R \union S^{-}} U$ as follows.
\begin{align*}
\bo{R \union S^{-}} U
&~\leadsto_{\text{(De Morgan's law)}}~ \bo{(R^{-} \intersection S)^{-}} U
~\leadsto_{\text{using $\setminus$} }~ \bo{(S \setminus R)^{-}} U
~\leadsto_{\text{using $\windowbo{-}$} }~ \windowbo{S \setminus R} \lnot U.
\end{align*}
Based on the example above, we can give a linear-time translation from $\BMLPwithC$ to $\BMLP$.
\ifthenelse{\boolean{arxiv}}{}{See the extended version~\cite{nakamuraUezatoDichotomyBoundedVariable2026arxiv} for the proof.}
\ifthenelse{\boolean{arxiv}}{%
\begin{proposition}[\Cref{section: proposition: bml+c to bml+}]%
}{%
\begin{proposition}%
}\label{proposition: bml+c to bml+}
\gdef\propositionbmlctobml{%
Let $\mlfml$ be an $\BMLPwithC$-"formula@@BMLP".
One can construct, in $\asympO(\len{\mlfml})$ time, a "semantically equivalent" $\BMLP$-"formula@@BMLP".
}\propositionbmlctobml
\end{proposition}

\AP
For $\fmlclass = \BMLP$ or $\BMLPwithC$,
the ""model-checking problem@@BMLP"" for $\fmlclass$, denoted by $\intro*\BMLPMC{\fmlclass}$, is the following problem:
given a "structure" $\struc$, a "valuation" $I$, and an $\fmlclass$ "atomic FO-formula@@BMLP" $\fml = \mlfml(x)$,
decide whether $\struc, I \models \fml$.\AP\phantomintro(BMLP){DAG-} %
$\BMLP$ has the following sparsity property.
\begin{proposition}\label{proposition: BMLP term sparsity}
Every $\BMLP$-"term@@BMLP" $\term$ satisfies $\sem{\term}^{\struc} \subseteq \sem{\corid}^{\struc} \cup \bigcup_{R \in \sig : \arity{R} = 2} \sem{R \union R^{\smile}}^{\struc}$.
\end{proposition}
\begin{proof}
By an easy induction on $\term$.
\end{proof}
Hence, $\BMLP$ has no "term@@BMLP" for the (dense) full "relation" $\cortop$; $\corz$ is $\corid \setminus \corid$.
This sparsity implies $\card{\sem{\term}^{\struc}} \le \asympO(\strucm{\struc} + \strucn{\struc}) \le \asympO(\strucm{\struc})$ by removing isolated "elements" ("cf" \Cref{proposition: isolated}; see \ifthenelse{\boolean{arxiv}}{\Cref{section: proposition: BMLP model checking}}{the extended version~\cite{nakamuraUezatoDichotomyBoundedVariable2026arxiv}}),
which shows the following algorithmic result.
\ifthenelse{\boolean{arxiv}}{}{See the extended version~\cite{nakamuraUezatoDichotomyBoundedVariable2026arxiv} for the proof.}
\ifthenelse{\boolean{arxiv}}{%
\begin{proposition}[\Cref{section: proposition: BMLP model checking}]%
}{%
\begin{proposition}%
}\label{proposition: BMLP model checking}
\gdef\propositionBMLPmodelchecking{%
$\BMLPMC{\BMLP}$ is solvable in $\asympO(\len{\fml} \cdot m)$ time.
}\propositionBMLPmodelchecking
\end{proposition}

Combining the "projection" (\Cref{lemma: projection}) with the translations above reduces $\MC{\kFO{2}}$ to $\BMLPMC{\BMLP}$:
\begin{theorem*}[Restatement of {\Cref{theorem: fo2}}]
\theoremFOTWOLIN
\end{theorem*}
\begin{proof}
(Below we use $\fml[2]$ and $\struc[2]$ to denote the formula and "structure" after the 
reduction, respectively.)
By \Cref{proposition: BMLP model checking},
$\BMLPMC{\BMLP}$ is in $\asympO(\len{\fml} \cdot m)$ time.
By \Cref{proposition: bml+c to bml+} with $\len{\fml[2]} = \asympO(\len{\fml[1]})$,
$\BMLPMC{\BMLPwithC}$ is also in $\asympO(\len{\fml} \cdot m)$ time.
By \Cref{proposition: fo2 to bml+c} with $\len{\fml[2]} = 2^{\asympO(\len{\fml})}$,
$\MC{\kFO{2}\sigres{(\ordinalomega, \ordinalomega)}}$ is in $2^{\asympO(\len{\fml})} \cdot m$ time (note that the input is a "sentence").
After removing isolated "elements" so that $n \le \asympO(\len{\fml} \cdot m)$ (\Cref{proposition: isolated}),
by \Cref{lemma:reduction higher order predicates} with $\strucm{\struc[2]} = \asympO(\len{\fml} \, \strucm{\struc[1]})$ and $\len{\fml[2]} = \asympO(\len{\fml[1]})$,
$\MC{\kFO{2}}$ is also solvable in time $2^{\asympO(\len{\fml})} \cdot m$.
This completes the proof.
\end{proof}

\begin{remark}\label{remark: BMLP to CoR}
After encoding each unary "relation symbol" as its diagonal binary "relation",
$\BMLP$ is strictly less expressive than $\CoR$:
$\BMLP$ is as expressive as $\kFO{2}$-"formulas" with at most one free "variable" (\Cref{proposition: fo2 to bml+c,proposition: bml+c to bml+} and the standard translation),
$\CoR$ is as expressive as $\kFO{3}$-"formulas" with at most two free "variables" \cite{tarskiFormalizationSetTheory1987},
and $\kFO{2}$ cannot express, "eg", the existence of a directed $3$-cycle\ifthenelse{\boolean{arxiv}}{ (\Cref{section: expressive power comparison})}{}.
See \ifthenelse{\boolean{arxiv}}{\Cref{section: BMLP to CoR}}{the extended version~\cite{nakamuraUezatoDichotomyBoundedVariable2026arxiv}} for a $\CoR$ variant of $\BMLP$ ("cf" $\CoRtwo$ in \cite[\S 3.2]{nakamuraExpressivePowerSuccinctness2022}).\ifcameraready\AP\phantomintro*\CoRtwo\fi
\end{remark}
\section{Further Consequences and Open Problems}\label{sec: conclusion}
We have established complexity dichotomies for
$\kPP{k}\sigres{\sig}$ / $\kFO{k}\sigres{\sig}$ "model checking" in the sparse setting;
the dense setting is also discussed in \Cref{sec: dichotomy dense}.
We conclude by discussing several consequences of our main dichotomy theorem, with some open problems.

\subsection{Connection to graph homomorphism}\label{sec: graph homomorphism reformulations}
The correspondence recalled in \Cref{section: bounded treewidth} turns the
"model-checking problem" for each fixed
\(\kPP{k}\sigres{(0,1)}\)-"sentence" into a fixed-pattern directed "graph
homomorphism problem" whose pattern has "treewidth" at most \(k-1\).
Thus, \Cref{theorem:sparse-k4} yields the following sparse lower bound for this problem.
\begin{corollary}\label{corollary: sparse hardness graph}
Fix \(k \ge 4\).
Assuming the \SparseMAXxSATHypothesis{(k-1)},
for any \(\eps > 0\), there exists a directed "graph" \(\graph[1]_{\eps}\) of
"treewidth" at most \(k-1\) such that the fixed-pattern
"graph homomorphism problem" \(\HOM(\graph[1]_{\eps}, -)\), over
directed target "graphs", cannot be solved in \(\asympO(m^{k-\eps})\) time,
where \(m\) denotes the number of edges of the target "graph".
\end{corollary}

\begin{remark}\label{remark: sparse hardness graph}
The same lower bound also applies to the fixed-pattern "graph homomorphism problem" restricted to undirected fixed patterns and undirected target "graphs",
by the standard edge-replacement construction ("cf" \cite[\S~4.4]{hellGraphsHomomorphisms2004}):
in both the fixed pattern and the target, replace every directed edge \((u,v)\) by a copy of the undirected "graph" \(G_0\) of \cite[Fig.~4.7]{hellGraphsHomomorphisms2004},
identifying its distinguished vertices \(a\) and \(b\) with \(u\) and \(v\).
This preserves the existence of "graph homomorphisms" and increases \(m\) only by a constant factor.
It also preserves "treewidth" at most \(k-1\) for \(k \ge 4\), since \(G_0\) has a tree decomposition of width \(3\) with a bag containing both \(a\) and \(b\).
Similarly, in the sparse setting and for \(k \ge 4\), the unary-encoding step in the proof
of \Cref{lemma: reduction monadic and binary predicates to single binary predicate} can be made undirected by using the undirected
"graphs" \(G_0,G_1,\dots\) of \cite[Fig.~4.7]{hellGraphsHomomorphisms2004}.
\end{remark}

\subsection{Dichotomy in the dense setting}\label{sec: dichotomy dense}
Combining \Cref{proposition: dense hardness} with the encoding into a single binary "relation" of \Cref{lemma: reduction monadic and binary predicates to single binary predicate}, we have:
\begin{corollary}\label{corollary: dense hardness}
Fix $k \ge 4$.
Assuming the $(qk, q)$-\HyperCliqueHypothesis{} for some $3 \le q < k$,
there is a fixed $\kPP{k}\sigres{(0, 1)}$-"sentence" $\fml$ such that for any $\eps > 0$,
$\MC{\set{\fml}}$ cannot be solved in $\asympO(n^{k - \eps})$ time.
\end{corollary}
\Cref{corollary: dense hardness} (with \Cref{theorem: fo3,theorem: monadic} and the unconditional lower bound for $k \le \sigarity{\sig}$\ifarxiv{} in \Cref{section: trivial lower bound}\fi{}) 
also yields a dichotomy in the dense setting analogous to \Cref{theorem: dichotomy sparse}
\ifthenelse{\boolean{arxiv}}{(\Cref{section: theorem: dichotomy dense})}{(see the extended version~\cite{nakamuraUezatoDichotomyBoundedVariable2026arxiv})}.
\Cref{corollary: dense hardness} also shows that the "graph homomorphism problem" is hard even without colors ("cf" \Cref{section: bounded treewidth} for the "colored subgraph isomorphism problem").
\begin{corollary}\label{corollary: dense hardness graph}
Fix $k \ge 4$.
Assuming the $(qk, q)$-\HyperCliqueHypothesis{} for some $3 \le q < k$,
there exists a directed "graph" $\graph[1]$ of "treewidth" at most $k-1$ such that 
the fixed-pattern "graph homomorphism problem" $\HOM(\graph[1], -)$ cannot be solved in $\asympO(n^{k - \eps})$ time for any $\eps > 0$, where $n$ denotes the number of "vertices" of the target "graph".
\end{corollary}

\subsection{On \texorpdfstring{$\eps$}{ε}-uniform fixed "sentence" hardness}\label{sec: fixed sentence}
The hardness results for $\kPP{k}$ and $\kPP{3}$ in \Cref{theorem:sparse-k4,theorem:sparse-k3} are non-uniform in $\eps$:
for each $\eps > 0$, the proof produces a fixed "sentence" $\fml_{\eps}$.
What remains open is the stronger uniform statement:
whether a single $\kPP{k}\sigres{(0, 1)}$-"sentence", independent of $\eps$,
rules out $m^{k-\eps}$ algorithms for any $\eps > 0$.
\begin{question}["Cf" \Cref{theorem:sparse-k4}]\label{problem: hardness by H-Clique fixed 4}
Let $k \ge 4$.
Assuming the \SparseMAXxSATHypothesis{(k-1)}, does there exist a $\kPP{k}\sigres{(0, 1)}$-"sentence" $\fml$ such that
$\MC{\set{\fml}}$ cannot be solved in $\asympO(m^{k - \eps})$ time for any $\eps > 0$?
Equivalently,
does there exist a directed "graph" $\graph[1]$ of "treewidth" at most $k-1$ for which the fixed-pattern "graph homomorphism problem" $\HOM(\graph[1], -)$ cannot be solved in $\asympO(m^{k - \eps})$ time for any $\eps > 0$?
\end{question}
\begin{question}["Cf" \Cref{theorem:sparse-k3}]\label{problem: hardness by H-Clique fixed 3}
Assuming the \SparseMAXxSATHypothesis{2}, does there exist a $\kPP{3}\sigres{(0, 1)}$-"sentence" $\fml$ such that
$\MC{\set{\fml}}$ cannot be solved in $\asympO(m^{\matrixexponent - \eps})$ time for any $\eps > 0$?
Equivalently,
does there exist a directed "graph" $\graph[1]$ of "treewidth" at most $2$ for which the fixed-pattern "graph homomorphism problem" $\HOM(\graph[1], -)$ cannot be solved in $\asympO(m^{\matrixexponent - \eps})$ time for any $\eps > 0$?
\end{question}
In the database context, queries with "self-joins" and ""self-join-free"" queries (where every "relation symbol" occurs at most once) are often treated separately; see, "eg", \cite{carmeliConjunctiveQueriesSelfJoins2023,bringmannTightFineGrainedBounds2025,mengelLowerBoundsConjunctive2025arxiv}.
For each fixed $\eps$, a "self-join-free" hard "sentence" is easily obtained from $\fml_{\eps}$ by replacing different occurrences of the single "relation symbol" with fresh "relation symbols" and interpreting them identically.
However, the resulting "signature" depends on $\eps$, because
the number of occurrences of "relation symbols" in the hard "sentences" $\fml_{\eps}$ is not uniformly bounded.
Thus, the corresponding "self-join-free" versions of \Cref{theorem:sparse-k4,theorem:sparse-k3} remain open over a fixed finite "signature" independent of $\eps$.

In contrast, if we admit disjunction and one universal quantifier, an $\eps$-independent hard "sentence" exists in the following positive fragment.
Here, $\intro*\kPFO{k}_{\exists^{k}\forall^{1}}$ consists of the \emph{positive} $\kFO{k}$-"formulas" ("ie", built from "atoms" using $\land$, $\lor$, $\exists$, and $\forall$ \cite{kozenPositiveFirstOrderLogic1981}) such that on every path from the root of the formula tree to an "atom", the sequence of "quantifiers" is a subsequence of $\exists^{k}\forall^{1}$.
This is shown by modifying the proofs of \Cref{theorem:sparse-k4,theorem:sparse-k3},
replacing the split-and-list technique with a single $\forall$-quantifier ranging over all "dimensions".
\ifthenelse{\boolean{arxiv}}{}{See the extended version~\cite{nakamuraUezatoDichotomyBoundedVariable2026arxiv} for the proofs.}
\ifthenelse{\boolean{arxiv}}{%
\begin{theorem}[\Cref{section: proof of hardness by H-Clique fo}]%
}{%
\begin{theorem}%
}\label{theorem:sparse-k4 fo}
\gdef\sparseFOhardnesstheoremFO{%
Fix $k \ge 4$.
Assuming the \SparseMAXxSATHypothesis{(k-1)},
there is a $\kPFO{k}_{\exists^{k}\forall^{1}}\sigres{(0, 1)}$-"sentence" $\fml$ such that
$\MC{\set{\fml}}$ cannot be solved in $\asympO(m^{k - \eps})$ time for any $\eps > 0$.
}
\sparseFOhardnesstheoremFO
\end{theorem}
\ifthenelse{\boolean{arxiv}}{%
\begin{theorem}[\Cref{section: proof of hardness by H-Clique fo}]%
}{%
\begin{theorem}%
}\label{theorem:sparse-k3 fo}
\gdef\sparsehardnessbyovtriangleFO{%
Assuming the \SparseMAXxSATHypothesis{2},
there is a $\kPFO{3}_{\exists^{3}\forall^{1}}\sigres{(0, 1)}$-"sentence" $\fml$ such that
$\MC{\set{\fml}}$ cannot be solved in $\asympO(m^{\matrixexponent - \eps})$ time for any $\eps > 0$.
}
\sparsehardnessbyovtriangleFO
\end{theorem}
Since the hard "sentences" are fixed,
\Cref{theorem:sparse-k4 fo,theorem:sparse-k3 fo} hold for $\kPFO{k}_{\exists^{k}\forall^{1}}\sigres{(0, \ordinalomega)}$-"sentences" where every "relation symbol" occurs at most once.%

\subsection{Hardness for \texorpdfstring{$\CoR$}{CoR}}\label{sec: target language hardness}
\emph{Primitive positive} $\CoR$ "terms@@CoR" \cite{nakamuraExpressivePowerSuccinctness2022} (the class $\intro*\PPCoR$)
are built from binary "relation symbols" and the constants $\corid$, $\cortop$, and $\corz$ (as primitives) using converse, intersection, and composition.
The $\kPP{3}$ lower bound of \Cref{theorem:sparse-k3} transfers to $\PPCoR$ "terms@@CoR"
(respectively, the $\kFO{3}$ lower bound of \Cref{theorem:sparse-k3 fo} transfers to $\CoR$ "terms@@CoR")
via the translation of \cite{nakamuraExpressivePowerSuccinctness2022}, yielding the following corollaries.
\ifarxiv \begin{corollary}[\Cref{section: CoR hardness}]\label{corollary: CoR fixed term hardness}
\else%
\begin{corollary}
\fi
Assuming the \SparseMAXxSATHypothesis{2},
for every $\eps>0$, there exists a fixed $\PPCoR$ "term@@CoR" $\term_{\eps}$ over a single binary "relation symbol" such that
deciding whether $\struc, I \models \term_{\eps}(x,y)$ cannot be done in $\asympO(m^{\matrixexponent-\eps})$ time.
The same holds for deciding whether $\sem{\term_{\eps}}^{\struc}$ is nonempty.
\end{corollary}
\ifarxiv \begin{corollary}[\Cref{section: CoR hardness}]\label{corollary: CoR uniform fixed term hardness}
\else%
\begin{corollary}%
\fi
Assuming the \SparseMAXxSATHypothesis{2},
there exists a fixed $\CoR$ "term@@CoR" $\term$ over a single binary "relation symbol" such that
deciding whether $\struc, I \models \term(x,y)$ cannot be done in $\asympO(m^{\matrixexponent-\eps})$ time for any $\eps>0$.
The same holds for deciding whether $\sem{\term}^{\struc}$ is nonempty.
\end{corollary}

\subsection{Polynomial formula-size dependence and succinctness} \label{section: succinctness}
The exponential dependence $2^{\asympO(\len{\fml})}$ in \Cref{theorem: fo3,theorem: fo2,theorem: monadic} (for the first part of \Cref{theorem: monadic}, when $k \ge 2$) cannot be replaced by a polynomial dependence, assuming the corresponding OV hypotheses\ifarxiv{} (recalled at the beginning of the appendices to \Cref{section: succinctness})\else{}~\cite{gaoCompletenessFirstorderProperties2018}\AP\phantomintro*\OVHypothesis{}\phantomintro{SETH}\fi{}, which follow from "SETH".
\ifthenelse{\boolean{arxiv}}{}{See the extended version~\cite{nakamuraUezatoDichotomyBoundedVariable2026arxiv} for the proofs.}
\ifthenelse{\boolean{arxiv}}{%
\begin{theorem}[\Cref{section: theorem: FO model checking hardness}]%
}{%
\begin{theorem}%
}
\label{theorem: FO model checking hardness}
\label{theorem: monadic FO model checking hardness}
    \gdef\theoremFOmodelcheckinghardness{%
    Fix $k \ge 2$.
    Assuming the $k$-\OVHypothesis{},
    $\MC{\PNFkEP{k}\sigres{(\ordinalomega)}}$
    cannot be solved in $\poly(\len{\fml}) \cdot m^{k - \eps}$ time
    for any $\eps > 0$.
    }%
    \theoremFOmodelcheckinghardness
\end{theorem}
\ifthenelse{\boolean{arxiv}}{%
\begin{theorem}[\Cref{section: theorem: FO model checking hardness one binary}]%
}{%
\begin{theorem}%
}\label{theorem: FO model checking hardness one binary}
    \gdef\theoremFOmodelcheckinghardnessonebinary{%
    Fix $k \ge 3$.
    Assuming the $k$-\OVHypothesis{},
    $\MC{\kEP{k}\sigres{(0, 1)}}$
    cannot be solved in $\poly(\len{\fml}) \cdot m^{k - \eps}$ time
    for any $\eps > 0$.
    }%
    \theoremFOmodelcheckinghardnessonebinary
\end{theorem}

Indeed, a polynomial-time translation to the corresponding target language,
combined with the polynomial model-checking algorithm for that target language,
would give a $\poly(\len{\fml})\cdot m^{k-\eps}$ algorithm for the source fragment,
contradicting \Cref{theorem: FO model checking hardness,theorem: FO model checking hardness one binary}.
Thus, we obtain:
\begin{corollary}\label{corollary: succinctness}
Assuming the $3$-\OVHypothesis{}, there is no polynomial-time translation from $\kFO{3}\sigres{(0,\ordinalomega)}$-"formulas" with at most two free "variables" to $\CoR$.
\end{corollary}
The same argument shows that there is no polynomial-time translation from $\kFO{2}\sigres{(\ordinalomega,\ordinalomega)}$ to $\BMLP$ under the $2$-\OVHypothesis{},
and from $\kFO{k}\sigres{(\ordinalomega)}$ ($k \ge 2$) to $\kC{1}$ under the $k$-\OVHypothesis{};
however, both are known to hold even unconditionally (\ifthenelse{\boolean{arxiv}}{see \Cref{section: succinctness unconditional}}{see the extended version~\cite{nakamuraUezatoDichotomyBoundedVariable2026arxiv}}).

The following related problem, posed, "eg", by Fletcher et al.~\cite[Remark~4.5]{fletcherRelativeExpressivePower2015} and Nakamura~\cite[Section~8]{nakamuraExpressivePowerSuccinctness2022},
remains open, but \Cref{corollary: succinctness} may be viewed as supporting evidence for a negative answer.
An unconditional negative answer is currently known for \emph{existential} $\CoR$ \cite{nakamuraExistentialCalculiRelations2023}, \cite[Theorem~6.1]{nakamuraExpressivePowerSuccinctness2022}, the fragment of $\CoR$ in which complement is restricted to the "atomic@@CoR" level.
\begin{question}[{\cite{fletcherRelativeExpressivePower2015,nakamuraExpressivePowerSuccinctness2022}}]\label{problem: succinctness gap between FO3 and CoR}
\gdef\problemsuccinctnessgapbetweenFOthreeandCoR{%
Is there a polynomial-\emph{size} translation from $\kFO{3}\sigres{(0,\ordinalomega)}$-"formulas" with at most two free "variables" to $\CoR$ "terms@@CoR"?
}
\problemsuccinctnessgapbetweenFOthreeandCoR
\end{question}

\begin{acks}
We thank the anonymous reviewers for their valuable comments.
Yoshiki Nakamura thanks Nobutaka Shimizu for suggesting a paper on fine-grained complexity theory in connection with \cite{nakamuraExpressivePowerSuccinctness2022}.
This work was supported by JSPS KAKENHI Grant Number JP25K14985.
Yuya Uezato was supported by JST, CREST Grant Number JPMJCR21M3.
\end{acks}

\bibliographystyle{ACM-Reference-Format}
\bibliography{main}
\vphantom{\cite{mengelLowerBoundsConjunctive2025}\cite{bringmannStructuralInvestigationApproximability2022}\cite{bringmannCurrentAlgorithmsDetecting2021}}

\clearpage
\appendix
\crefalias{section}{appendix}
\crefalias{subsection}{appendix}

\section{Supplement to Theorem~\ref{theorem: dichotomy sparse}}
In this section, we give some facts supplementing \Cref{theorem: dichotomy sparse}.

\subsection{Increasing "arity"}\label{section: proposition: arity monotone}
In "model checking", clearly,
we can increase the "arities" of "relation symbols" in the "signature" without changing the complexity of the problem, as follows.
By the following proposition, \Cref{theorem:sparse-k4} immediately implies that, for $k \ge 4$ and when $\sigarity{\sig} \ge 2$, assuming the \SparseMAXxSATHypothesis{(k-1)}, for any $\eps > 0$ there is a $\kPP{k}\sigres{\sig}$-"sentence" $\fml_{\eps}$ such that $\MC{\set{\fml_{\eps}}}$ cannot be solved in $\asympO(m^{k-\eps})$ time.

\begin{proposition}\label{proposition: arity monotone}
Let $\fmlclass \in \set{\kPP{k}, \kEP{k}, \kFO{k} \mid k \ge 1}$.
Let $\fml[1]$ be an $\fmlclass\sigres{(C_1, \dots, C_r)}$-"sentence" and $\struc[1]$ be a "structure".
One can construct, in $\asympO(r (\len{\fml[1]} + \strucm{\struc[1]}))$ time, an $\fmlclass\sigres{(\underbracket{0, \dots, 0}_{\text{$r-1$ times}}, \sum_{i=1}^r C_i)}$-"sentence" $\fml[2]$ and a "structure" $\struc[2]$
such that
$\struc[1] \models \fml[1]$ "iff" $\struc[2] \models \fml[2]$.
Moreover, $\len{\fml[2]} = \asympO(r \len{\fml[1]})$,
$\strucn{\struc[2]} = \strucn{\struc[1]}$, and
$\strucm{\struc[2]} = \strucm{\struc[1]}$.
Additionally, if $\fml[1]$ is "tree-represented", then so is $\fml[2]$.
\end{proposition}
\begin{proof}
    Let $\fml[2]$ be obtained from the "sentence" $\fml[1]$ by replacing each "atom" $R(x_1, \dots, x_i)$ with the "atom" $S(x_1, \dots, x_{i-1}, x_i, \dots, x_i)$, where $S$ is a fresh $r$-ary "relation symbol" prepared for each $R$.
    Let $\struc[2]$ be obtained from $\struc[1]$ by replacing each $R^{\struc[1]}$ with $S^{\struc[2]} \defeq \set{ \tuple{a_1, \dots, a_{i-1}, a_i, \dots, a_i} \mid \tuple{a_1, \dots, a_i} \in R^{\struc[1]}}$.
    This gives the desired construction.
\end{proof}

\subsection{Removing "equality"}\label{section: removing equality}
Although the logics in this paper include "equality" ($=$) by default,
we can remove it from the "signature" without changing the complexity of the problem, as follows.

\begin{proposition}\label{proposition: equality elimination}
Let $\fmlclass \in \set{\kPP{k}, \kEP{k}, \kFO{k} \mid k \ge 1}$ be a class of "formulas".
Let $\fml[1]$ be a ["tree-represented"] $\fmlclass\sigres{(C_1, C_2, C_3, \dots, C_r)}$-"sentence" and $\struc[1]$ be a "structure".
One can construct, in $\asympO(\len{\fml[1]} + \strucn{\struc[1]} + \strucm{\struc[1]})$ time,
a ["tree-represented"] $\fmlclass\sigres{(C_1, C_2 + 1, C_3, \dots, C_r)}$-"sentence" $\fml[2]$ without "equality" ($=$) and a "structure" $\struc[2]$
such that $\struc[1] \models \fml[1]$ "iff" $\struc[2] \models \fml[2]$.
Moreover, $\len{\fml[2]} = \asympO(\len{\fml[1]})$,
$\strucn{\struc[2]} = \strucn{\struc[1]}$, and $\strucm{\struc[2]} = \strucm{\struc[1]} + \strucn{\struc[1]}$.
\end{proposition}
\begin{proof}
Let $\fml[2]$ be obtained from $\fml[1]$ by replacing $x = y$ with $E(x, y)$,
where $E$ is a fresh binary "relation symbol".
Let $\struc[2]$ be obtained from $\struc[1]$ by setting $E^{\struc[2]} \defeq \set{ \tuple{a,a} \mid a \in \univ{\struc[1]}}$.
This gives the desired construction.
\end{proof}

\section{Proof of \texorpdfstring{\Cref{lemma: reduction monadic and binary predicates to single binary predicate}}{Lemma~\ref{lemma: reduction monadic and binary predicates to single binary predicate}}:
\texorpdfstring{$\kPP{k}\sigres{(\ordinalomega, \ordinalomega)}$}{PPk[(∞, ∞)]} to \texorpdfstring{$\kPP{k}\sigres{(0, 1)}$}{PPk[(0, 1)]}}\label{section: lemma: reduction monadic and binary predicates to single binary predicate}
In this section, we give a reduction from $\kPP{k}\sigres{(\ordinalomega, \ordinalomega)}$ to $\kPP{k}\sigres{(0, 1)}$, as follows:
\begin{align*}
\kPP{k}\sigres{(\ordinalomega, \ordinalomega)}
~\leadsto_{\text{\Cref{proposition: arity monotone}}}~ \kPP{k}\sigres{(0, \ordinalomega)}
&~\leadsto_{\text{\Cref{proposition: equality elimination}}}~ \kPP{k}\sigres{(0, \ordinalomega)} \text{ without "equality"} \\
&~\leadsto_{\text{\Cref{lemma: unary encoding}}}~ \kPP{k}\sigres{(0, 1)} \text{ without "equality"},
\end{align*}
where $k \ge 3$.
The reduction works for $\kEP{k}$ and $\kFO{k}$ as well.

In \Cref{proposition: arity monotone,proposition: equality elimination},
we encode unary "relations" by diagonal binary "relations" and eliminate "equality" by introducing a fresh binary "relation".
In \Cref{lemma: unary encoding},
we encode multiple "relations" by a single "relation symbol" via a unary encoding.

\begin{lemma}\label{lemma: unary encoding}
Let $\fmlclass \in \set{\kPP{k}, \kEP{k}, \kFO{k} \mid k \ge 3}$.
Let $\fml[1]$ be an $\fmlclass\sigres{(0, \ordinalomega)}$-"sentence" without "equality", let $C$ be the number of "relation symbols" occurring in $\fml[1]$, and let $\struc[1]$ be a "structure".
We can construct, in $\asympO_k(\len{\fml[1]} + C \, \strucn{\struc[1]} + \strucm{\struc[1]})$ time, an $\fmlclass\sigres{(0, 1)}$-"sentence" $\fml[2]$ without "equality" and a "structure" $\struc[2]$
such that
$\struc[1] \models \fml[1]$ "iff" $\struc[2] \models \fml[2]$.
Moreover, $\len{\fml[2]} = \asympO_k(\len{\fml[1]})$,
$\strucn{\struc[2]} = \asympO(C \, \strucn{\struc[1]})$,
and $\strucm{\struc[2]} = \asympO(C \strucn{\struc[1]} + \strucm{\struc[1]})$.
Additionally, if $\fml[1]$ is "tree-represented",
then the "tree-represented" $\fml[2]$ is constructed in $\asympO(C \, \len{\fml[1]} + C\, \strucn{\struc[1]} + \strucm{\struc[1]})$ time
and has "size" $\len{\fml[2]} = \asympO(C \, \len{\fml[1]})$.
\end{lemma}
\begin{proof}
We may assume that $\struc[1]$ interprets only these $C$ "relation symbols".
Note that $C \le \len{\fml[1]}$.
Let $\fml \in \fmlclass\sigres{\set{R_1, \ldots, R_C}}$,
where $R_1, \dots, R_C$ are binary "relation symbols".
Let $E$ be a fresh binary "relation symbol".
We construct $\struc[2]$ by extending $\struc[1]$ with dummy "elements" that encode the "relation symbols" $R_i$.
In the construction,
we put an $E^{2}$-labeled self-loop on "elements" originally in $\struc[1]$
and put an $E^{3}$- or $E^{5}$-labeled self-loop on the other (dummy) "elements", to distinguish them.
For $i \ge 1$, we define the $\kPP{3}$-"formula" $E^i(x, y)$ as follows:
\[E^i(x, y) \defeq E(x, y) \text{ if $i = 1$, and } E^i(x, y) \defeq \exists z.~ E(x, z) \land E^{i-1}(z, y) \text{ if $i \ge 2$,}\]
where $z$ is a "variable" distinct from $x$ and $y$; in recursive
occurrences, bound "variables" are implicitly alpha-renamed so that the displayed definition is read as a well-formed $\kPP{3}$ expression.
We then define the following "formula" $\intro*\pathfml_i(z, z')$:
\[\pathfml_i(z, z')
~\defeq~
\begin{cases}
  E(z, z') & \text{if $i = 0$,}\\
  \exists z''. E(z, z'') \land E^3(z'', z'') \land
  \exists z'''. E(z'', z''') \land E^5(z''', z''') \land \pathfml_{i-1}(z''', z') & \text{ if $i \ge 1$},
\end{cases}
\]
where $z'' \notin \set{z, z'}$ and $z''' \notin \set{z'', z'}$ ("eg", $z''' = z$ if $z \neq z'$).
Namely, $\pathfml_i(z,z')$ expresses the existence of a path from $z$ to $z'$
in which each recursive step first moves to a "vertex" carrying an $E^3$ self-loop,
then to a "vertex" carrying an $E^5$ self-loop, and then continues recursively.
Observe that $\pathfml_i(z, z')$ can be implemented using only three "variables" by repeatedly renaming bound "variables",
and hence $\pathfml_i(z, z')$ is a $\kPP{3}$-"formula".

For each "subformula" $\fml[3]$ of $\fml[1]$ (without "equality"),
we define its translation $\fml[3]^{\dagger}$ inductively as follows,
where $(\cdot)^{\dagger}$ commutes with the Boolean connectives:
\begin{align*}
(\exists x. \fml[3])^{\dagger} &\defeq \exists x. E^2(x,x) \land \fml[3]^{\dagger},&
(\forall x. \fml[3])^{\dagger} &\defeq \forall x. E^2(x,x) \to \fml[3]^{\dagger},&
R_i(x,y)^{\dagger} &\defeq \pathfml_{i+2}(x,y).
\end{align*}
Let $\fml[2] \defeq \fml[1]^{\dagger}$.
Since $\pathfml_i(z, z')$ is a $\kPP{3}$-"formula" and $k \geq 3$,
$\fml[2]$ is in $\fmlclass\sigres{\set{E}}$.
By $\len{\pathfml_i} = \asympO(C)$,
$\len{\fml[2]} = \asympO(k^2 C + \len{\fml[1]}) = \asympO_k(\len{\fml[1]})$ for "DAG-represented" "sentences"
and $\len{\fml[2]} = \asympO(C \len{\fml[1]})$ for "tree-represented" "sentences".

Let $\struc[2] = \tuple{\univ{\struc[2]}, E^{\struc[2]}}$ be the "structure" defined as follows:
\begin{itemize}
    \item 
    Let
    \[p(i) \defeq \begin{cases}
        2 & (\text{if $i = 0$})\\
        3 & (\text{if $i \ge 1$ odd})\\
        5 & (\text{if $i \ge 1$ even}).
    \end{cases}\]
    \item 
    Let $\univ{\struc[2]} \defeq \set{\tuple{v, i, j} \mid v \in \univ{\struc}, i \in \range{0\,..\,2(C+2)}, j \in \range{0\,..\,p(i)-1}}$.
    We write $v^{i,j}$ for $\tuple{v, i, j} \in \univ{\struc[2]}$, and say that $v^{i,j}$ is in \emph{layer} $i$.
    \item
    $E^{\struc[2]} \defeq E_{\circlearrowright} \cup E_{\to} \cup \bigcup^C_{i = 1} E_{R_i}$ where
    \begin{align*}
    E_{\circlearrowright} &\defeq
    \set{\tuple{v^{i, j}, v^{i, (j + 1) \bmod p(i)}} \mid  v \in \univ{\struc}, i \in \range{0\,..\,2(C+2)}, j \in \range{0\,..\, p(i)-1}}, \\
    E_{\to} &\defeq \set{\tuple{v^{i,j}, v^{i+1,j'}} \mid v \in \univ{\struc}, i \in \range{0\,..\,2(C+2)-1}, j \in \range{0\,..\, p(i)-1},
    j' \in \range{0\,..\, p(i+1)-1}},\\
    E_{R_i} &\defeq \set{\tuple{v^{2(i+2), j}, w^{0, j'}} \mid \tuple{v, w} \in R_i^{\struc}, j \in \range{0\,..\, p(2(i+2))-1},
    j' \in \set{0, 1}}.
    \end{align*}
\end{itemize}

The following schematically illustrates $\struc[2]$ for $C = 3$:
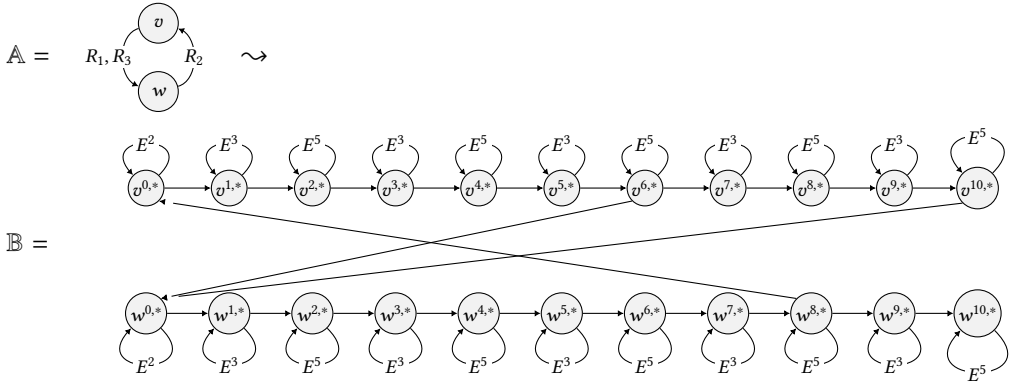
\knowledgeconfigure {quotation=false}
\begin{align*}
    & \struc =\quad \begin{tikzpicture}[baseline = -3.5ex]
    \tikzstyle{vert}=[draw = black, circle, fill = gray!10, inner sep = .5pt, minimum size = 3pt, font = \scriptsize]
    \graph[grow right = 1.cm, branch down = 6.ex]{
    {s1/{$v$}[vert, inner sep = 1pt, minimum size = 15pt],s2/{$w$}[vert, inner sep = 1pt, minimum size = 15pt]} -!-
    };
    \path (s1) edge[->, bend right = 75, earrow] node[fill=white, inner sep=1pt,font= \scriptsize, pos = .5, xshift = -.5em]{$R_1, R_3$} (s2);
    \path (s2) edge[->, bend right = 75, earrow] node[fill=white, inner sep=1pt,font= \scriptsize, pos = .5]{$R_2$} (s1);
\end{tikzpicture}
\quad \leadsto \\
& \struc[2] =\quad
\begin{tikzpicture}[baseline = -5.5ex]
    \tikzstyle{vert}=[draw = black, circle, fill = gray!10, inner sep = .5pt, minimum size = 3pt, font = \scriptsize]
    \graph[grow right = 1.1cm, branch down = 11.ex]{
    {s10/{$v^{0,*}$}[vert], s20/{$w^{0,*}$}[vert]} -!-
    {s11/{$v^{1,*}$}[vert], s21/{$w^{1,*}$}[vert]} -!-
    {s12/{$v^{2,*}$}[vert], s22/{$w^{2,*}$}[vert]} -!-
    {s13/{$v^{3,*}$}[vert], s23/{$w^{3,*}$}[vert]} -!-
    {s14/{$v^{4,*}$}[vert], s24/{$w^{4,*}$}[vert]} -!-
    {s15/{$v^{5,*}$}[vert], s25/{$w^{5,*}$}[vert]} -!-
    {s16/{$v^{6,*}$}[vert], s26/{$w^{6,*}$}[vert]} -!-
    {s17/{$v^{7,*}$}[vert], s27/{$w^{7,*}$}[vert]} -!-
    {s18/{$v^{8,*}$}[vert], s28/{$w^{8,*}$}[vert]} -!-
    {s19/{$v^{9,*}$}[vert], s29/{$w^{9,*}$}[vert]} -!-
    {s110/{$v^{10,*}$}[vert], s210/{$w^{10,*}$}[vert]} -!-
    };
    \graph[use existing nodes, edges={color=black, pos = .5, earrow}, edge quotes={fill=white, inner sep=1pt,font= \scriptsize}]{
        s10 ->[] s11 ->[] s12 ->[] s13 ->[] s14 ->[] s15 -> s16 ->[] s17 ->[] s18 ->[] s19 ->[] s110;
        {s10} ->["$E^2$", out = 45, in = 135, looseness =6] {s10};
        {s11,s13,s15,s17,s19} ->["$E^3$", out = 45, in = 135, looseness =6] {s11,s13,s15,s17,s19};
        {s12,s14,s16,s18,s110} ->["$E^5$", out = 45, in = 135, looseness =6] {s12,s14,s16,s18,s110};
        s20 ->[] s21 ->[] s22 ->[] s23 ->[] s24 ->[] s25 -> s26 ->[] s27 ->[] s28 ->[] s29 ->[] s210;
        {s20} ->["$E^2$", out = -45, in = -135, looseness =6] {s20};
        {s21,s23,s25,s27,s29} ->["$E^3$", out = -45, in = -135, looseness =6] {s21,s23,s25,s27,s29};
        {s22,s24,s26,s28,s210} ->["$E^5$", out = -45, in = -135, looseness =6] {s22,s24,s26,s28,s210};
        s16 ->[out=-135, in = 45, looseness = 0] s20;
        s110->[out=-135, in = 45, looseness = 0] s20;
        s28 ->[out=135, in=-45, looseness = 0] s10;
    };
\end{tikzpicture}
\end{align*}
(Each node $v^{i,*}$ represents the copies $v^{i,0}, \dots, v^{i,p(i)-1}$,
and the self-loop labeled $E^{p(i)}$ represents their $E_{\circlearrowright}$-cycle $v^{i,0} \to v^{i,1} \to \dots \to v^{i,p(i)-1} \to v^{i,0}$;
each arrow between nodes represents the edges between all these copies.)
\knowledgeconfigure {quotation}

By construction, we have:
\begin{align*}
\strucn{\struc[2]}
&= \asympO(C \, \strucn{\struc}), &
\strucm{\struc[2]}
&= \card{E^{\struc[2]}}
= \card{E_{\circlearrowright}} + \card{E_{\to}} + \sum^C_{i = 1} \card{E_{R_i}}
= \asympO(C \, \strucn{\struc} + C \, \strucn{\struc} + \strucm{\struc}) = \asympO(C \, \strucn{\struc} + \strucm{\struc}).
\end{align*}
Also, the total construction time is
$\asympO(\len{\fml[2]} + \strucn{\struc[2]} + \strucm{\struc[2]}) \le \asympO_k(\len{\fml} + C \, \strucn{\struc[1]} + \strucm{\struc[1]})$
for "DAG-represented" "sentences" and
$\asympO(C \, \len{\fml} + C \, \strucn{\struc[1]} + \strucm{\struc[1]})$ for "tree-represented" "sentences".

Finally, we prove correctness.
Note that $\struc[2] \models E^\ell(u, u)$ "iff" there is a walk of length $\ell$ from $u$ back to $u$.
Since $E_{\to}$-edges increase the layer by one and only $E_{R_i}$-edges decrease it (by $2(i+2) \ge 6$),
every walk from $u$ back to $u$ that leaves the layer of $u$ has length at least $7$.
Hence, for $\ell \le 5$, every walk of length $\ell$ from $u$ back to $u$ stays in the $E_{\circlearrowright}$-cycle through $u$ (of length $p(h)$, where $u$ is in layer $h$). This gives the following.

\begin{claim*}
$\set{u \in \univ{\struc[2]} \mid \struc[2] \models E^2(u, u)} = \univ{\struc} \times \set{0} \times \set{0,1}$.
\end{claim*}
\begin{claim*}
$\set{u \in \univ{\struc[2]} \mid \struc[2] \models E^3(u, u)} = \univ{\struc} \times \set{i \in \range{2(C+2)} \mid i \text{ odd}} \times \range{0\,..\,2}$.
\end{claim*}
\begin{claim*}
$\set{u \in \univ{\struc[2]} \mid \struc[2] \models E^5(u, u)} = \univ{\struc} \times \set{i \in \range{2(C+2)} \mid i \text{ even}} \times \range{0\,..\,4}$.
\end{claim*}
Using them, we have the following.
\begin{claim*}
For all $u \in \univ{\struc}$, $s \in \range{0\,..\,C+1}$, and $j \in \range{0\,..\,p(2s)-1}$,
the pairs $\tuple{c, d}$ of "elements" with $\struc[2] \models E(u^{2s,j}, c) \land E^3(c, c) \land E(c, d) \land E^5(d, d)$
are exactly the pairs $\tuple{u^{2s+1,j_1}, u^{2s+2,j_2}}$.
\end{claim*}
\begin{proof}
Every $c$ with $\struc[2] \models E(u^{2s,j}, c)$ lies in layer $2s$ or $0$, or is of the form $u^{2s+1,j_1}$;
by the claims above, only the latter satisfies $E^3(c, c)$.
Similarly, every $d$ with $\struc[2] \models E(u^{2s+1,j_1}, d)$ lies in layer $2s+1$ or is of the form $u^{2s+2,j_2}$,
and only the latter satisfies $E^5(d, d)$.
\end{proof}
This gives the correctness for atomic formulas:
\begin{claim*}
For every $i \in \range{C}$, $u, v\in \univ{\struc}$, and $j,j' \in \set{0,1}$,
$\struc[2] \models \pathfml_{i+2}(u^{0,j},v^{0,j'}) \iff \struc[1] \models R_{i}(u,v)$.
\end{claim*}
\begin{proof}
Unfolding $\pathfml_{i+2}$ by applying the previous claim $i+2$ times (for $s = 0, 1, \dots, i+1$),
we have that $\struc[2] \models \pathfml_{i+2}(u^{0,j},v^{0,j'})$ "iff" $\struc[2] \models E(u^{2(i+2),j''}, v^{0,j'})$ for some $j''$.
Since the edges from layer $2(i+2)$ to layer $0$ are the $E_{R_i}$-edges,
this is equivalent to $\struc[1] \models R_i(u, v)$.
\end{proof}
Let $\pi \colon \univ{\struc} \times \set{0} \times \set{0,1} \to \univ{\struc}$ be the projection $\pi(v^{0,j}) \defeq v$.
Then, by a straightforward induction on "subformulas",
for every "subformula" $\fml[3]$ of $\fml[1]$, its translation $\fml[3]^{\dagger}$ in $\fml[2]$, and
every "valuation" $I$ of "variables" to $\univ{\struc} \times \set{0} \times \set{0,1}$,
we have:
$\struc, \pi \circ I \models \fml[3]$ "iff" $\struc[2], I \models \fml[3]^{\dagger}$.
Hence, $\struc \models \fml[1]$ "iff" $\struc[2] \models \fml[2]$.
\end{proof}

Combining the reductions above yields the following.
\begin{lemma*}[Restatement of {\Cref{lemma: reduction monadic and binary predicates to single binary predicate}}]
\reductionmonadicandbinarypredicatestosinglebinarypredicatelemma
\end{lemma*}
\begin{proof}
Let $\fml'$ and $\struc'$ be obtained from $\fml$ and $\struc$ by applying \Cref{proposition: arity monotone,proposition: equality elimination}.
Then $\len{\fml'} = \asympO(\len{\fml})$, $\strucn{\struc'} = \asympO(\strucn{\struc})$, and $\strucm{\struc'} = \strucm{\struc} + \strucn{\struc}$.
Let $\fml[2]$ and $\struc[2]$ be obtained from $\fml'$ and $\struc'$ by applying \Cref{lemma: unary encoding}.
We then have:
\begin{align*}
\len{\fml[2]} &= \asympO(\len{\fml'}) = \asympO(\len{\fml}),&
\strucn{\struc[2]} &= \asympO(C \, \strucn{\struc'}) = \asympO(C \, \strucn{\struc}),&
\strucm{\struc[2]} &= \asympO(C \, \strucn{\struc'} + \strucm{\struc'}) = \asympO(C \, \strucn{\struc} + \strucm{\struc}).
\end{align*}
The construction takes $\asympO(\len{\fml} + \strucn{\struc} + \strucm{\struc}) + \asympO(\len{\fml'} + C \, \strucn{\struc'} + \strucm{\struc'})
= \asympO(\len{\fml} + C \, \strucn{\struc} + \strucm{\struc})$ time.
As the truth of $\fml$ on $\struc$ is preserved through the reductions, this completes the proof.
\end{proof}
 
\ifcameraready
\clearpage
\section{Supplement to \texorpdfstring{\Cref{proposition: alman VVW}}{Proposition~\ref{proposition: alman VVW}}}\label{section: proposition: alman VVW}
In this section, we supplement \Cref{proposition: alman VVW} by giving a proof.
\subsection{\texorpdfstring{\PartitionedOVHyperCliqueProblem{n, d}}{Partitioned-OV-HyperClique}: revisiting the reduction of \texorpdfstring{\cite{almanOVGraphsAre2020}}{[..]}}
Given $k > r \ge 1$ and $n$-element sets $\mathcal{V}_1, \dots, \mathcal{V}_k \subseteq \set{0, 1}^d$ of "vectors",
the $k$-partite $r$-\reintro*\kl{uniform OV-hypergraph} $\reintro*\OVgraph{\mathcal{V}_1, \dots, \mathcal{V}_k, r}$ is defined as the $r$-"uniform hypergraph" $\OVgraph{\mathcal{V}_1, \dots, \mathcal{V}_k, r} = \tuple{
  \bigsqcup_{i = 1}^{k} \mathcal{V}_i,
  E
}$ where $S \in E$ "iff" $S$ is an $r$-element subset of $\bigsqcup_{i = 1}^{k} \mathcal{V}_i$ such that $\GIP{S\ith{1}, \dots, S\ith{r}} = 0$ and $S$ contains at most one "vector" from each $\mathcal{V}_i$.
Here, $\bigsqcup$ denotes the disjoint union ("ie", a "vector" occurring in several $\mathcal{V}_i$ yields distinct vertices, one for each such $i$).
The $(k, r)$-\intro*\PartitionedOVHyperCliqueProblem{n, d} is the following problem:
given $\mathcal{V}_1, \dots, \mathcal{V}_k \subseteq \set{0, 1}^d$ where each $\mathcal{V}_i$ is an $n$-element set, decide whether $\OVgraph{\mathcal{V}_1, \dots, \mathcal{V}_k, r}$ has a $(k, r)$-"hyperclique".
In particular, we write \intro*\PartitionedOVTriangleProblem{n, d} for the $(3, 2)$-\PartitionedOVHyperCliqueProblem{n, d}.
The reduction of Alman and Vassilevska Williams \cite{almanOVGraphsAre2020} implies the corresponding hypotheses below.
\AP
\begin{hypothesis*}[$(k, r)$-\intro*\PartitionedOVHyperCliqueHypothesis{} {\cite[Theorem~23]{almanOVGraphsAre2020}}%
]
Fix $k > r \ge 3$.
For any $\eps > 0$,
there is a "constant" $c$ such that
the $(k, r)$-\PartitionedOVHyperCliqueProblem{n, \le c \log n} cannot be solved in $\asympO(n^{k - \eps})$ time.
\end{hypothesis*}
\begin{hypothesis*}[\intro*\PartitionedOVTriangleHypothesis{}, $(3, 2)$-\PartitionedOVHyperCliqueHypothesis{} {\cite[Theorem~5]{almanOVGraphsAre2020}}]
For any $\eps > 0$,
there is a "constant" $c$ such that
the \PartitionedOVTriangleProblem{n, \le c \log n} cannot be solved in $\asympO(n^{\matrixexponent - \eps})$ time.
\end{hypothesis*}
\begin{proposition}\label{proposition: alman VVW'}
For each fixed $k > r \ge 3$ ("resp", $k = 3$ and $r = 2$),
the \SparseMAXxSATHypothesis{r} implies the $(k, r)$-\PartitionedOVHyperCliqueHypothesis{} ("resp", \PartitionedOVTriangleHypothesis{}).
\end{proposition}
\begin{proof}
We use the reduction of Alman and Vassilevska Williams~\cite[Theorem~23 ("resp", Theorem~5)]{almanOVGraphsAre2020}.
For \cite[Theorem~5]{almanOVGraphsAre2020},
this follows by combining
the reduction from \SparseMAXxSATProblem{2} to multiple instances of the \emph{partitioned version} of \intro*\ExactIPTriangleProblem{} and
the reduction from the \emph{partitioned version} of \ExactIPTriangleProblem{}
to multiple instances of the \emph{non-partitioned version} of \OVTriangleProblem{} in \cite[Lemma~19]{almanOVGraphsAre2020}
(see \cite{almanOVGraphsAre2020} for the definition of \ExactIPTriangleProblem{}).
The reduction of \cite[Lemma~19]{almanOVGraphsAre2020}
works even if we keep the original partition in \ExactIPTriangleProblem{} (denoted by $A, B, C$ in \cite[Lemma~19]{almanOVGraphsAre2020}),
and hence it also gives a reduction to multiple instances of \PartitionedOVTriangleProblem{}.
In the same way,
the reduction of \cite[Theorem~23]{almanOVGraphsAre2020}
also works for the partitioned version of \intro*\ExactIPHyperCliqueProblem{} ("ie", the problem corresponding to \ExactIPTriangleProblem{} for $k > r \ge 3$), keeping the partition.
The number of output instances, their dimensions, and their sizes are the same as in the original reduction.
\end{proof}

\subsection{From \texorpdfstring{\PartitionedOVHyperCliqueProblem{n, d}}{Partitioned-OV-HyperClique} to \texorpdfstring{\OVHyperCliqueProblem{\asympO(n), d + \asympO(1)}}{OV-HyperClique}}

\begin{proposition}\label{proposition: (kr)-labelling}
Fix $k > r \ge 2$.
There exist $\const{enc}_1, \dots, \const{enc}_k \in \set{0,1}^{r \binom{k}{r}}$
such that for all $i_1, \dots, i_r \in \range{k}$, we have:
$\GIP{\const{enc}_{i_1}, \dots, \const{enc}_{i_r}} = 0$
"iff"
$i_1, \dots, i_r$ are pairwise distinct.
Hence, for every nonempty subset $S$ of $\range{k}$ with $r' \defeq \card{S} < r$,
$\GIP{\const{enc}_{S\ith{1}}, \dots, \const{enc}_{S\ith{r'}}} \neq 0$.
\end{proposition}
\begin{proof}
We enumerate the $r$-element subsets of $\range{k}$ as $S_1, \dots, S_{\binom{k}{r}}$.
For each $i \in \range{k}$ and $j \in \range{\binom{k}{r}}$,
let $\const{enc}_{i, j} \in \set{0, 1}^{r}$ be the following "vector":
\[\const{enc}_{i, j} ~\defeq~
\begin{cases}
\quad\underbracket{1 \dots 1}_{\text{$\ell-1$ times}}0\underbracket{1 \dots 1}_{\text{$r-\ell$ times}}
& \text{ if the $\ell$-th smallest element of $S_{j}$ is $i$ for some $\ell$, and}\\
\quad\underbracket{0 \dots 0}_{\text{$r$ times}}
& \text{ otherwise ("ie", $i \not\in S_{j}$)}.
\end{cases}\]
We then let $\const{enc}_i \defeq \const{enc}_{i, 1} \const{enc}_{i, 2} \dots \const{enc}_{i, \binom{k}{r}}$.
Observe that:
\[\GIP{\const{enc}_{i_1, j}, \dots, \const{enc}_{i_r, j}} = 0 \quad\iff\quad \text{$\set{i_1, \dots, i_r} = S_j$ or $\set{i_1, \dots, i_r} \not\subseteq S_j$}.\]
Since $\GIP{\const{enc}_{i_1}, \dots, \const{enc}_{i_r}} = \sum_{j} \GIP{\const{enc}_{i_1, j}, \dots, \const{enc}_{i_r, j}}$ is a sum of non-negative integers,
it is $0$ "iff" every block $j$ has $\GIP{\const{enc}_{i_1, j}, \dots, \const{enc}_{i_r, j}} = 0$.
Using this, we obtain the desired equivalence:
\proofcase{($\Leftarrow$)}
Suppose that $i_1, \dots, i_r$ are pairwise distinct, "ie", $\card{\set{i_1, \dots, i_r}} = r$.
For each $j$, since $\card{S_j} = r$, if $S_j \neq \set{i_1, \dots, i_r}$, then $\set{i_1, \dots, i_r} \not\subseteq S_j$.
Hence every block $j$ has $\GIP{\const{enc}_{i_1, j}, \dots, \const{enc}_{i_r, j}} = 0$.
\proofcase{($\Rightarrow$)}
If $i_1, \dots, i_r$ are not pairwise distinct, then $\card{\set{i_1, \dots, i_r}} < r \le k$, so there is an $S_j$ such that $\set{i_1, \dots, i_r} \subsetneq S_j$.
Then this block $j$ has $\GIP{\const{enc}_{i_1, j}, \dots, \const{enc}_{i_r, j}} \neq 0$, and hence $\GIP{\const{enc}_{i_1}, \dots, \const{enc}_{i_r}} \neq 0$.

For the second sentence of the statement, let $S \subseteq \range{k}$ with $r' = \card{S} < r$.
As $r' < r \le k$, there is an $S_j$ such that $S \subsetneq S_j$.
Let $\ell$ be such that the $\ell$-th smallest element of $S_j$ is not in $S$.
Then, for every $i \in S$, the $\ell$-th bit of $\const{enc}_{i, j}$ is $1$;
hence $\GIP{\const{enc}_{S\ith{1}, j}, \dots, \const{enc}_{S\ith{r'}, j}} \ge 1$, and thus $\GIP{\const{enc}_{S\ith{1}}, \dots, \const{enc}_{S\ith{r'}}} \neq 0$.
\end{proof}

\begin{lemma}\label{lemma: non-partitioned binary hyperclique hypothesis}
Fix $k > r \ge 2$.
Then there exists an $\asympO(n (d + r \binom{k}{r}))$-time reduction
from an instance $(\mathcal{V}_1, \dots, \mathcal{V}_k)$ of the $(k,r)$-\PartitionedOVHyperCliqueProblem{n, d} to an instance $\mathcal{V}'$ of the $(k,r)$-$\OVHyperCliqueProblem{k n, d + r \binom{k}{r}}$ such that:
$\OVgraph{\mathcal{V}_1,\dots, \mathcal{V}_k, r}$ has some $(k, r)$-"hyperclique" "iff"
$\OVgraph{\mathcal{V}', r}$ has some $(k, r)$-"hyperclique".
\end{lemma}
\begin{proof}
Let $\const{enc}_1, \dots, \const{enc}_k \in \set{0, 1}^{r \binom{k}{r}}$ be the "vectors"
from \Cref{proposition: (kr)-labelling}.
Given an instance $\tuple{\mathcal{V}_1, \dots, \mathcal{V}_k}$,
where each $\mathcal{V}_i$ is an $n$-element set of "vectors" in $\set{0, 1}^{d}$,
let $\mathcal{V}' \defeq \bigsqcup_{i = 1}^{k} \set{v \ \const{enc}_{i} \mid v \in \mathcal{V}_i}$.
By \Cref{proposition: (kr)-labelling}, $\mathcal{V}'$ is indeed "well-formed" (several "vectors" may share a label, but repeated $0/1$ "vectors" do not change the "generalized inner product").
Also, $\mathcal{V}'$ can be computed in $\asympO(n (d + r \binom{k}{r}))$ time.
We show that $\OVgraph{\mathcal{V}_1,\dots, \mathcal{V}_k, r}$ has a $(k, r)$-"hyperclique" "iff"
$\OVgraph{\mathcal{V}', r}$ has a $(k, r)$-"hyperclique".
\proofcase{($\Rightarrow$)}
Suppose $\OVgraph{\mathcal{V}_1, \dots, \mathcal{V}_k, r}$ has a $(k, r)$-"hyperclique" $v_1 \in \mathcal{V}_1, \dots, v_k \in \mathcal{V}_k$.
Then $v_1 \const{enc}_1, \dots, v_k \const{enc}_k$ form a $(k, r)$-"hyperclique" in $\OVgraph{\mathcal{V}', r}$,
since $\GIP{\const{enc}_{i_1}, \dots, \const{enc}_{i_r}} = 0$ for every $r$-element subset $\set{i_1, \dots, i_r}$ of $\range{k}$, by the property of \Cref{proposition: (kr)-labelling}.
\proofcase{($\Leftarrow$)}
Suppose $\OVgraph{\mathcal{V}', r}$ has a $(k, r)$-"hyperclique" $S$.
By \Cref{proposition: (kr)-labelling},
$S$ consists of $v_1' \const{enc}_{i_1}, \dots, v_k' \const{enc}_{i_k}$ such that
$i_1, \dots, i_k$ are pairwise distinct.
Hence, $S$ consists of $v_1 \const{enc}_1, \dots, v_k \const{enc}_k$.
Then $v_1 \in \mathcal{V}_1, \dots, v_k \in \mathcal{V}_k$ form a $(k, r)$-"hyperclique" in $\OVgraph{\mathcal{V}_1,\dots, \mathcal{V}_k, r}$.
This completes the proof.
\end{proof}
Using this reduction, we have the following.
\begin{proposition}\label{proposition: partitioned OV-hyperclique hypothesis}
Fix $k > r \ge 3$.
The $(k, r)$-\PartitionedOVHyperCliqueHypothesis{} implies the $(k, r)$-\OVHyperCliqueHypothesis{}.
Similarly, the \PartitionedOVTriangleHypothesis{} implies the \OVTriangleHypothesis{}.
\end{proposition}
\begin{proof}
We show the contrapositive, that is, if
there is some $\eps > 0$ such that for all $c > 0$,
the $(k, r)$-\OVHyperCliqueProblem{n, \le c \log n} can be solved in $\asympO(n^{k - \eps})$ time, then
for all $c' > 0$,
the $(k, r)$-\PartitionedOVHyperCliqueProblem{n, \le c' \log n} can be solved in $\asympO(n^{k - \eps})$ time.
Let $\eps > 0$ witness this. We may assume $k - \eps > 1$.
Fix any $c' > 0$ and choose $c$ so that $c > c'$.
By $c > c'$, for all sufficiently large $n$, we have $c' \log n + r \binom{k}{r} \le c \log n$.
By \Cref{lemma: non-partitioned binary hyperclique hypothesis},
the $(k, r)$-\PartitionedOVHyperCliqueProblem{n, \le c' \log n} reduces to
the $(k, r)$-\OVHyperCliqueProblem{k n, \le c \log (k n)};
hence, by the assumption, it can be solved in
time $\asympO(n (c' \log n + r \binom{k}{r})) + \asympO((k n)^{k - \eps}) \le \asympO(n^{k - \eps})$.
The case of the \OVTriangleHypothesis{} is shown in the same way with $(k, r) = (3, 2)$, replacing the exponent $k - \eps$ with $\matrixexponent - \eps$;
note that \Cref{lemma: non-partitioned binary hyperclique hypothesis} works for $k = 3$ and $r = 2$ as well.
\end{proof}
\Cref{proposition: alman VVW} follows from \Cref{proposition: alman VVW',proposition: partitioned OV-hyperclique hypothesis}.
 \else
\clearpage
\section{Supplement to \texorpdfstring{\Cref{proposition: alman VVW}}{Proposition~\ref{proposition: alman VVW}}}\label{section: proposition: alman VVW}
In this section, we supplement \Cref{proposition: alman VVW} by giving a proof.
\subsection{\texorpdfstring{\PartitionedOVHyperCliqueProblem{n, d}}{Partitioned-OV-HyperClique}: revisiting the reduction of \texorpdfstring{\cite{almanOVGraphsAre2020}}{[..]}}
Given $k > r \ge 1$ and $n$-element sets $\mathcal{V}_1, \dots, \mathcal{V}_k \subseteq \set{0, 1}^d$ of "vectors",
the $k$-partite $r$-\reintro*\kl{uniform OV-hypergraph} $\reintro*\OVgraph{\mathcal{V}_1, \dots, \mathcal{V}_k, r}$ is defined as the $r$-"uniform hypergraph" $\OVgraph{\mathcal{V}_1, \dots, \mathcal{V}_k, r} = \tuple{
  \bigsqcup_{i = 1}^{k} \mathcal{V}_i,
  E
}$ where $S \in E$ "iff" $S$ is an $r$-element subset of $\bigsqcup_{i = 1}^{k} \mathcal{V}_i$ such that $\GIP{S\ith{1}, \dots, S\ith{r}} = 0$ and $S$ contains at most one "vector" from each $\mathcal{V}_i$.
Here, $\bigsqcup$ denotes the disjoint union ("ie", a "vector" occurring in several $\mathcal{V}_i$ yields distinct vertices, one for each such $i$).
The $(k, r)$-\intro*\PartitionedOVHyperCliqueProblem{n, d} is the following problem:
given $\mathcal{V}_1, \dots, \mathcal{V}_k \subseteq \set{0, 1}^d$ where each $\mathcal{V}_i$ is an $n$-element set, decide whether $\OVgraph{\mathcal{V}_1, \dots, \mathcal{V}_k, r}$ has a $(k, r)$-"hyperclique".
In particular, we write \intro*\PartitionedOVTriangleProblem{n, d} for the $(3, 2)$-\PartitionedOVHyperCliqueProblem{n, d}.
The reduction of Alman and Vassilevska Williams \cite{almanOVGraphsAre2020} implies the corresponding hypotheses below.
\AP
\begin{hypothesis*}[$(k, r)$-\intro*\PartitionedOVHyperCliqueHypothesis{} {\cite[Theorem~23]{almanOVGraphsAre2020}}%
]
Fix $k > r \ge 3$.
For any $\eps > 0$,
there is a "constant" $c$ such that
the $(k, r)$-\PartitionedOVHyperCliqueProblem{n, \le c \log n} cannot be solved in $\asympO(n^{k - \eps})$ time.
\end{hypothesis*}
\begin{hypothesis*}[\intro*\PartitionedOVTriangleHypothesis{}, $(3, 2)$-\PartitionedOVHyperCliqueHypothesis{} {\cite[Theorem~5]{almanOVGraphsAre2020}}]
For any $\eps > 0$,
there is a "constant" $c$ such that
the \PartitionedOVTriangleProblem{n, \le c \log n} cannot be solved in $\asympO(n^{\matrixexponent - \eps})$ time.
\end{hypothesis*}
\begin{proposition}\label{proposition: alman VVW'}
For each fixed $k > r \ge 3$ ("resp", $k = 3$ and $r = 2$),
the \SparseMAXxSATHypothesis{r} implies the $(k, r)$-\PartitionedOVHyperCliqueHypothesis{} ("resp", \PartitionedOVTriangleHypothesis{}).
\end{proposition}
\begin{proof}
We use the reduction of Alman and Vassilevska Williams~\cite[Theorem~23 ("resp", Theorem~5)]{almanOVGraphsAre2020}.
For \cite[Theorem~5]{almanOVGraphsAre2020},
this follows by combining
the reduction from \SparseMAXxSATProblem{2} to multiple instances of the \emph{partitioned version} of \intro*\ExactIPTriangleProblem{} and
the reduction from the \emph{partitioned version} of \ExactIPTriangleProblem{}
to multiple instances of the \emph{non-partitioned version} of \OVTriangleProblem{} in \cite[Lemma~19]{almanOVGraphsAre2020}
(see \cite{almanOVGraphsAre2020} for the definition of \ExactIPTriangleProblem{}).
The reduction of \cite[Lemma~19]{almanOVGraphsAre2020}
works even if we keep the original partition in \ExactIPTriangleProblem{} (denoted by $A, B, C$ in \cite[Lemma~19]{almanOVGraphsAre2020}),
and hence it also gives a reduction to multiple instances of \PartitionedOVTriangleProblem{}.
In the same way,
the reduction of \cite[Theorem~23]{almanOVGraphsAre2020}
also works for the partitioned version of \intro*\ExactIPHyperCliqueProblem{} ("ie", the problem corresponding to \ExactIPTriangleProblem{} for $k > r \ge 3$), keeping the partition.
The number of output instances, their dimensions, and their sizes are the same as in the original reduction.
\end{proof}

\subsection{From \texorpdfstring{\PartitionedOVHyperCliqueProblem{n, d}}{Partitioned-OV-HyperClique} to \texorpdfstring{\OVHyperCliqueProblem{\asympO(n), d + \asympO(1)}}{OV-HyperClique}}

\begin{proposition}\label{proposition: (kr)-labelling}
Fix $k > r \ge 2$.
There exist $\const{enc}_1, \dots, \const{enc}_k \in \set{0,1}^{r \binom{k}{r}}$
such that for all $i_1, \dots, i_r \in \range{k}$, we have:
$\GIP{\const{enc}_{i_1}, \dots, \const{enc}_{i_r}} = 0$
"iff"
$i_1, \dots, i_r$ are pairwise distinct.
Hence, for every nonempty subset $S$ of $\range{k}$ with $r' \defeq \card{S} < r$,
$\GIP{\const{enc}_{S\ith{1}}, \dots, \const{enc}_{S\ith{r'}}} \neq 0$.
\end{proposition}
\begin{proof}
We enumerate the $r$-element subsets of $\range{k}$ as $S_1, \dots, S_{\binom{k}{r}}$.
For each $i \in \range{k}$ and $j \in \range{\binom{k}{r}}$,
let $\const{enc}_{i, j} \in \set{0, 1}^{r}$ be the following "vector":
\[\const{enc}_{i, j} ~\defeq~
\begin{cases}
\quad\underbracket{1 \dots 1}_{\text{$\ell-1$ times}}0\underbracket{1 \dots 1}_{\text{$r-\ell$ times}}
& \text{ if the $\ell$-th smallest element of $S_{j}$ is $i$ for some $\ell$, and}\\
\quad\underbracket{0 \dots 0}_{\text{$r$ times}}
& \text{ otherwise ("ie", $i \not\in S_{j}$)}.
\end{cases}\]
We then let $\const{enc}_i \defeq \const{enc}_{i, 1} \const{enc}_{i, 2} \dots \const{enc}_{i, \binom{k}{r}}$.
Observe that:
\[\GIP{\const{enc}_{i_1, j}, \dots, \const{enc}_{i_r, j}} = 0 \quad\iff\quad \text{$\set{i_1, \dots, i_r} = S_j$ or $\set{i_1, \dots, i_r} \not\subseteq S_j$}.\]
Since $\GIP{\const{enc}_{i_1}, \dots, \const{enc}_{i_r}} = \sum_{j} \GIP{\const{enc}_{i_1, j}, \dots, \const{enc}_{i_r, j}}$ is a sum of non-negative integers,
it is $0$ "iff" every block $j$ has $\GIP{\const{enc}_{i_1, j}, \dots, \const{enc}_{i_r, j}} = 0$.
Using this, we obtain the desired equivalence:
\proofcase{($\Leftarrow$)}
Suppose that $i_1, \dots, i_r$ are pairwise distinct, "ie", $\card{\set{i_1, \dots, i_r}} = r$.
For each $j$, since $\card{S_j} = r$, if $S_j \neq \set{i_1, \dots, i_r}$, then $\set{i_1, \dots, i_r} \not\subseteq S_j$.
Hence every block $j$ has $\GIP{\const{enc}_{i_1, j}, \dots, \const{enc}_{i_r, j}} = 0$.
\proofcase{($\Rightarrow$)}
If $i_1, \dots, i_r$ are not pairwise distinct, then $\card{\set{i_1, \dots, i_r}} < r \le k$, so there is an $S_j$ such that $\set{i_1, \dots, i_r} \subsetneq S_j$.
Then this block $j$ has $\GIP{\const{enc}_{i_1, j}, \dots, \const{enc}_{i_r, j}} \neq 0$, and hence $\GIP{\const{enc}_{i_1}, \dots, \const{enc}_{i_r}} \neq 0$.

For the second sentence of the statement, let $S \subseteq \range{k}$ with $r' = \card{S} < r$.
As $r' < r \le k$, there is an $S_j$ such that $S \subsetneq S_j$.
Let $\ell$ be such that the $\ell$-th smallest element of $S_j$ is not in $S$.
Then, for every $i \in S$, the $\ell$-th bit of $\const{enc}_{i, j}$ is $1$;
hence $\GIP{\const{enc}_{S\ith{1}, j}, \dots, \const{enc}_{S\ith{r'}, j}} \ge 1$, and thus $\GIP{\const{enc}_{S\ith{1}}, \dots, \const{enc}_{S\ith{r'}}} \neq 0$.
\end{proof}

\begin{lemma}\label{lemma: non-partitioned binary hyperclique hypothesis}
Fix $k > r \ge 2$.
Then there exists an $\asympO(n (d + r \binom{k}{r}))$-time reduction
from an instance $(\mathcal{V}_1, \dots, \mathcal{V}_k)$ of the $(k,r)$-\PartitionedOVHyperCliqueProblem{n, d} to an instance $\mathcal{V}'$ of the $(k,r)$-$\OVHyperCliqueProblem{k n, d + r \binom{k}{r}}$ such that:
$\OVgraph{\mathcal{V}_1,\dots, \mathcal{V}_k, r}$ has some $(k, r)$-"hyperclique" "iff"
$\OVgraph{\mathcal{V}', r}$ has some $(k, r)$-"hyperclique".
\end{lemma}
\begin{proof}
Let $\const{enc}_1, \dots, \const{enc}_k \in \set{0, 1}^{r \binom{k}{r}}$ be the "vectors"
from \Cref{proposition: (kr)-labelling}.
Given an instance $\tuple{\mathcal{V}_1, \dots, \mathcal{V}_k}$,
where each $\mathcal{V}_i$ is an $n$-element set of "vectors" in $\set{0, 1}^{d}$,
let $\mathcal{V}' \defeq \bigsqcup_{i = 1}^{k} \set{v \ \const{enc}_{i} \mid v \in \mathcal{V}_i}$.
By \Cref{proposition: (kr)-labelling}, $\mathcal{V}'$ is indeed "well-formed" (several "vectors" may share a label, but repeated $0/1$ "vectors" do not change the "generalized inner product").
Also, $\mathcal{V}'$ can be computed in $\asympO(n (d + r \binom{k}{r}))$ time.
We show that $\OVgraph{\mathcal{V}_1,\dots, \mathcal{V}_k, r}$ has a $(k, r)$-"hyperclique" "iff"
$\OVgraph{\mathcal{V}', r}$ has a $(k, r)$-"hyperclique".
\proofcase{($\Rightarrow$)}
Suppose $\OVgraph{\mathcal{V}_1, \dots, \mathcal{V}_k, r}$ has a $(k, r)$-"hyperclique" $v_1 \in \mathcal{V}_1, \dots, v_k \in \mathcal{V}_k$.
Then $v_1 \const{enc}_1, \dots, v_k \const{enc}_k$ form a $(k, r)$-"hyperclique" in $\OVgraph{\mathcal{V}', r}$,
since $\GIP{\const{enc}_{i_1}, \dots, \const{enc}_{i_r}} = 0$ for every $r$-element subset $\set{i_1, \dots, i_r}$ of $\range{k}$, by the property of \Cref{proposition: (kr)-labelling}.
\proofcase{($\Leftarrow$)}
Suppose $\OVgraph{\mathcal{V}', r}$ has a $(k, r)$-"hyperclique" $S$.
By \Cref{proposition: (kr)-labelling},
$S$ consists of $v_1' \const{enc}_{i_1}, \dots, v_k' \const{enc}_{i_k}$ such that
$i_1, \dots, i_k$ are pairwise distinct.
Hence, $S$ consists of $v_1 \const{enc}_1, \dots, v_k \const{enc}_k$.
Then $v_1 \in \mathcal{V}_1, \dots, v_k \in \mathcal{V}_k$ form a $(k, r)$-"hyperclique" in $\OVgraph{\mathcal{V}_1,\dots, \mathcal{V}_k, r}$.
This completes the proof.
\end{proof}
Using this reduction, we have the following.
\begin{proposition}\label{proposition: partitioned OV-hyperclique hypothesis}
Fix $k > r \ge 3$.
The $(k, r)$-\PartitionedOVHyperCliqueHypothesis{} implies the $(k, r)$-\OVHyperCliqueHypothesis{}.
Similarly, the \PartitionedOVTriangleHypothesis{} implies the \OVTriangleHypothesis{}.
\end{proposition}
\begin{proof}
We show the contrapositive, that is, if
there is some $\eps > 0$ such that for all $c > 0$,
the $(k, r)$-\OVHyperCliqueProblem{n, \le c \log n} can be solved in $\asympO(n^{k - \eps})$ time, then
for all $c' > 0$,
the $(k, r)$-\PartitionedOVHyperCliqueProblem{n, \le c' \log n} can be solved in $\asympO(n^{k - \eps})$ time.
Let $\eps > 0$ witness this. We may assume $k - \eps > 1$.
Fix any $c' > 0$ and choose $c$ so that $c > c'$.
By $c > c'$, for all sufficiently large $n$, we have $c' \log n + r \binom{k}{r} \le c \log n$.
By \Cref{lemma: non-partitioned binary hyperclique hypothesis},
the $(k, r)$-\PartitionedOVHyperCliqueProblem{n, \le c' \log n} reduces to
the $(k, r)$-\OVHyperCliqueProblem{k n, \le c \log (k n)};
hence, by the assumption, it can be solved in
time $\asympO(n (c' \log n + r \binom{k}{r})) + \asympO((k n)^{k - \eps}) \le \asympO(n^{k - \eps})$.
The case of the \OVTriangleHypothesis{} is shown in the same way with $(k, r) = (3, 2)$, replacing the exponent $k - \eps$ with $\matrixexponent - \eps$;
note that \Cref{lemma: non-partitioned binary hyperclique hypothesis} works for $k = 3$ and $r = 2$ as well.
\end{proof}
\Cref{proposition: alman VVW} follows from \Cref{proposition: alman VVW',proposition: partitioned OV-hyperclique hypothesis}.

\clearpage
\mbox{}
\clearpage
\section*{Additional Appendices}
\section*{Appendices to \Cref{sec: preliminaries}}

\section{Proof of \Cref{proposition: isolated}}\label{section: proposition: isolated}

\begin{proposition*}[Restatement of \Cref{proposition: isolated}]
\propositionisolated
\end{proposition*}
\begin{proof}
Let $Z$ be the set of isolated "elements" in $\struc$,
"ie", the set of "elements" that occur in no tuple of any input "relation".
Choose a subset $Z_0 \subseteq Z$ of size $\min(k,\card{Z})$.
Let $\struc[2]$ be the substructure of $\struc$ induced by
\((\univ{\struc}\setminus Z)\cup Z_0\).
A straightforward induction on \(\fml\) then shows that this preserves the truth of $\fml$ ("cf" pebble games~\cite[\S 3.3.1]{ebbinghausFiniteModelTheory1995}).
This runs in $\asympO_{k}(\max(\fmlarity{\fml}, 1) \, \strucm{\struc})$ time using a standard coordinate compression (with an uninitialized array of size $\strucn{\struc}$).
\end{proof}

\clearpage
\section*{Appendices to \Cref{section: dichotomy sparse}}

\section{Explicit Proof of Theorem~\ref{theorem: dichotomy sparse}}\label{section: proof of theorem: dichotomy sparse}
In this section, we give an explicit proof of \Cref{theorem: dichotomy sparse}.
This is a direct consequence of combining the theorems presented in \Cref{section: dichotomy sparse}.

\begin{proof}[Proof of \Cref{theorem: dichotomy sparse}]
\proofcase{For Tractable Cases}
If $\sigarity{\sig}\le 1$, the claim follows from \Cref{theorem: monadic}, since $n \le \asympO(m)$ after removing isolated "elements" (\Cref{proposition: isolated}).
If $k=2$ and $\sigarity{\sig}\ge2$, it follows from \Cref{theorem: fo2}.
If $k=3$ and $\sigarity{\sig}\ge2$, it follows from \Cref{theorem: fo3}.
These give $m^{k-\eps}$ bounds with $\eps = k-1$ ($\ge 1$, as $k \ge 2$), $\eps=1$, and $\eps = 3-\matrixexponent$, respectively.
\proofcase{For Hard Cases}
It remains to consider $k\ge4$ and $\sigarity{\sig}\ge2$.
When $\sigarity{\sig} = 2$, it follows from \Cref{theorem:sparse-k4} (note that the \SparseMAXxSATHypothesis{3} implies the \SparseMAXxSATHypothesis{(k-1)}).
When $\sigarity{\sig} > 2$, it follows from \Cref{proposition: arity monotone} and the case of $\sigarity{\sig} = 2$.
\end{proof}

\section{Supplement to Footnote~\ref{footnote:k=1}: On the Case \texorpdfstring{$k = 1$}{k = 1}}\label{section: footnote:k=1}
When $k = 1$, except for the case where $\sigarity{\sig} = 1$ and $\card{\sig} = 1$, $\MC{\kPP{1}\sigres{\sig}}$ and $\MC{\kFO{1}\sigres{\sig}}$ have the same complexity.
If $\sigarity{\sig} = 0$, both are solvable in $\asympO(\len{\fml})$ time (\Cref{theorem: monadic}).
Otherwise, both are solvable in $2^{\asympO(\len{\fml})} \cdot m$ time (\Cref{theorem: monadic,lemma: projection}),
and this is optimal in $m$, as neither can be solved in $\asympO(m^{1 - \eps})$ time for any $\eps > 0$:
\begin{itemize}
    \item for the case $\sigarity{\sig} \ge 2$, by the diagonal query $\exists x.\ R(x, \dots, x)$,
    \item for the case $\sigarity{\sig} = 1$ and $\card{\sig} \ge 2$, by the array-intersection query $\exists x.\ U(x) \land U'(x)$.
\end{itemize}

When $\sigarity{\sig} = 1$ and $\card{\sig} = 1$, $\MC{\kPP{1}\sigres{\sig}}$ is solvable in $\asympO(\len{\fml})$ time by pushing quantifiers inwards and evaluating $\exists x. U(x)$ by checking whether $U$ is nonempty in $\asympO(1)$ time.
(Here, we recall that the input is given by a list of tuples, and hence the nonemptiness of $U$ can be checked in $\asympO(1)$ time by checking whether the list of tuples for $U$ is empty.)
However, the array-non-universal query $\exists x.\ \lnot U(x)$ cannot be solved in $\asympO(m^{1-\eps})$ time for any $\eps > 0$,
provided that the list of tuples for $U$ may contain duplicates or its length is not given explicitly.
 
\clearpage

\section{Monadic \texorpdfstring{$\kFO{k}$}{FOk} Model Checking via Monadic \texorpdfstring{$\kC{1}$}{C1}:
Proof of \texorpdfstring{\Cref{theorem: monadic}}{Theorem~\ref{theorem: monadic}}}
\label{section: monadic}

In this section, we present a $2^{\asympO(\len{\fml})}$-time translation (for fixed $k$) from $\kFO{k}\sigres{(\ordinalomega)}$
to "DAG-represented" $\kC{1}\sigres{(\ordinalomega)}$ (\Cref{lemma: monadic FOEk translation}),
by removing "equality" using "counting quantifiers".
Below, we assume $\sigarity{\sig} \le 1$.

\subsection{Monadic \texorpdfstring{$\kC{1}$}{C1}: definition and overview of translation}
We recall the one-variable counting logic $\kC{1}\sigres{(\ordinalomega)}$,
which is obtained by augmenting $\kFO{1}\sigres{(\ordinalomega)}$ with ""counting quantifiers"" ``$\exists^{\ge d}$'' and ``$\exists^{= d}$'' for integers $d \ge 0$;
see, "eg", \cite{immermanDescribingGraphsFirstOrder1990}.

\begin{definition}[$\kC{1}$ Formula]\label{definition: C}
\AP
We write $\intro*\kC{1}\sigres{\sig}$ (or simply $\kC{1}$) for the set of "formulas" (without "equality"), given by the following grammar:
\begin{align*}
\fml[1], \fml[2], \fml[3] \Coloneqq R(x)
\mid \fmltrue \mid \fmlfalse
\mid \fml[1] \land \fml[2]
\mid \fml[1] \lor \fml[2]
\mid \lnot \fml[1]
\mid \exists^{\ge d} x. \fml[1]
\mid \exists^{= d} x. \fml[1],
\end{align*}
where $d \ge 0$ is an integer and $R$ is a unary "relation symbol" in $\sig$.
The semantics of the existential \reintro*\kl{counting quantifiers} ``$\exists^{\ge d}$'' and ``$\exists^{= d}$'' are defined as follows:
\begin{align*}
\tuple{\struc, I} \models \exists^{\bigstar d} x. \fml[1] &\hspace{1.5em}\text{"iff"}\hspace{1.5em}
\bigl|\,\set{a \in \univ{\struc} \mid \tuple{\struc, I[x \mapsto a]} \models \fml[1]} \,\bigr| \mathrel{\bigstar} d \tag*{ for $\bigstar \in \{\ge, =\}$.}
\end{align*}
\end{definition}

$\MC{\kC{1}\sigres{\sig}}$ denotes the "model-checking problem" for $\kC{1}\sigres{\sig}$-"sentences":
given a (possibly "DAG-represented") $\kC{1}\sigres{\sig}$-"sentence" $\fml$ and a "structure" $\struc$, decide whether $\struc \models \fml$.
Because $\kC{1}$ is a one-variable logic, we can evaluate $\kC{1}$ sentences in linear time as follows.
\begin{proposition}\label{proposition: C1 eval}
$\MC{\kC{1}\sigres{(\ordinalomega)}}$ is solvable in $\asympO(\len{\fml}) \cdot n$ time.
\end{proposition}
\begin{proof}
Use the same straightforward bottom-up evaluation as for $\kFO{k}$; "cf" \cite[Proposition 3.1]{vardiComplexityBoundedvariableQueries1995}.
"Counting quantifiers" can be evaluated in $\asympO(n)$ time by a single pass over the domain.
\end{proof}
Based on the translation of \cite[Theorem 5]{mintsLogicalEquationsMonadic2009},
we can translate $\kFO{k}\sigres{(\ordinalomega)}$-"sentences" into "semantically equivalent" $\kC{1}\sigres{(\ordinalomega)}$-"sentences".
The essence is that "equality" can be eliminated using "counting quantifiers" by the following equivalence:
\[
\exists x.~ (y \fmlneq x \land \fml[3](x))
\quad \leftrightarrow \quad
(\exists^{\ge 2} x. \fml[3](x)) \lor (\lnot \fml[3](y) \land \exists x. \fml[3](x)).
\]
Observe that a naive translation based on this equivalence may blow up the "size" exponentially, even for "DAG-represented" "formulas", because $\fml[3]$ is duplicated as $\fml[3](x)$ and $\fml[3](y)$ in a single step.
Hence, a naive translation gives only an upper bound of $2^{\asympO(\len{\fml}^2)}$ time.
Below, we present an inductive translation using "types".

\begin{lemma}[\Cref{section: C1 DAG}]\label{lemma: monadic FOEk translation}
\gdef\lemmamonadicFOEktranslation{
Given an $\kFO{k}\sigres{(C)}$-"sentence" $\fml$,
we can construct, in time $\len{\fml} \cdot k^{\asympO(k)} \cdot 2^{\asympO(C k)}$, a ("DAG-represented") $\kC{1}\sigres{(C)}$-"sentence" $\fml[2]$
such that
\begin{itemize}
    \item for every structure $\struc$,
$\struc \models \fml$ "iff" $\struc \models \fml[2]$;
    \item the maximum number $d$ in "counting quantifiers" $\exists^{\ge d}$ or $\exists^{= d}$ appearing in $\fml[2]$ is at most $k$.
\end{itemize}
}
\lemmamonadicFOEktranslation
\end{lemma}
Combining \Cref{lemma: monadic FOEk translation,proposition: C1 eval} yields the following.
\begin{theorem} \label{theorem:formula-1single monadic model checking}
$\MC{\kFO{k}\sigres{(C)}}$
is solvable in $\len{\fml} \cdot k^{\asympO(k)} \cdot 2^{\asympO(C k)} \cdot n$ time.
If $k$ is fixed, then $\MC{\kFO{k}\sigres{(\ordinalomega)}}$ is solvable in $2^{\asympO(\len{\fml})} \cdot n$ time,
by applying the above bound with $C = \card{\fmlsig{\fml}} \le \len{\fml}$.
Moreover, if $k$ and $C$ are fixed,
then it is solvable in $\asympO_{k,C}(\len{\fml} \cdot n)$ time.
\end{theorem}
Consequently, the first part of \Cref{theorem: monadic} immediately follows from \Cref{theorem:formula-1single monadic model checking} when $k$ is fixed.
\begin{theorem*}[Restatement of {\Cref{theorem: monadic}} (first part)]
Fix $k \ge 1$.
\theoremmonadicFOEk
\end{theorem*}

\subsection{Translating into "DAG-represented" monadic \texorpdfstring{$\kC{1}$}{C1}: proof of \texorpdfstring{\Cref{lemma: monadic FOEk translation}}{Lemma~\ref{lemma: monadic FOEk translation}}}\label{section: C1 DAG}
In this subsection, we prove~\Cref{lemma: monadic FOEk translation} using "types".

Let $P_1, \dots, P_C$ be $C$ unary "relation symbols" and
let $x_1, x_2, \dots$ be a sequence of "variables";
we write $\intro*\Varset \defeq \set{ x_1, x_2, \ldots, x_k }$ for the set of the first $k$ "variables".

\subsubsection{Types and realizability: definition}
Given a "structure" $\struc$ and $a_1, \dots, a_\ell \in \univ{\struc}$ (where $\ell \ge 0$),
the (atomic) $\ell$-""type"" on $\struc$ with $a_1, \dots, a_\ell$,
denoted by $\intro*\type_{\struc}(a_1, \dots, a_\ell)$,
is the set of "atoms" or negated "atoms" (including "equalities"/"inequalities") over $\sig$ with free "variables" among $x_1, \dots, x_\ell$ that are satisfied by the map $x_i \mapsto a_i$ for $i \in \range{\ell}$.
More generally, a set $\beta$ is called an $\ell$-"type" if $\beta = \type_{\struc}(a_1, \dots, a_\ell)$ for some "structure" $\struc$ and $a_1, \dots, a_\ell \in \univ{\struc}$;
hence every $\ell$-"type" is consistent.
Below, quantifications over $\ell$-"types" range over all sets $\beta$ of this form.
We say that an $\ell$-"type" $\beta$ is ""realizable"" in a "structure" $\struc$ if $\beta = \type_{\struc}(a_1, \dots, a_\ell)$ for some $a_1, \dots, a_\ell \in \univ{\struc}$.

\begin{proposition}\label{proposition: invariant lemma in C1 existential}
Consider a unary "signature" and let $\ell \ge 0$.
Let $\struc$ be a "structure" and let $a_1, \dots, a_{\ell} \in \univ{\struc}$ and $b_1, \dots, b_{\ell} \in \univ{\struc}$.
Suppose $\type_{\struc}(a_1, \dots, a_{\ell}) = \type_{\struc}(b_1, \dots, b_{\ell})$.
For every $a_{\ell+1} \in \univ{\struc}$, there is $b_{\ell+1} \in \univ{\struc}$
such that
\[\type_{\struc}(a_1, \dots, a_{\ell}, a_{\ell+1}) = \type_{\struc}(b_1, \dots, b_{\ell}, b_{\ell+1}).\]
\end{proposition}
\begin{proof}
We define $b_{\ell+1}$ as follows.
\begin{itemize}
    \item If $a_{\ell+1} = a_j$ for some $j \in \range{\ell}$, let $b_{\ell+1} \defeq b_j$.
    \item Otherwise, let $\tau \defeq \type_{\struc}(a_{\ell+1})$ and let $N$ be the number of distinct "elements" of $1$-"type" $\tau$ among $a_1, \dots, a_\ell$.
    Since $\type_{\struc}(a_1, \dots, a_{\ell}) = \type_{\struc}(b_1, \dots, b_{\ell})$, $N$ is also the number of distinct "elements" of $1$-"type" $\tau$ among $b_1, \dots, b_\ell$.
    As $a_{\ell+1}$ is not among $a_1, \dots, a_\ell$, at least $N + 1$ "elements" of $\univ{\struc}$ have $1$-"type" $\tau$;
    let $b_{\ell+1}$ be one of them that is distinct from $b_1, \dots, b_\ell$.
\end{itemize}
In both cases, $\type_{\struc}(b_{\ell+1}) = \type_{\struc}(a_{\ell+1})$,
and $b_{\ell+1} = b_j$ if and only if $a_{\ell+1} = a_j$ for each $j \in \range{\ell}$.
Since the "signature" is unary, the (atomic) $(\ell+1)$-"type" of a tuple is determined by the $\ell$-"type" of its first $\ell$ entries,
the $1$-"type" of its last entry, and the "equalities" between its last entry and the others.
Hence, $\type_{\struc}(a_1, \dots, a_{\ell}, a_{\ell+1}) = \type_{\struc}(b_1, \dots, b_{\ell}, b_{\ell+1})$.
\end{proof}

\begin{proposition}\label{proposition: invariant lemma in C1 equivalence}
Let $\struc$ be a "structure".
Suppose that $\fml$ is an arbitrary $\kFO{k}\sigres{(\ordinalomega)}$-"formula" such that $\vars(\fml) \subseteq \set{x_1, \dots, x_{k}}$.
For every $a_1, \dots, a_{k}, b_1, \dots, b_{k} \in \univ{\struc}$ such that $\type_{\struc}(a_1, \dots, a_{k}) = \type_{\struc}(b_1, \dots, b_{k})$, we have:
\[
\tuple{\struc, \set{x_j \mapsto a_j \mid j \in \range{k}}} \models \fml
\quad\Longleftrightarrow\quad 
\tuple{\struc, \set{x_j \mapsto b_j \mid j \in \range{k}}} \models \fml
.\]
\end{proposition}
\begin{proof}
By induction on the structure of $\fml$.

\proofcase{Case $\fml = \fml[2] \land \fml[3]$, $\fml[2] \lor \fml[3]$, or $\lnot \fml[2]$}
Immediate from IH.

\proofcase{Case $\fml = \alpha$ where $\alpha$ is an "atom"}
Immediate from the definition of "types", since $\FV(\fml) \subseteq \set{x_1, \dots, x_{k}}$.

\proofcase{Case $\fml = \exists x_i. \fml[2]$}
\proofcase{($\Rightarrow$)} Suppose $\tuple{\struc, \set{x_j \mapsto a_j \mid j \in \range{k}}} \models \exists x_i. \fml[2]$.
There exists $a_i' \in \univ{\struc}$ such that $\tuple{\struc, \set{x_j \mapsto a_j \mid j \in \range{k} \setminus \set{i}} \cup \set{x_i \mapsto a_i'}} \models \fml[2]$.
Since $\type_{\struc}(a_1, \dots, a_{i-1}, a_{i+1}, \dots, a_{k}) = \type_{\struc}(b_1, \dots, b_{i-1}, b_{i+1}, \dots, b_{k})$,
by \Cref{proposition: invariant lemma in C1 existential} with $\ell = k-1$ (up to reordering the entries),
there exists $b_i' \in \univ{\struc}$ such that
\[\type_{\struc}(a_1, \dots, a_{i-1}, a_i', a_{i+1}, \dots, a_{k}) = \type_{\struc}(b_1, \dots, b_{i-1}, b_i', b_{i+1}, \dots, b_{k}).\]
By IH, $\tuple{\struc, \set{x_j \mapsto b_j \mid j \in \range{k} \setminus \set{i}} \cup \set{x_i \mapsto b_i'}} \models \fml[2]$.
Thus, $\tuple{\struc, \set{x_j \mapsto b_j \mid j \in \range{k}}} \models \exists x_i. \fml[2]$ by $x_i \not\in \FV(\exists x_i. \fml[2])$.
\proofcase{($\Leftarrow$)} The proof is symmetric to the above.

\proofcase{Case $\fml = \forall x_i. \fml[2]$}
Similar to the previous case, based on the equivalence to $\lnot \exists x_i. \lnot \fml[2]$.
\end{proof}

\subsubsection{Expressing realizability}
For a $k$-"type" $\beta$ and a subset $X \subseteq \Varset$, we denote by $\beta \restriction X$ the set of "formulas" in $\beta$ with free "variables" among $X$.
For $i \in \range{k}$, we denote by $\beta \restriction x_i$ the $1$-"type" obtained from $\beta \restriction \set{x_i}$ by renaming $x_i$ to $x_1$.
Observe that each $k$-"type" satisfies the following properties:
\begin{itemize}
    \item the binary relation $\sim_{\beta}$ on $\Varset$, defined by $x_i \sim_{\beta} x_j$ "iff" $(x_i = x_j) \in \beta$,
    is an "equivalence relation";
    \item if $x_i \sim_{\beta} x_j$,
    then $\alpha[x_i/z] \in \beta$ if and only if $\alpha[x_j/z] \in \beta$ for every "atom" $\alpha$ and "variable" $z$.
\end{itemize}
Note that if $x_i \sim_{\beta} x_j$, then $x_i$ and $x_j$ have the same $1$-"type".
We denote by $[\Varset]_{\sim_{\beta}}$ the set of "equivalence classes" of $\Varset$ under $\sim_{\beta}$.

For a $k$-"type" $\beta$ and a $1$-"type" $\tau$,
we denote by $\intro*\Nclass_{\beta}(\tau)$ the number of "equivalence classes" $X \in [\Varset]_{\sim_{\beta}}$ such that $(\beta \restriction x) = \tau$ for some (equivalently, every) $x \in X$.
We then define the $\kC{1}$-"sentence" $\Realizable_{\beta}$ as follows:
\phantomintro{\Realizable}%
\begin{align*}
&\reintro*\Realizable_{\beta} \defeq \bigwedge_{\tau: \text{$1$-"type" occurring in $\beta$}}\hspace{-.3em} \exists^{\ge \Nclass_{\beta}(\tau)} z.\ \bigwedge \tau(z).
\end{align*}
Here, $\tau(z)$ denotes $\tau$ with $x_1$ renamed to $z$, and the "equality" $z = z$ is omitted from $\bigwedge \tau(z)$, so that $\Realizable_{\beta}$ is a $\kC{1}$-"sentence" without "equality".
By construction,
$\Realizable_{\beta}$ expresses the "realizability@realizable" of "type" $\beta$, as follows:
\begin{proposition}
For every "structure" $\struc$ and $k$-"type" $\beta$,
$\struc \models \Realizable_{\beta}$ "iff" $\beta$ is "realizable" in $\struc$.
\end{proposition}
\begin{proof}
\proofcase{($\Rightarrow$)} Suppose $\struc \models \Realizable_{\beta}$.
For each $1$-"type" $\tau$ occurring in $\beta$,
there are at least $\Nclass_{\beta}(\tau)$ distinct "elements" of $1$-"type" $\tau$.
Thus, for each such $\tau$, we can assign distinct "elements" of $1$-"type" $\tau$ to the $\Nclass_{\beta}(\tau)$ "equivalence classes" $X \in [\Varset]_{\sim_{\beta}}$ with $(\beta \restriction x) = \tau$ for $x \in X$,
and we map each "variable" in $\Varset$ to the "element" assigned to its "equivalence class".
Since the "signature" is unary, a $k$-"type" is determined by its "equivalence relation" $\sim_{\beta}$ and the $1$-"types" of the "variables".
Hence, the resulting tuple realizes $\beta$.
\proofcase{($\Leftarrow$)} 
Suppose $\beta$ is "realizable" in $\struc$.
Then there are $a_1, \dots, a_k \in \univ{\struc}$ such that $\beta = \type_{\struc}(a_1, \dots, a_k)$.
Then for each $1$-"type" $\tau$ occurring in $\beta$, there are at least $\Nclass_{\beta}(\tau)$ "elements" in $\set{a_1, \dots, a_k}$ having the corresponding $1$-"type" $\tau$.
Thus, $\struc \models \Realizable_{\beta}$.
\end{proof}
\subsubsection{Translation} We now define the translation. 
\begin{definition}\label{definition: C1 translation}
    For an $\kFO{k}$-"formula" $\fml$ and a $k$-"type" $\beta$,
    the $\kC{1}$ \emph{"sentence"} $T_{\beta}(\fml)$ is defined as follows:
    \begin{align*}
        T_{\beta}(\alpha) &\defeq \begin{cases}
            \fmltrue  & \text{if $\alpha \in \beta$}, \\
            \fmlfalse & \text{otherwise},
        \end{cases} \text{ for "atoms" $\alpha$} 
        & T_{\beta}(\lnot \fml[1]) &\defeq \lnot T_{\beta}(\fml[1]) \\
        T_{\beta}(\fml[1] \land \fml[2]) &\defeq T_{\beta}(\fml[1]) \land T_{\beta}(\fml[2]) &
        T_{\beta}(\fml[1] \lor \fml[2]) &\defeq T_{\beta}(\fml[1]) \lor T_{\beta}(\fml[2]) \\
        T_{\beta}(\exists x_i. \fml[1]) &\defeq \hspace{-3.em}\bigvee_{\substack{\beta' \text{ a $k$-"type" "st"} \\ \beta' \restriction (\Varset \setminus \set{x_i}) = \beta \restriction (\Varset \setminus \set{x_i})}} \hspace{-3.em} \left(\Realizable_{\beta'} \land T_{\beta'}(\fml[1]) \right) &
        T_{\beta}(\forall x_i. \fml[1]) &\defeq \hspace{-3.em}\bigwedge_{\substack{\beta' \text{ a $k$-"type" "st"} \\ \beta' \restriction (\Varset \setminus \set{x_i}) = \beta \restriction (\Varset \setminus \set{x_i})}} \hspace{-3.em} \left(\lnot \Realizable_{\beta'} \lor T_{\beta'}(\fml[1]) \right).
    \end{align*}
\end{definition}

\begin{lemma}\label{lemma: invariant lemma in C1 reduction}
Let $\struc$ be a "structure".
Suppose that
$\fml$ is an arbitrary $\kFO{k}\sigres{(\ordinalomega)}$-"formula" such that $\vars(\fml) \subseteq \set{x_1, \dots, x_k}$.
For every $a_1, \dots, a_k \in \univ{\struc}$, we have:
\[\struc \models T_{\type_{\struc}(a_1, \dots, a_k)}(\fml) \quad\Longleftrightarrow\quad \tuple{\struc, \set{x_j \mapsto a_j \mid j \in \range{k}}} \models \fml.\]
\end{lemma}
\begin{proof}
By induction on the structure of $\fml$.

\proofcase{Case $\fml = \fml[2] \land \fml[3]$, $\fml[2] \lor \fml[3]$, or $\lnot \fml[2]$}
Immediate from IH.

\proofcase{Case $\fml = \alpha$ where $\alpha$ is an "atom"}
Immediate from the definition of $T_{\beta}(\alpha)$.

\proofcase{Case $\fml = \exists x_i. \fml[2]$}
($\Leftarrow$) Suppose $\tuple{\struc, \set{x_j \mapsto a_j \mid j \in \range{k}}} \models \exists x_i. \fml[2]$.
Then there is $a_i' \in \univ{\struc}$ such that
$\tuple{\struc, \set{x_j \mapsto a_j \mid j \in \range{k} \setminus \set{i}} \cup \set{x_i \mapsto a_i'}} \models \fml[2]$.
By IH, $\struc \models T_{\type_{\struc}(a_1, \dots, a_{i-1}, a_i', a_{i+1}, \dots, a_k)}(\fml[2])$.
Also, $\struc \models \Realizable_{\type_{\struc}(a_1, \dots, a_{i-1}, a_i', a_{i+1}, \dots, a_k)}$ because $\type_{\struc}(a_1, \dots, a_{i-1}, a_i', a_{i+1}, \dots, a_k)$ is "realizable" in $\struc$.
We thus have $\struc \models T_{\type_{\struc}(a_1, \dots, a_k)}(\exists x_i. \fml[2])$.

($\Rightarrow$) Suppose $\struc \models T_{\type_{\struc}(a_1, \dots, a_k)}(\exists x_i. \fml[2])$.
Then there is a $k$-"type" $\beta'$ such that
$\beta' \restriction (\Varset \setminus \set{x_i}) = \type_{\struc}(a_1, \dots, a_k) \restriction (\Varset \setminus \set{x_i})$, 
$\struc \models \Realizable_{\beta'}$, and
$\struc \models T_{\beta'}(\fml[2])$.
By $\struc \models \Realizable_{\beta'}$, there are $a_1', \dots, a_k' \in \univ{\struc}$ such that $\beta' = \type_{\struc}(a_1', \dots, a_k')$.
By IH, we have:
$\tuple{\struc, \set{x_j \mapsto a_j' \mid j \in \range{k}}} \models \fml[2]$.
Since
$\beta' \restriction (\Varset\setminus\{x_i\})
 = \type_{\struc}(a_1,\ldots,a_k)\restriction(\Varset\setminus\{x_i\})$,
and $\beta'$ is realizable, by \Cref{proposition: invariant lemma in C1 existential} (applied as in the proof of \Cref{proposition: invariant lemma in C1 equivalence}) there exists $b_i\in\univ{\struc}$ such that
\[
\type_{\struc}(a_1,\ldots,a_{i-1},b_i,a_{i+1},\ldots,a_k)=\beta'.
\]
By \Cref{proposition: invariant lemma in C1 equivalence}, the truth of
$\fml[2]$ is the same under this valuation as under any realization of
$\beta'$.
Thus, $\tuple{\struc, \set{x_j \mapsto a_j \mid j \in \range{k} \setminus \set{i}}\cup\set{x_i\mapsto b_i}}\models \fml[2]$, and hence
$\tuple{\struc, \set{x_j \mapsto a_j \mid j \in \range{k}}} \models \exists x_i. \fml[2]$.

\proofcase{Case $\fml = \forall x_i. \fml[2]$}
Similar to the previous case, based on the equivalence to $\lnot \exists x_i. \lnot \fml[2]$.
\end{proof}

\begin{proof}[Proof of \Cref{lemma: monadic FOEk translation}]
From $\fml$, define
\[
\fml[2] \defeq
\bigvee_{\beta \text{ an atomic $k$-"type"}}
\left(\Realizable_{\beta} \land T_{\beta}(\fml)\right).
\]
By \Cref{lemma: invariant lemma in C1 reduction}, each "realizable" disjunct has the same "truth value" as the original "sentence" $\fml$.
Hence, $\struc \models \fml[2]$ "iff" $\struc \models \fml$.
We construct the "DAG-represented" $\fml[2]$ by recursion on "subformulas" $\fml[3]$ of $\fml$.
We memoize "vertices" $v_{\fml[3],\beta}$ where $\fml[3]$ is a "subformula" and $\beta$ is a $k$-"type".
For the case $\fml[3] = \exists x_i. \fml[3]'$,
we first introduce shared group "vertices" $v'_{\fml[3]', \chi}$,
where $\chi = \beta' \restriction (\Varset \setminus \set{x_i})$ for some $k$-"type" $\beta'$.
Each $v'_{\fml[3]', \chi}$ is defined as the disjunction of $\Realizable_{\beta'}\land v_{\fml[3]', \beta'}$ over all $k$-"types" $\beta'$ such that $\beta' \restriction (\Varset \setminus \set{x_i}) = \chi$.
Then $v_{\fml[3], \beta}$ points to $v'_{\fml[3]', \beta \restriction (\Varset \setminus \set{x_i})}$.
The $\forall$ case is symmetric, with a conjunction.
Since these groups partition the set of $k$-"types",
the total number of child pointers introduced per quantified "subformula" is linear in the number of $k$-"types".

As each $k$-"type" is determined by an "equivalence relation" on $\Varset$ and $1$-"types" of $k$ "variables",
the number of $k$-"types" is bounded by
\[B_k \cdot (2^C)^k \le \asympO(k^k \cdot 2^{Ck}),\]
where $B_k$ denotes the $k$-th Bell number, which counts the number of "equivalence relations" on a set of size $k$.
$B_k$ is bounded by $k^k$, because each "equivalence relation" can be encoded injectively by a function $\range{k} \to \range{k}$.
Also, the number of $1$-"types" is equal to $2^C$.
Hence, the total time for constructing the "DAG" is $\asympO(\len{\fml} \cdot k^{k} \cdot 2^{C k}) + k^{k} \cdot 2^{C k} \cdot \mathrm{poly}(k, C) \le \len{\fml} \cdot k^{\asympO(k)} \cdot 2^{\asympO(C k)}$ (and hence, so is the "size" of $\fml[2]$ as a "DAG").
Additionally, the maximum number $d$ in "counting quantifiers" $\exists^{\ge d}$ or $\exists^{= d}$ appearing in $\fml[2]$ is at most $k$,
because $\Nclass_{\beta}(\tau) \le k$ for every $1$-"type" $\tau$.
Hence, this translation satisfies the desired properties.
\end{proof}

\subsection{On the case of \texorpdfstring{$C = 0$}{C=0}}
In the case of $C = 0$, we can obtain a better upper bound.

\begin{proposition}\label{proposition: C1 eval 0}
$\MC{\kC{1}\sigres{(0)}}$ is solvable in $\asympO(\max(d, 1) \, \len{\fml})$ time (independent of the "structure"),
where $d$ is the maximum number in "counting quantifiers" $\exists^{\ge d}$ or $\exists^{= d}$ appearing in the input "sentence" $\fml$.
\end{proposition}
\begin{proof}
If the input "structure" $\struc$ has at most $d$ "elements", then
we can evaluate $\fml$ in $\asympO(\max(d, 1) \, \len{\fml})$ time by \Cref{proposition: C1 eval}.
If $\struc$ has more than $d$ "elements", then
let $\struc[2]_{d+1}$ be the "structure" with universe $\set{1, \dots, d+1}$ and no "relations".
Since there are no "relation symbols" and no "equality", each "subformula" of $\fml$ holds either at every "element" or at none;
moreover, every number $d'$ in the "counting quantifiers" $\exists^{\ge d'}$ or $\exists^{= d'}$ appearing in $\fml$ satisfies $d' \le d$, which is smaller than both the number of "elements" of $\struc$ and $d + 1$.
Hence, by induction on "subformulas", we have: $\struc \models \fml$ "iff" $\struc[2]_{d+1} \models \fml$.
Hence, we can evaluate $\fml$ by evaluating $\fml$ on $\struc[2]_{d+1}$, which can be done in $\asympO(\max(d, 1) \, \len{\fml})$ time by \Cref{proposition: C1 eval}.
Hence, the total running time is $\asympO(\max(d, 1) \, \len{\fml})$.
\end{proof}

\begin{proof}[Proof of \Cref{theorem: monadic} (second part)]
By \Cref{lemma: monadic FOEk translation}, we can translate the input $\kFO{k}\sigres{(0)}$-"sentence" $\fml$ into a $\kC{1}\sigres{(0)}$-"sentence" $\fml[2]$ in $\len{\fml} \cdot k^{\asympO(k)} \le \asympO_{k}(\len{\fml})$ time.
By \Cref{proposition: C1 eval 0}, we can evaluate $\fml[2]$ in $\asympO(k \, \len{\fml[2]}) \le \asympO_k(\len{\fml})$ time.
Hence, the total running time is $\asympO_k(\len{\fml})$.
\end{proof}

\clearpage
\section*{Appendices to \Cref{section: reduction from OV Hyperclique to MCFO}}

\section{\texorpdfstring{\Cref{lemma: reduction monadic and binary predicates to single binary predicate}}{Lemma~\ref{lemma: reduction monadic and binary predicates to single binary predicate}} in the \texorpdfstring{\kl[tree-represented]{Tree-Represented}}{Tree-Represented} Case}\label{section: lemma: reduction monadic and binary predicates to single binary predicate for tree}
When the "formula" is "tree-represented",
the unary encoding of \Cref{lemma: unary encoding} requires $\asympO(C \, \len{\fml})$ time.
Below, we refine this complexity via a standard binary encoding of "relation symbols" by paths of length $\asympO(\log C)$.

\begin{proposition}\label{proposition: binary encoding}
Let $\fmlclass \in \set{\kPP{k}, \kEP{k}, \kFO{k} \mid k \ge 3}$.
Let $\fml[1]$ be an $\fmlclass\sigres{(0, C)}$-"sentence" without "equality" and $\struc[1]$ be a "structure".
Let $L \defeq \max(\lceil \log_2 C \rceil, 2)$.
One can construct, in $\asympO_k(\len{\fml[1]} + \strucn{\struc[1]} + L \, \strucm{\struc[1]})$ time,
an $\fmlclass\sigres{(0, 2)}$-"sentence" $\fml[2]$ without "equality" and a "structure" $\struc[2]$
such that
$\struc[1] \models \fml[1]$ "iff" $\struc[2] \models \fml[2]$.
Moreover, $\len{\fml[2]} = \asympO_k(\len{\fml[1]})$,
$\strucn{\struc[2]} = \asympO(\strucn{\struc[1]} + L \, \strucm{\struc[1]})$, and $\strucm{\struc[2]} = \asympO(\strucn{\struc[1]} + L \, \strucm{\struc[1]})$.
Additionally, if $\fml[1]$ is "tree-represented",
then the "tree-represented" $\fml[2]$ is constructed in $\asympO(L \, \len{\fml[1]} + \strucn{\struc[1]} + L \, \strucm{\struc[1]})$ time and has "size" $\len{\fml[2]} = \asympO(L \, \len{\fml[1]})$.
\end{proposition}
\begin{proof}
We first discard "relation symbols" not occurring in $\fml[1]$; after this preprocessing, $C \le \len{\fml[1]}$.
Below, $L$ denotes the value for this new $C$; it is at most the original value, and $L = \asympO(\len{\fml[1]})$.
Let $\fml \in \fmlclass\sigres{\set{R_0, \ldots, R_{C-1}}}$,
where $R_0, \dots, R_{C-1}$ are binary "relation symbols".
We construct $\struc[2]$ by extending $\struc[1]$ with dummy "elements" that encode the binary "relation symbols" $R_i$.
In the construction,
we put an $E_0$-labeled self-loop on "elements" originally in $\struc[1]$ to distinguish them from the dummy "elements".

Let $\intro*\binrep_L(i)$ denote the binary representation of $i$ of length $L$.
Let $E_0$ and $E_1$ be two fresh binary "relation symbols".
The "formula" $\intro*\pathfmltree_{s}(x, y)$ is defined as follows:
\begin{align*}
    \pathfmltree_{b}(x, y) &~\defeq~ E_b(x, y), &
    \pathfmltree_{bs}(x, y) &~\defeq~ \exists z.~ E_b(x, z) \land \pathfmltree_{s}(z, y),
\end{align*}
where $b \in \set{0,1}$, $s \in \set{0,1}^{\ge 1}$, and $z$ is a "variable" distinct from $x$ and $y$ (bound "variables" in recursive occurrences are renamed as in \Cref{lemma: unary encoding}).
Let $\fml[2]$ be the $\fmlclass\sigres{\set{E_0, E_1}}$-"sentence" obtained from $\fml[1]$ (without "equality") by replacing each "subformula" of the left-hand side form with the right-hand side "formula" as follows:
\begin{align*}
\exists x. \fml[3] &~\leadsto~ \exists x. E_0(x,x) \land \fml[3],&
\forall x. \fml[3] &~\leadsto~ \forall x. E_0(x,x) \to \fml[3],&
R_i(x, y) &~\leadsto~ \pathfmltree_{1\binrep_L(i)}(x, y).
\end{align*}
Since $\pathfmltree_{s}(x, y)$ is a $\kPP{3}$-"formula" and $k \geq 3$,
the "formula" $\fml[2]$ is in $\fmlclass\sigres{\set{E_0, E_1}}$.
Because $\len{\pathfmltree_{s}} = \asympO(L)$,
we have $\len{\fml[2]} = \asympO(L \, \len{\fml})$ for "tree-represented" "sentences".
For "DAG-represented" "sentences", the "formulas" $\pathfmltree_{s}$ can be shared along the trie of the reversed strings $1\binrep_L(i)$ ($i \in \range{0\,..\,C-1}$),
which has $\asympO(C)$ nodes, since for each $t$ there are at most $\min(2^t, C)$ distinct suffixes of length $t$ and $L = \max(\lceil \log_2 C \rceil, 2)$;
hence $\len{\fml[2]} = \asympO(k^2 C + \len{\fml}) = \asympO_k(\len{\fml})$.

Let $\struc[2] = \tuple{\univ{\struc[2]}, E_0^{\struc[2]}, E_1^{\struc[2]}}$ be the "structure" defined as follows:
\begin{itemize}
    \item $\univ{\struc[2]} \defeq \univ{\struc[1]} \sqcup \bigsqcup_{i = 0}^{C-1} \set{\tuple{v, w, i, \ell} \mid \tuple{v, w} \in R_i^{\struc[1]}, \ell \in \range{L}}$.
    For simplicity, we write $\tuple{v, w, i, 0}$ for $v$ and $\tuple{v, w, i, L+1}$ for $w$.
    \item $E_0$ and $E_1$ are defined as follows:
    \[
    \begin{aligned}
    E_b^{\struc[2]} \defeq{}&
    \set{\tuple{v,v} \mid v \in \univ{\struc[1]}, b=0} \sqcup {}\\
    &\set{\tuple{\tuple{v,w,i,\ell-1}, \tuple{v,w,i,\ell}}
    \mid \tuple{v,w} \in R_i^{\struc[1]},\
    \ell \in \range{L+1},\
    \text{the $\ell$-th bit of $1\binrep_L(i)$ is $b$}}.
    \end{aligned}
    \]
\end{itemize}

The following illustrates $\struc[2]$ for $C = 3$ (and $L = 2$):
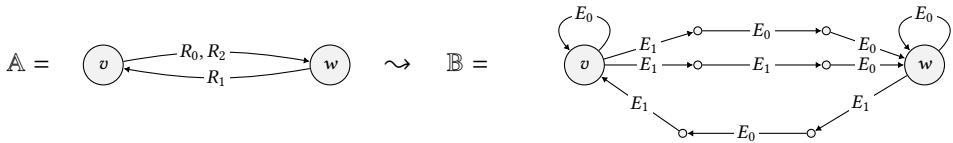
\knowledgeconfigure {quotation=false}
\begin{align*}
    & \struc =\quad \begin{tikzpicture}[baseline = -.5ex]
    \tikzstyle{vert}=[draw = black, circle, fill = gray!10, inner sep = .5pt, minimum size = 3pt, font = \scriptsize]
    \graph[grow right = 3.cm, branch down = 6.ex]{ 
    {s1/{$v$}[vert, inner sep = 1pt, minimum size = 15pt]}-!-
    {s2/{$w$}[vert, inner sep = 1pt, minimum size = 15pt]}
    };
    \path (s1) edge[->, bend left = 10, earrow] node[fill=white, inner sep=1pt,font= \scriptsize, pos = .5, xshift = -.5em]{$R_0, R_2$} (s2);
    \path (s2) edge[->, bend left = 10, earrow] node[fill=white, inner sep=1pt,font= \scriptsize, pos = .5]{$R_1$} (s1);
\end{tikzpicture}
\quad \leadsto\quad
\struc[2] =\quad
\begin{tikzpicture}[baseline = -3.5ex]
    \tikzstyle{vert}=[draw = black, circle, fill = gray!10, inner sep = .5pt, minimum size = 3pt, font = \scriptsize]
    \graph[grow right = 1.5cm, branch down = 3.ex]{
    {/, s10/{$v$}[vert,inner sep = 1pt, minimum size = 15pt],} -!-
    {s11/{}[vert], s121/{}[vert], /, s22/{}[vert, xshift = -.2cm]} -!-
    {s12/{}[vert, xshift = .2cm], s122/{}[vert, xshift = .2cm], /, s21/{}[vert]} -!-
    {/, s20/{$w$}[vert,inner sep = 1pt, minimum size = 15pt]}
    };
    \graph[use existing nodes, edges={color=black, pos = .5, earrow}, edge quotes={fill=white, inner sep=1pt,font= \scriptsize}]{
        {s10} ->["$E_0$", out = 45, in = 135, looseness =6] {s10};
        s10 ->["$E_1$"] s11 ->["$E_0$"] s12;
        s10 ->["$E_1$"] s121 ->["$E_1$"] s122;
        {s20} ->["$E_0$", out = 45, in = 135, looseness =6] {s20};
        s20 ->["$E_1$"] s21 ->["$E_0$"] s22;
        s12 ->["$E_0$"] s20;
        s122 ->["$E_0$"] s20;
        s22 ->["$E_1$"] s10;
    };
\end{tikzpicture}
\end{align*}
\knowledgeconfigure {quotation}

By construction, we have:
\[
\strucn{\struc[2]}
= \asympO(\strucn{\struc} + L \, \strucm{\struc}), \quad
\strucm{\struc[2]}
=\asympO(\strucn{\struc} + L \, \strucm{\struc}).
\]
The total construction time is
$\asympO(\len{\fml[2]} + \strucn{\struc[2]} + \strucm{\struc[2]})
\le \asympO_k(\len{\fml} + \strucn{\struc[1]} + L \, \strucm{\struc[1]})$ for "DAG-represented" "sentences"
and
$\asympO(L \, \len{\fml} + \strucn{\struc[1]} + L \, \strucm{\struc[1]})$ for "tree-represented" "sentences".

\begin{claim*}
For every $i \in \range{0\,..\,C-1}$, $u, v \in \univ{\struc}$,
we have:
$\struc[2] \models \pathfmltree_{1 \binrep_L(i)}(u,v) \iff \struc[1] \models R_{i}(u,v)$.
\end{claim*}
\begin{proof}
\proofcase{($\Longleftarrow$)}
This follows directly from the definition of $\pathfmltree_{1 \binrep_L(i)}$.
\proofcase{($\Longrightarrow$)}
Unfold the definition of $\pathfmltree_{1 \binrep_L(i)}(u,v)$.
Every satisfying path must start with the leading \(1\)-edge and then follow
the unique bit pattern \(\binrep_L(i)\). By construction, the only
paths with this label from an original "element" \(u\) to an original "element"
\(v\) are precisely the paths introduced for tuples \((u,v)\in R_i^{\struc[1]}\).
Hence \(\struc[1] \models R_i(u,v)\).
\end{proof}
Thus, by a straightforward induction on the structure of "subformulas",
for "valuations" of "variables" to $\univ{\struc}$,
we have:
$\struc \models \fml$ "iff" $\struc[2] \models \fml[2]$.
This gives the desired construction.
\end{proof}

\begin{lemma}[\Cref{lemma: reduction monadic and binary predicates to single binary predicate} for "tree-represented" "formulas"]\label{lemma: reduction monadic and binary predicates to single binary predicate for tree}
Fix $k \ge 3$.
Let $\fml$ be a "tree-represented" $\kPP{k}\sigres{(C_1,C_2)}$-"sentence" ("resp", $\kEP{k}\sigres{(C_1,C_2)}$-"sentence", $\kFO{k}\sigres{(C_1,C_2)}$-"sentence"), and let $\struc$ be a "structure".
Let $C \defeq \card{\fmlsig{\fml}} + 1$ where the extra binary relation is used
to eliminate "equality".
Let $L \defeq \max(\lceil \log_2 C \rceil, 2)$.
One can construct, in $\asympO(L \, (\len{\fml} + \strucm{\struc}))$ time, a "tree-represented" $\kPP{k}\sigres{(0,1)}$-"sentence" ("resp", $\kEP{k}\sigres{(0,1)}$-"sentence", $\kFO{k}\sigres{(0,1)}$-"sentence") $\fml[2]$ without "equality" and a "structure" $\struc[2]$ such that
$\struc \models \fml$ "iff" $\struc[2] \models \fml[2]$.
Moreover,
$\len{\fml[2]} = \asympO(L \, \len{\fml})$,
$\strucn{\struc[2]} = \asympO(\strucn{\struc} + L \, \strucm{\struc})$, and
$\strucm{\struc[2]} = \asympO(L \, \strucm{\struc})$.
\end{lemma}
\begin{proof}
Let $\widehat{\struc}$ be the "structure" obtained from $\struc$ by removing isolated "elements" (\Cref{proposition: isolated}).
Then $\strucn{\widehat{\struc}} = \asympO(\strucm{\widehat{\struc}}) \le \asympO(\strucm{\struc})$.
Let $\fml'$ and $\struc'$ be obtained from $\fml$ and $\widehat{\struc}$ by encoding unary "relations" as binary ones and eliminating "equality" (\Cref{proposition: arity monotone,proposition: equality elimination}).
Then $\len{\fml'} = \asympO(\len{\fml})$, $\strucn{\struc'} = \asympO(\strucn{\widehat{\struc}}) \le \asympO(\strucm{\struc})$, and $\strucm{\struc'} = \asympO(\strucm{\widehat{\struc}}) \le \asympO(\strucm{\struc})$.
Let $\fml[2]$ and $\struc[2]$ be obtained from $\fml'$ and $\struc'$ by applying \Cref{proposition: binary encoding} and \Cref{lemma: unary encoding}.
We then have:
\begin{align*}
\len{\fml[2]} &= \asympO(L \, \len{\fml'}) = \asympO(L \, \len{\fml}),\\
\strucn{\struc[2]} &= \asympO(\strucn{\struc'} + L \, \strucm{\struc'}) = \asympO(L \, \strucm{\struc}),&
\strucm{\struc[2]} &= \asympO(\strucn{\struc'} + L \, \strucm{\struc'}) = \asympO(L \, \strucm{\struc}).
\end{align*}
The construction takes $\asympO(\strucm{\struc}) + \asympO(\len{\fml} + \strucm{\widehat{\struc}}) + \asympO(L \, \len{\fml'} + \strucn{\struc'} + L \, \strucm{\struc'}) + \asympO(L \, (\len{\fml} + \strucm{\struc}))
= \asympO(L \, (\len{\fml} + \strucm{\struc}))$ time.
As the truth of $\fml$ on $\struc$ is preserved through the reductions, this completes the proof.
\end{proof}
\section{Proof of \texorpdfstring{\Cref{lemma:sparse-k3'}}{Lemma~\ref{lemma:sparse-k3'}}}\label{section:lemma:sparse-k3'}

\begin{lemma}["Cf" \Cref{lemma:sparse-k4}]\label{lemma:sparse-k3}
Fix $\eps > 0$.
Assuming the \OVTriangleHypothesis{},
there is a $\kPP{3}\sigres{(\ordinalomega, 2)}$-"sentence" $\fml_{\eps}$ such that
$\MC{\set{\fml_{\eps}}}$ cannot be solved in $\asympO(m^{\matrixexponent - \eps})$ time.
\end{lemma}
\begin{proof}
This follows the same argument as \Cref{lemma:sparse-k4};
note that \Cref{lemma: new new partitioned binary hyperclique hypothesis} works for $k = 3$ as well.
We give the details for completeness.

As in the proof of \Cref{lemma:sparse-k4}, we may assume $\eps < 1$.
Let $\delta \defeq \eps/2$ and let $c$ be the constant supplied by the
\OVTriangleHypothesis{} for the exponent loss $\delta$.
Choose $B$ so large that $\gamma \defeq c/B \cdot \log 2 < \delta/\matrixexponent$, or equivalently, $\matrixexponent\gamma < \delta$.
Note that $B$ is chosen after $\eps$ and after the constant $c$ (here, $c$ is determined by $\eps$).
Let $\fml_{\eps}$ be the fixed "sentence" produced by \Cref{lemma: new new partitioned binary hyperclique hypothesis} for $k = 3$, $r = 2$, and the parameter $B$.
Suppose, toward a contradiction, that $\MC{\set{\fml_{\eps}}}$ can be solved in $\asympO(m^{\matrixexponent - \eps})$ time.
Let $\mathcal{V} \subseteq \set{0, 1}^d$ be an arbitrary instance of \OVTriangleProblem{n, \le c \log n}.
Without loss of generality, we may assume that $B$ divides $d$,
by padding each "vector" with at most $B-1$ zeros.
This padding does not change the "generalized orthogonality@generalized-orthogonal" of any tuple of "vectors", so the instance stays "well-formed" and the answer is unchanged.
By $d \le c \log n + B - 1$, we have:
\[
2^{{d/B}}
~\le~ 2^{(\frac{c}{B} \cdot \log n) + 1}
~=~ 2 \cdot n^{\frac{c}{B} \cdot \log 2}
~=~ \asympO(n^{\gamma}).
\]
Let $\struc \defeq \struc_{\mathcal{V}, 2}$ be given by \Cref{lemma: new new partitioned binary hyperclique hypothesis}.
By $2^{{d/B}} = \asympO(n^{\gamma})$ and $\gamma < \delta/\matrixexponent < 1$, we have:
\begin{align*}
\strucn{\struc} &= \asympO(n^{\gamma} + n) = \asympO(n), &\strucm{\struc} &= \asympO(n^{1+\gamma}), & \text{and the construction time is } \asympO(n^{1 + \gamma}).
\end{align*}
By applying the assumed $\asympO(m^{\matrixexponent - \eps})$-time algorithm,
the \OVTriangleProblem{n, \le c \log n} can be solved in
$\asympO(n^{1+\gamma}) + \asympO(n^{(1+\gamma)(\matrixexponent - \eps)}) \le \asympO(n^{\matrixexponent - \delta})$ time, where
$(1+\gamma)(\matrixexponent-\eps)
< \matrixexponent-\eps + \matrixexponent\gamma
< \matrixexponent-\eps + \delta
= \matrixexponent-\delta$
by $\matrixexponent\,\gamma <\delta$ and $\eps = 2 \delta$.
This contradicts the choice of $c$.
Hence, $\MC{\set{\fml_{\eps}}}$ cannot be solved in $\asympO(m^{\matrixexponent - \eps})$ time.
\end{proof}

\begin{lemma*}[Restatement of {\Cref{lemma:sparse-k3'}}]
\sparsekthreeprimelemma
\end{lemma*}
\begin{proof}
This follows by the same argument as in the proof of \Cref{lemma:sparse-k4'}, except that we use \Cref{lemma:sparse-k3} instead of \Cref{lemma:sparse-k4}.

As above, we may assume $\eps < 1$.
Note that $\matrixexponent - \eps > 1$, as $\matrixexponent \ge 2$ and $\eps < 1$.
Let $\fml_{\eps}$ be the $\kPP{3}\sigres{(\ordinalomega, 2)}$-"sentence" given by \Cref{lemma:sparse-k3} for the parameter $\eps$.
Let $\fml[2]_{\eps}$ be the "sentence" obtained from $\fml_{\eps}$ by \Cref{lemma: reduction monadic and binary predicates to single binary predicate}; it depends only on $\eps$.
Let $\struc$ be an arbitrary "structure".
Let $\widehat{\struc}$ be the "structure" obtained from $\struc$ by applying \Cref{proposition: isolated}.
Then $\strucn{\widehat{\struc}} = \asympO_{\eps}(\strucm{\struc})$ and $\strucm{\widehat{\struc}} = \asympO(\strucm{\struc})$.
Let $\struc[2]$ be the "structure" obtained from $\widehat{\struc}$ by applying \Cref{lemma: reduction monadic and binary predicates to single binary predicate} to $\fml_{\eps}$ and $\widehat{\struc}$.
Then $\fml[2]_{\eps}$ is a $\kPP{3}\sigres{(0,1)}$-"sentence" without "equality",
$\len{\fml[2]_{\eps}} = \asympO(\len{\fml_{\eps}})$, $\strucn{\struc[2]} = \asympO_{\eps}(\strucn{\widehat{\struc}}) = \asympO_{\eps}(\strucm{\struc})$, and $\strucm{\struc[2]} = \asympO_{\eps}(\strucn{\widehat{\struc}} + \strucm{\widehat{\struc}}) = \asympO_{\eps}(\strucm{\struc})$.
Toward a contradiction, we assume that $\MC{\set{\fml[2]_{\eps}}}$ is solvable in $\asympO(m^{\matrixexponent - \eps})$ time.
Since $\fml_{\eps}$ is fixed, the construction time is $\asympO_{\eps}(\strucm{\struc})$.
Then, by constructing $\struc[2]$ from $\struc$ and running the assumed algorithm on $\struc[2]$, $\MC{\set{\fml_{\eps}}}$ can be solved in
\[
\asympO_{\eps}(\strucm{\struc}) + \asympO(\strucm{\struc[2]}^{\matrixexponent-\eps})
\le \asympO_{\eps}(\strucm{\struc}^{\matrixexponent-\eps})
\]
time,
which contradicts \Cref{lemma:sparse-k3}.
\end{proof}
 
\clearpage
\section*{Appendices to \Cref{section: fo3}}

\section{Example of the \texorpdfstring{$\kFO{3}$}{FO3}-to-\texorpdfstring{$\CoR$}{CoR} Translation}\label{section: example: proposition: fo3 to atomic CoR}
\begin{example}
Here are additional examples, where the input "formula" has one or no free "variables".
\begin{align*}
\fml(y) \equiv \exists x. P(x, y)
&\quad\leadsto\quad \fml[2](y, y) \equiv (\cortop \compo P)(y, y),\\
\fml(y) \equiv \exists x. P(x, x) \land Q(x, y)
&\quad \leadsto\quad \fml[2](y, y) \equiv \left(\cortop \compo ((\corid \intersection P) \compo Q) \right) (y, y),\\
\fml \equiv \exists x. P(x, x)
&\quad\leadsto\quad \fml[2](v, v) \equiv (\cortop \compo (\corid \intersection P) \compo \cortop)(v, v).
\end{align*}
The input "formula" and the translated $\CoR$ "atomic FO-formula@@CoR" may have different free "variables",
since $\CoR$ "atomic FO-formulas@@CoR" have at least one free "variable".
The quantifier elimination \eqref{QE} shows that
we can eliminate 
the "quantifier" $\exists$ from $\exists x. \fml[1]_1 \land \fml[1]_2$ on the $\CoR$ side.
Recall that $\exists x. \fml[1]_1 \lor \fml[1]_2$ can be distributed to $(\exists x. \fml[1]_1) \lor (\exists x. \fml[1]_2)$.
We can translate more complex "formulas" by first translating them to \phantomintro{DNF}"DNF" and then pushing "quantifiers" inward.
For instance, we can translate the following "formula" $\fml(x, z)$ into a $\CoR$ "atomic FO-formula@@CoR" $\fml[2]$, as follows:
\begin{align*}
& \fml(x, z) \equiv (\exists y. (P(x, y) \lor Q(y, z)) \land R(x, y)) \land S(x, z)\\
& \leadsto (\exists y. P(x, y) \land R(x, y) \land S(x, z)) \lor (\exists y. Q(y, z) \land R(x, y) \land S(x, z)) \tag{to "DNF"}\\
& \leadsto ((\exists y. P(x, y) \land R(x, y)) \land S(x, z)) \lor ((\exists y. Q(y, z) \land R(x, y)) \land S(x, z)) \tag{push $\exists$ inward}\\
& \leadsto (((P \intersection R) \compo \cortop)(x, z) \land S(x, z)) \lor ((R \compo Q)(x, z) \land S(x, z)) \tag{to $\CoR$ (intermediate)}\\
& \leadsto ((((P \intersection R) \compo \cortop) \intersection S) \union ((R \compo Q) \intersection S))(x, z) \equiv \fml[2](x, z). \tag{to $\CoR$}
\end{align*}
\end{example}

\section{Proof of \texorpdfstring{\Cref{lemma: guarded-pairing}}{Lemma~\ref{lemma: guarded-pairing}}: \texorpdfstring{\kl[guarded-pairing]{Guarded-Pairing}}{Guarded-Pairing}}\label{section: lemma: guarded-pairing}
In this section, we prove the following.
\begin{lemma*}[Restatement of \Cref{lemma: guarded-pairing}]
    \guardedpairinglemma
\end{lemma*}

For a "structure" $\struc$ over a "signature" $\sig$,
let $\intro*\Gaifman^{\struc} \subseteq \univ{\struc}^2$
be the binary relation of the ""Gaifman graph"" of $\struc$:
$\tuple{w, w'} \in \Gaifman^{\struc}$
"iff"
$w = w'$, or
there are $R \in \sig$, $\tuple{v_1, \dots, v_{\arity{R}}} \in R^{\struc}$, and $i, j \in \range{\arity{R}}$ such that $v_i = w$ and $v_j = w'$.
A pair $\tuple{v, w} \in \univ{\struc}^2$ is ""guarded"" if $\tuple{v, w} \in \Gaifman^{\struc}$.
Given a "valuation" $I$,
a pair $\tuple{x, y}$ of "variables" is \reintro*\kl{guarded} on $I$ if the pair $\tuple{I(x), I(y)}$ is "guarded".
Observe $\card{\Gaifman^{\struc}} = \asympO(\strucn{\struc} + \strucm{\struc})$, which is $\asympO(\strucm{\struc})$ after applying \Cref{proposition: isolated}.
We encode each "guarded" pair as a single "variable" by using a newly added "vertex" for a tuple of $R$, together with two indices that record the pair.

\begin{proof}
By \Cref{proposition: isolated}, we first replace $\struc[1]$ by an
equivalent "structure", still denoted by $\struc[1]$, in $\asympO(\strucm{\struc[1]})$ time, so that $\strucn{\struc[1]} \le \asympO(\strucm{\struc[1]})$.
Before the construction, by applying \Cref{proposition: equality elimination},
we eliminate "equality" with a fresh binary "relation symbol" $\intro*\EQ$, where $\EQ^{\struc[1]}$ is the diagonal relation.
Then $\fml$ is an $\kFO{3}\sigres{(C_1, C_2 + 1, C_3)}$-"sentence" without "equality".
Let $\intro*\sigEQ \defeq \fmlsig{\fml} \cup \set{\EQ}$ and regard $\struc[1]$ as a "structure" over $\sigEQ$ (discarding the "relations" not in $\sigEQ$);
in particular, $\Gaifman^{\struc[1]}$ below is taken with respect to $\sigEQ$.
Without loss of generality, we may assume that $C_i$ is exactly the number of "relation symbols" of "arity" $i$ in the input "sentence" (decreasing $C_i$ only strengthens the claim); in particular, $C_1 + C_2 + C_3 \le \len{\fml}$.
Then $\sigEQ$ has exactly $C_1$, $C_2 + 1$, and $C_3$ "relation symbols" of "arity" $1$, $2$, and $3$, respectively.

Given the "structure" $\struc[1]$ over $\sigEQ$,
we construct the "structure" $\struc[2]$ as follows:
\begin{itemize}
\item
$\univ{\struc[2]} \defeq \univ{\struc[1]} \sqcup \set{w_{R v_1, \dots, v_{\ell}} \mid R \in \sigEQ, \tuple{v_1, \dots, v_{\ell}} \in R^{\struc[1]}}$.
(Namely, for each $R \in \sigEQ$ and each tuple $\tuple{v_1, \dots, v_{\ell}} \in R^{\struc[1]}$,
a new "vertex" $w_{R v_1, \dots, v_{\ell}}$ is added.)
\item
$(R^{[i]})^{\struc[2]} \defeq \set*{ \tuple{w_{R v_1, \dots, v_{\ell}}, v_i} \;\middle|\; 
\tuple{v_1, \dots, v_{\ell}} \in R^{\struc[1]}}$ for each $R \in \sigEQ$ and $i \in \range{\arity{R}}$.
\end{itemize}
We also define the following "formulas":
\begin{itemize}
    \item $\intro*\Proj_i(x, y) \defeq \bigvee_{R \in \sigEQ ; \arity{R} \ge i} R^{[i]}(x, y)$
    for each $i \in \range{3}$.
    The "formulas" $\Proj_i(x, y)$ are shared as "DAGs"; hence their total
contribution to the translated "formula" is $\asympO(1 + \sum_{i=1}^3 i C_i) = \asympO(\len{\fml})$.

    \item $\intro*\IsV(x) \defeq \exists z. \bigvee_{i = 1}^3 \Proj_i(z, x)$, where $z$ is a "variable" distinct from $x$.
    Observe $\struc[2], I \models \IsV(x)$ "iff" $I(x) \in \univ{\struc[1]}$, because $\EQ \in \sigEQ$ yields the "vertex" $w_{\EQ a a}$ with $\tuple{w_{\EQ a a}, a} \in (\EQ^{[1]})^{\struc[2]}$ for every $a \in \univ{\struc[1]}$.

    \item $\intro*\Gdd(x, y) \defeq \exists z. \bigvee_{i, j \in \range{3}} \left(\Proj_i(z, x) \land \Proj_j(z, y)\right)$, where $z$ is a "variable" distinct from $x$ and $y$.
    Observe $\struc[2], I \models \Gdd(x, y)$ "iff" $\tuple{I(x), I(y)} \in \Gaifman^{\struc[1]}$.
\end{itemize}
Let $x_1, x_2, x_3$ be the distinct "variables" covering those of $\fml$.
We consider two types of $\breve{p}$:
the empty set $\emptyset$ or a pair of positions $\tuple{j_1, j_2} \in \range{3}^2$ so that the following invariants hold:
\begin{itemize}
    \item In $\Lambda_{\emptyset}$, 
    $x_1$ and $x_2$ directly denote original "vertices" and are required not to be "guarded"; 

    \item In $\Lambda_{\tuple{j_1, j_2}}$, 
    the translated "variable" $x_1$ denotes a tuple-"vertex" $w_{R v_1, \dots, v_{\ell}}$
    whose $j_1$-projection $v_{j_1}$ and $j_2$-projection $v_{j_2}$ are the original "vertices" of $x_1$ and $x_2$.
    Also, in the translated "formula", $x_2$ is not a free "variable".
\end{itemize}
We now define the map $\Lambda_{\breve{p}}$ from $\kFO{3}\sigres{(C_1, C_2 + 1, C_3)}$-"formulas" to $\kFO{3}\sigres{(0, 1+\sum_{i=1}^{3} i C_i)}$-"formulas"%
\footnote{While $\EQ \in \sigEQ$ yields the two "relation symbols" $\EQ^{[1]}$ and $\EQ^{[2]}$,
we can identify $\EQ^{[1]}$ and $\EQ^{[2]}$ as a single "relation symbol"
because $(\EQ^{[1]})^{\struc[2]} = (\EQ^{[2]})^{\struc[2]}$ since $\EQ^{\struc[1]}$ is the diagonal relation.
This accounts for the additional $1$ in $\sigres{(0,1+\sum_{i=1}^3 iC_i)}$.}%
, as follows:
\begin{align*}
\Lambda_{\breve{p}}(\lnot \fml)
\defeq{}& \lnot \Lambda_{\breve{p}}(\fml),&
\hspace{-.5em}\Lambda_{\breve{p}}(\fml[1] \land \fml[2])
\defeq{}& \Lambda_{\breve{p}}(\fml[1]) \land \Lambda_{\breve{p}}(\fml[2]), &
\hspace{-.5em}\Lambda_{\breve{p}}(\fml[1] \lor \fml[2])
\defeq \Lambda_{\breve{p}}(\fml[1]) \lor \Lambda_{\breve{p}}(\fml[2]),\\
\Lambda_{\emptyset}(R(y_1, \dots, y_{\ell}))
\defeq{}& \begin{cases}
\fmlfalse & \text{if $x_1, x_2 \in \set{y_1, \dots, y_{\ell}}$}\\
\exists z. \bigwedge_{i = 1}^{\ell}  R^{[i]}(z, y_i) & \text{otherwise}
\end{cases} \tag*{(where $z \not\in \set{y_1, \dots, y_{\ell}}$),}\span\span\span\\
\Lambda_{\tuple{j_1, j_2}}(R(y_1, \dots, y_{\ell}))
\defeq{}& 
\exists x_2. \bigwedge_{i = 1}^{\ell} \begin{cases}
    \exists z. (\Proj_{j_q}(x_1, z) \land R^{[i]}(x_2, z)) & \text{if $y_i = x_q$ for $q \in \range{2}$}\\
    R^{[i]}(x_2, y_i) & \text{otherwise}
\end{cases} \tag*{(where $z \not\in \set{x_1, x_2}$),} \span\span\span\\
\Lambda_{\breve{p}}(\exists y. \fml[1])
\defeq{}& \exists y. (\IsV(y) \land \Lambda_{\breve{p}}(\fml[1])) \text{ if $y \not\in \set{x_1, x_2}$}, \span\span\span\\
\Lambda_{\emptyset}(\exists x_1. \fml[1])
\defeq{}&
\left( 
\exists x_1.~ \bigvee_{j_1, j_2}
\Bigl(
    (\exists z. \Proj_{j_1}(x_1, z)) \land (\Proj_{j_2}(x_1, x_2)) \land \Lambda_{\tuple{j_1,j_2}}(\fml[1]) \Bigr)
\right) \lor {} \span\span\span\\
&\exists x_1. (\IsV(x_1) \land \lnot \Gdd(x_1, x_2) \land \Lambda_{\emptyset}(\fml[1])) \tag*{(where $z \not\in \set{x_1, x_2}$),} \span\span\span\\
\Lambda_{\tuple{j_1, j_2}}(\exists x_1. \fml[1])
\defeq{}&
\exists x_2. (\Proj_{j_2}(x_1, x_2) \land \Lambda_{\emptyset}(\exists x_1. \fml[1])), \span\span\span\\
\Lambda_{\emptyset}(\exists x_2. \fml[1])
\defeq{}&
\exists x_2.~ \Bigl( x_2 = x_1 \land \exists x_1.~ \bigvee_{j_1, j_2} \bigl( \Proj_{j_1}(x_1, x_2) \land (\exists z. \Proj_{j_2}(x_1, z)) \land \Lambda_{\tuple{j_1,j_2}}(\fml[1]) \bigr) \Bigr) \lor {}\hspace{-11.2pt} \span\span\span\\
&\exists x_2. (\IsV(x_2) \land \lnot \Gdd(x_1, x_2) \land \Lambda_{\emptyset}(\fml[1])) \tag*{(where $z \not\in \set{x_1, x_2}$),} \span\span\span\\
\Lambda_{\tuple{j_1, j_2}}(\exists x_2. \fml[1])
\defeq{}&
\exists x_2. (x_2 = x_1 \land \exists x_1. (\Proj_{j_1}(x_2, x_1) \land \Lambda_{\emptyset}(\exists x_2. \fml[1]))),\span\span\span\\
\Lambda_{\breve{p}}(\forall y. \fml[1])
\defeq{}&
\lnot \Lambda_{\breve{p}}(\exists y. \lnot \fml[1]). \span\span\span
\end{align*}
The translated "formula" is constructed as a "DAG".
For every "subformula" $\fml[2]$ of $\fml$ and every $\breve p \in \set{\emptyset} \cup \range{3}^{2}$,
we create at most one "vertex" representing $\Lambda_{\breve p}(\fml[2])$ and reuse it.
Thus the disjunctions over $\range{3}^{2}$ in the quantifier clauses contribute only a factor $\asympO(3^{2}) = \asympO(1)$, and the
auxiliary "formulas" $\Proj_i$, $\IsV$, $\Gdd$ are shared.
Consequently, $\len{\Lambda_{\breve p}(\fml)} = \asympO(\len{\fml})$ for every $\breve p$.

We prove the following invariant by induction on $\fml$: for every "valuation" $I$ on $\struc$,
\begin{itemize}
    \item if the pair $\tuple{x_1, x_2}$ is not "guarded" on $I$, then
    $\struc, I \models \fml$ "iff" $\struc[2], I \models \Lambda_{\emptyset}(\fml)$;
    \item for every $w \in \univ{\struc[2]}$ and $j_1, j_2 \in \range{3}$,
    if 
    $\tuple{w, I(x_1)} \in \sem{\Proj_{j_1}}^{\struc[2]}$ and 
    $\tuple{w, I(x_2)} \in \sem{\Proj_{j_2}}^{\struc[2]}$,
    then
    $\struc, I \models \fml$ "iff" $\struc[2], I[w/x_1] \models \Lambda_{\tuple{j_1, j_2}}(\fml)$.
\end{itemize}

\proofcase{Case $\Lambda_{\breve{p}}(\lnot \fml)$, $\Lambda_{\breve{p}}(\fml[1] \land \fml[2])$, $\Lambda_{\breve{p}}(\fml[1] \lor \fml[2])$}
Immediate from IH.

\proofcase{Case $\Lambda_{\emptyset}(R(y_1, \dots, y_{\ell}))$}
If $x_1, x_2 \in \set{y_1, \dots, y_{\ell}}$,
then $\struc[1], I \models R(y_1, \dots, y_{\ell})$
"iff" $\struc[1], I \models \fmlfalse$ (because $\tuple{x_1, x_2}$ is not "guarded" on $I$)
"iff" false
"iff" $\struc[2], I \models \fmlfalse$.
Otherwise, we have $\card{\set{y_1, \dots, y_{\ell}}} < 3$.
By the construction of $\struc[2]$, we have:
$\struc[1], I \models R(y_1, \dots, y_{\ell})$ "iff"
$\struc[2], I \models \exists z. \bigwedge_{i = 1}^{\ell}  R^{[i]}(z, y_i)$.

\proofcase{Case $\Lambda_{\tuple{j_1, j_2}}(R(y_1, \dots, y_{\ell}))$}
Assume that $w$ encodes the guarded pair
$\tuple{I(x_1),I(x_2)}$ in positions $\tuple{j_1,j_2}$, that is,
$\tuple{w, I(x_1)} \in \sem{\Proj_{j_1}}^{\struc[2]}$ and
$\tuple{w, I(x_2)} \in \sem{\Proj_{j_2}}^{\struc[2]}$.
The existentially quantified "variable" $x_2$ in the translated atom is
used as the witness for the original $R$-tuple.
Thus any satisfying witness has the form $u = w_{R v_1,\ldots,v_\ell}$ with
$(v_1,\ldots,v_\ell)\in R^{\struc[1]}$.
For each $i \in \range{\ell}$,
the corresponding conjunct forces $v_i = I(y_i)$:
if $y_i = x_q$ for some $q\in \set{1,2}$,
this follows from the conjunct
$\exists z.(\Proj_{j_q}(x_1,z)\land R^{[i]}(x_2,z))$:
it forces the $j_q$-th projection of the tuple-"vertex" $w$ and the $i$-th
projection of the witness tuple-"vertex" $u$ to be the same original "vertex".
Otherwise, the conjunct $R^{[i]}(x_2,y_i)$ (with $x_2 \mapsto u$) directly forces the $i$-th projection
of $u$ to be $I(y_i)$.
Hence, the translated "atom" holds in $\struc[2]$ "iff" $(I(y_1),\ldots,I(y_\ell))\in R^{\struc[1]}$.

\proofcase{Case $\Lambda_{\emptyset}(\exists y. \fml[1])$ where $y \not\in \set{x_1, x_2}$}
We then have:
\begin{align*}
\struc[1], I \models \exists y. \fml[1]
&\iff \struc[1], I[v/y] \models \fml[1] \text{ for some $v \in \univ{\struc[1]}$}\\
&\iff \struc[2], I[v/y] \models \Lambda_{\emptyset}(\fml[1]) \text{ for some $v \in \univ{\struc[1]}$} \tag{By IH}\\
&\iff \struc[2], I \models \exists y. (\IsV(y) \land \Lambda_{\emptyset}(\fml[1])).
\end{align*}

\proofcase{Case $\Lambda_{\tuple{j_1, j_2}}(\exists y. \fml[1])$ where $y \not\in \set{x_1, x_2}$}
Let $w$ encode the "guarded" pair $\tuple{I(x_1), I(x_2)}$ in the position $\tuple{j_1, j_2}$; note that $\tuple{I[v/y](x_1), I[v/y](x_2)} = \tuple{I(x_1), I(x_2)}$ for every $v \in \univ{\struc[1]}$, as $y \not\in \set{x_1, x_2}$.
We then have:
\begin{align*}
\struc[1], I \models \exists y. \fml[1]
&\iff \struc[1], I[v/y] \models \fml[1] \text{ for some $v \in \univ{\struc[1]}$}\\
&\iff \struc[2], I[v/y][w/x_1] \models \Lambda_{\tuple{j_1, j_2}}(\fml[1]) \text{ for some $v \in \univ{\struc[1]}$} \tag{By IH}\\
&\iff \struc[2], I[w/x_1] \models \exists y. (\IsV(y) \land \Lambda_{\tuple{j_1, j_2}}(\fml[1])). 
\end{align*}

\proofcase{Case $\Lambda_{\emptyset}(\exists x_1. \fml[1])$}
Let $\intro*\DgOne$ and $\intro*\DnOne$ be the first and second disjuncts of $\Lambda_{\emptyset}(\exists x_1. \fml[1])$, respectively.
Since $\struc[1], I \models \exists x_1. \fml[1]$ "iff" $\struc[1], I[v/x_1] \models \fml[1]$ for some $v \in \univ{\struc[1]}$,
it suffices to show the following:
\begin{itemize}
    \item $\struc[2], I \models \DgOne$ "iff" $\struc[1], I[v/x_1] \models \fml[1]$ for some $v \in \univ{\struc[1]}$ such that $\tuple{v, I(x_2)}$ is "guarded";
    \item $\struc[2], I \models \DnOne$ "iff" $\struc[1], I[v/x_1] \models \fml[1]$ for some $v \in \univ{\struc[1]}$ such that $\tuple{v, I(x_2)}$ is not "guarded".
\end{itemize}

\subproofcase{For $\DgOne$}
$\struc[2], I \models \DgOne$ "iff" there are $u \in \univ{\struc[2]}$, $j_1, j_2 \in \range{3}$, and $v$ such that
$\tuple{u, v} \in \sem{\Proj_{j_1}}^{\struc[2]}$,
$\tuple{u, I(x_2)} \in \sem{\Proj_{j_2}}^{\struc[2]}$, and
$\struc[2], I[u/x_1] \models \Lambda_{\tuple{j_1, j_2}}(\fml[1])$.
By IH (applied to $I[v/x_1]$ and $u$), the last condition is equivalent to $\struc[1], I[v/x_1] \models \fml[1]$.
The conditions $\tuple{u, v} \in \sem{\Proj_{j_1}}^{\struc[2]}$ and $\tuple{u, I(x_2)} \in \sem{\Proj_{j_2}}^{\struc[2]}$ say that $u$ encodes the pair $\tuple{v, I(x_2)}$ in the position $\tuple{j_1, j_2}$.
Such $u$, $j_1$, and $j_2$ exist "iff" $v \in \univ{\struc[1]}$ and $\tuple{v, I(x_2)}$ is "guarded".
Hence, $\struc[2], I \models \DgOne$ "iff" $\struc[1], I[v/x_1] \models \fml[1]$ for some $v \in \univ{\struc[1]}$ such that $\tuple{v, I(x_2)}$ is "guarded".

\subproofcase{For $\DnOne$}
$\struc[2], I \models \DnOne$ "iff" there is $v$ such that
$v \in \univ{\struc[1]}$ (by $\IsV(x_1)$),
$\tuple{v, I(x_2)}$ is not "guarded" (by $\lnot \Gdd(x_1, x_2)$), and
$\struc[2], I[v/x_1] \models \Lambda_{\emptyset}(\fml[1])$;
by IH, the last condition is equivalent to $\struc[1], I[v/x_1] \models \fml[1]$.

\proofcase{Case $\Lambda_{\tuple{j_1, j_2}}(\exists x_1. \fml[1])$}
Let $w$ be such that $\tuple{w, I(x_1)} \in \sem{\Proj_{j_1}}^{\struc[2]}$ and $\tuple{w, I(x_2)} \in \sem{\Proj_{j_2}}^{\struc[2]}$.
As $I(x_2)$ is the unique $j_2$-projection of $w$ and $x_1 \not\in \FV(\Lambda_{\emptyset}(\exists x_1. \fml[1]))$,
we have $\struc[2], I[w/x_1] \models \Lambda_{\tuple{j_1, j_2}}(\exists x_1. \fml[1])$ "iff" $\struc[2], I \models \Lambda_{\emptyset}(\exists x_1. \fml[1])$,
which is equivalent to $\struc[1], I \models \exists x_1. \fml[1]$ by the previous case (its proof does not use that $\tuple{x_1, x_2}$ is not "guarded" on $I$).

\proofcase{Case $\Lambda_{\emptyset}(\exists x_2. \fml[1])$}
Let $\intro*\DgTwo$ and $\intro*\DnTwo$ be the first and second disjuncts of $\Lambda_{\emptyset}(\exists x_2. \fml[1])$, respectively.
Since $\struc[1], I \models \exists x_2. \fml[1]$ "iff" $\struc[1], I[v/x_2] \models \fml[1]$ for some $v \in \univ{\struc[1]}$,
it suffices to show the following:
\begin{itemize}
    \item $\struc[2], I \models \DgTwo$ "iff" $\struc[1], I[v/x_2] \models \fml[1]$ for some $v \in \univ{\struc[1]}$ such that $\tuple{I(x_1), v}$ is "guarded";
    \item $\struc[2], I \models \DnTwo$ "iff" $\struc[1], I[v/x_2] \models \fml[1]$ for some $v \in \univ{\struc[1]}$ such that $\tuple{I(x_1), v}$ is not "guarded".
\end{itemize}

\subproofcase{For $\DgTwo$}
$\struc[2], I \models \DgTwo$ "iff" there are $u \in \univ{\struc[2]}$, $j_1, j_2 \in \range{3}$, and $v$ such that
$\tuple{u, I(x_1)} \in \sem{\Proj_{j_1}}^{\struc[2]}$,
$\tuple{u, v} \in \sem{\Proj_{j_2}}^{\struc[2]}$, and
$\struc[2], I[I(x_1)/x_2][u/x_1] \models \Lambda_{\tuple{j_1, j_2}}(\fml[1])$.
As $x_2 \not\in \FV(\Lambda_{\tuple{j_1, j_2}}(\fml[1]))$, the last condition is equivalent to $\struc[2], I[v/x_2][u/x_1] \models \Lambda_{\tuple{j_1, j_2}}(\fml[1])$,
and hence to $\struc[1], I[v/x_2] \models \fml[1]$ by IH (applied to $I[v/x_2]$ and $u$).
The conditions $\tuple{u, I(x_1)} \in \sem{\Proj_{j_1}}^{\struc[2]}$ and $\tuple{u, v} \in \sem{\Proj_{j_2}}^{\struc[2]}$ say that $u$ encodes the pair $\tuple{I(x_1), v}$ in the position $\tuple{j_1, j_2}$.
Such $u$, $j_1$, and $j_2$ exist "iff" $v \in \univ{\struc[1]}$ and $\tuple{I(x_1), v}$ is "guarded".
Hence, $\struc[2], I \models \DgTwo$ "iff" $\struc[1], I[v/x_2] \models \fml[1]$ for some $v \in \univ{\struc[1]}$ such that $\tuple{I(x_1), v}$ is "guarded".

\subproofcase{For $\DnTwo$}
$\struc[2], I \models \DnTwo$ "iff" there is $v$ such that
$v \in \univ{\struc[1]}$ (by $\IsV(x_2)$),
$\tuple{I(x_1), v}$ is not "guarded" (by $\lnot \Gdd(x_1, x_2)$), and
$\struc[2], I[v/x_2] \models \Lambda_{\emptyset}(\fml[1])$;
by IH, the last condition is equivalent to $\struc[1], I[v/x_2] \models \fml[1]$.

\proofcase{Case $\Lambda_{\tuple{j_1, j_2}}(\exists x_2. \fml[1])$}
Let $w$ be such that $\tuple{w, I(x_1)} \in \sem{\Proj_{j_1}}^{\struc[2]}$ and $\tuple{w, I(x_2)} \in \sem{\Proj_{j_2}}^{\struc[2]}$.
By definition, $\Lambda_{\tuple{j_1, j_2}}(\exists x_2. \fml[1])$ first binds $x_2$ to $w$ by $x_2 = x_1$ and then binds $x_1$ to a $j_1$-projection of $w$.
As $I(x_1)$ is the unique $j_1$-projection of $w$ and $x_2 \not\in \FV(\Lambda_{\emptyset}(\exists x_2. \fml[1]))$,
we have $\struc[2], I[w/x_1] \models \Lambda_{\tuple{j_1, j_2}}(\exists x_2. \fml[1])$ "iff" $\struc[2], I \models \Lambda_{\emptyset}(\exists x_2. \fml[1])$,
which is equivalent to $\struc[1], I \models \exists x_2. \fml[1]$ by the previous case (its proof does not use that $\tuple{x_1, x_2}$ is not "guarded" on $I$).

Finally, for an input "sentence" $\fml$, define
\(\fml[2] \defeq
\exists x_1\exists x_2.\,
\bigl(\Proj_1(x_1,x_2)\land
\Lambda_{\tuple{1,1}}(\fml)\bigr)\).
Since $\EQ \in \sigEQ$ and $\tuple{a,a}\in\EQ^{\struc[1]}$ for any $a \in \univ{\struc[1]}$, some $w \in \univ{\struc[2]}$ encodes a pair of original "vertices" in the position $\tuple{1,1}$ ("eg", $w = w_{\EQ a a}$);
hence, by the invariant, we have $\struc[1]\models\fml$ "iff" $\struc[2]\models\fml[2]$.
The output "sentence" $\fml[2]$ uses only $x_1,x_2,x_3$ and has size $\asympO(\len{\fml})$.
Moreover, $\strucn{\struc[2]} = \strucn{\struc[1]} + \sum_{R \in \sigEQ} \card{R^{\struc[1]}}$ and $\strucm{\struc[2]} \le \sum_{R \in \sigEQ} \arity{R} \card{R^{\struc[1]}}$;
since $\card{\EQ^{\struc[1]}} = \strucn{\struc[1]} = \asympO(\strucm{\struc[1]})$ (after \Cref{proposition: isolated}), both are $\asympO(\strucm{\struc[1]})$ for the input "structure",
and the whole construction takes $\asympO(\strucm{\struc[1]} + \len{\fml})$ time.%
\end{proof}
 
\clearpage
\section*{Appendices to \Cref{section: fo2}}

\section{Example of \texorpdfstring{\Cref{proposition: fo2 to bml+c}}{Proposition~\ref{proposition: fo2 to bml+c}}}\label{section: proposition: fo2 to bml+c}
In this section, we present an example of \Cref{proposition: fo2 to bml+c},
which translates $\kFO{2}\sigres{(\ordinalomega, \ordinalomega)}$ into $\BMLPwithC$.
The following translation is based on \cite[Theorem 1]{lutzModalLogicTwoVariable2001}.
We use $\intro*\dia{\term} \mlfml[2]$ to denote the "formula@@BMLP" $\lnot \bo{\term} \lnot \mlfml[2]$.
\begin{example*}
Let $\fml$ be the $\kFO{2}\sigres{(\ordinalomega, \ordinalomega)}$-"formula"
$\exists y ( (S(x, x) \lor R(x, y)) \land \exists x \fml[3] )$.
We focus on
unary types (having one free "variable" $x$), here $S(x, x)$, and
binary types (having two free "variables" $x$ and $y$), here $R(x, y)$.
We first move the "subformula" $S(x, x)$, which has no free occurrence of $y$, out of the scope of $\exists$, as follows:
\begin{align*}
&(S(x, x) \land \exists y ( (\fmltrue \lor R(x, y)) \land \exists x \fml[3] )) \lor (\lnot S(x, x) \land \exists y ( (\fmlfalse \lor R(x, y)) \land \exists x \fml[3] )).
\end{align*}
Second, we ``guess'' binary types:
\begin{align*}
&(S(x, x) ~\land~
(\exists y~ (R(x, y) \land (( \fmltrue \lor \fmltrue) \land \exists x \fml[3] )) ~\lor~
\exists y~ (\lnot R(x, y) \land (( \fmltrue \lor \fmlfalse) \land \exists x \fml[3] )))) ~\lor {}\\
&(\lnot S(x, x) ~\land~
(\exists y~ (R(x, y) \land (( \fmlfalse \lor \fmltrue) \land \exists x \fml[3] )) ~\lor~
\exists y~ (\lnot R(x, y) \land (( \fmlfalse \lor \fmlfalse) \land \exists x \fml[3] )))).
\end{align*}
Finally, we translate this "formula" canonically into the following $\BMLPwithC$-"formula@@BMLP" (here, $\mlfml[2]$ is the $\BMLPwithC$ translation of $\exists x \fml[3]$):
\begin{align*}
&(\dia{S \intersection \corid} \fmltrue \land
(\dia{R} ((\fmltrue \lor \fmltrue) \land \mlfml[2]) \lor
\dia{R^{-}} ((\fmltrue \lor \fmlfalse) \land \mlfml[2]))) ~\lor {} \\
&(\dia{S^{-} \intersection \corid} \fmltrue \land
(\dia{R} ((\fmlfalse \lor \fmltrue) \land \mlfml[2]) \lor
\dia{R^{-}} ((\fmlfalse \lor \fmlfalse) \land \mlfml[2]))).
\end{align*}
This "formula@@BMLP" can be simplified to the following "semantically equivalent" "formula@@BMLP":
\begin{align*}
&\mlfml \defeq (\dia{S \intersection \corid} \fmltrue \land (\dia{R} \mlfml[2] \lor \dia{R^{-}} \mlfml[2])) \lor
(\dia{S^{-} \intersection \corid} \fmltrue \land \dia{R} \mlfml[2]).
\end{align*}
It is straightforward to check that $\mlfml(x)$ is "semantically equivalent" to $\fml$.
\end{example*}

\section{Proof of \texorpdfstring{\Cref{proposition: bml+c to bml+}}{Proposition~\ref{proposition: bml+c to bml+}}}\label{section: proposition: bml+c to bml+}

\begin{proposition*}[Restatement of \Cref{proposition: bml+c to bml+}]
\propositionbmlctobml
\end{proposition*}
\begin{proof}
We define the map $\Lambda$ from  $\BMLPwithC$ "terms@@BMLP" to  $\BMLP$ "terms@@BMLP" and 
the map $\Omega$ from $\BMLPwithC$ "terms@@BMLP" to $\set{0, 1}$, as follows:
\begin{align*}
\tuple{\Lambda(R), \Omega(R)} &\defeq \tuple{R, 1}, \hspace{1.5em}
\tuple{\Lambda(\corid), \Omega(\corid)} \defeq \tuple{\corid, 1}, &
\tuple{\Lambda(\term^{\smile}), \Omega(\term^{\smile})} &\defeq \tuple{\Lambda(\term)^{\smile}, \Omega(\term)}, \\
\tuple{\Lambda(\term[1] \intersection \term[2]), \Omega(\term[1] \intersection \term[2])} &\defeq
\begin{cases}
\tuple{\Lambda(\term[1]) \intersection \Lambda(\term[2]), 1} & \text{if $\tuple{\Omega(\term[1]), \Omega(\term[2])} = \tuple{1, 1}$}\\
\tuple{\Lambda(\term[1]) \setminus \Lambda(\term[2]), 1} & \text{if $\tuple{\Omega(\term[1]), \Omega(\term[2])} = \tuple{1, 0}$}\\
\tuple{\Lambda(\term[2]) \setminus \Lambda(\term[1]), 1} & \text{if $\tuple{\Omega(\term[1]), \Omega(\term[2])} = \tuple{0, 1}$}\\
\tuple{\Lambda(\term[1]) \union \Lambda(\term[2]), 0} & \text{if $\tuple{\Omega(\term[1]), \Omega(\term[2])} = \tuple{0, 0}$},
\end{cases} \span\span\span\span\\
\tuple{\Lambda(\term[1] \union \term[2]), \Omega(\term[1] \union \term[2])} &\defeq
\tuple{\Lambda((\term[1]^{-} \intersection \term[2]^{-})^{-}), \Omega((\term[1]^{-} \intersection \term[2]^{-})^{-})},\\
\tuple{\Lambda(\term[1] \setminus \term[2]), \Omega(\term[1] \setminus \term[2])} &\defeq
\tuple{\Lambda(\term[1] \intersection \term[2]^{-}), \Omega(\term[1] \intersection \term[2]^{-})}, \span\span\\
\tuple{\Lambda(\term[1]^{-}), \Omega(\term[1]^{-})} &\defeq
\begin{cases}
\tuple{\Lambda(\term[1]), 0} & \text{if $\Omega(\term[1]) = 1$}\\
\tuple{\Lambda(\term[1]), 1} & \text{if $\Omega(\term[1]) = 0$}.
\end{cases} \span\span\span\span
\end{align*}
By a straightforward induction on $\term$,
$\Lambda(\term)$ is "semantically equivalent" to $\term$ if $\Omega(\term) = 1$
and to $\term^{-}$ if $\Omega(\term) = 0$.

We then define the map $\Lambda'$ from $\BMLPwithC$-"formulas@@BMLP" to $\BMLP$-"formulas@@BMLP", as follows:
\begin{align*}
\Lambda' (U) &\defeq U, \quad \Lambda'(\fmltrue) \defeq \fmltrue, &
\Lambda' (\lnot \mlfml[1]) &\defeq \lnot \Lambda'(\mlfml[1]),\\
\Lambda' (\mlfml[1] \lor \mlfml[2]) &\defeq \Lambda'(\mlfml[1]) \lor \Lambda'(\mlfml[2]), &
\Lambda' (\mlfml[1] \land \mlfml[2]) &\defeq \Lambda'(\mlfml[1]) \land \Lambda'(\mlfml[2]),\\
\Lambda' (\bo{\term} \mlfml[1]) &\defeq \begin{cases}
\bo{\Lambda(\term)} \Lambda'(\mlfml[1]) & \text{if } \Omega(\term) = 1\\
\windowbo{\Lambda(\term)} \lnot \Lambda'(\mlfml[1]) & \text{if } \Omega(\term) = 0,
\end{cases} &
\Lambda' (\windowbo{\term} \mlfml[1]) &\defeq \begin{cases}
\windowbo{\Lambda(\term)} \Lambda'(\mlfml[1]) & \text{if } \Omega(\term) = 1\\
\bo{\Lambda(\term)} \lnot \Lambda'(\mlfml[1]) & \text{if } \Omega(\term) = 0.
\end{cases}
\end{align*}
By a straightforward induction on $\mlfml$,
we have that $\Lambda'(\mlfml)$ is "semantically equivalent" to $\mlfml$.

The translation runs in $\asympO(\len{\mlfml})$ time.
\end{proof}

\section{Proof of \Cref{proposition: BMLP model checking}} \label{section: proposition: BMLP model checking}
In this section, we prove the following:
\begin{proposition*}[Restatement of \Cref{proposition: BMLP model checking}]
\propositionBMLPmodelchecking
\end{proposition*}

\begin{proposition}\label{proposition: BMLP term evaluation}
\gdef\propositionBMLPtermevaluation{%
For $\BMLP$ "terms@@BMLP" $\term$,
$\sem{\term}^{\struc}$ can be computed in $\asympO(\len{\term}) \cdot (\strucn{\struc} + \strucm{\struc})$ time.
}\propositionBMLPtermevaluation
\end{proposition}
\begin{proof}
We evaluate $\term$ in a bottom-up manner.
Given $\sem{\term[1]}^{\struc}$ and $\sem{\term[2]}^{\struc}$,
we can compute
$\sem{R}^{\struc}$, 
$\sem{\corid}^{\struc}$, 
$\sem{\term[1]^{\smile}}^{\struc}$,
$\sem{\term[1] \intersection \term[2]}^{\struc}$, 
$\sem{\term[1] \union \term[2]}^{\struc}$, and
$\sem{\term[1] \setminus \term[2]}^{\struc}$ in $\asympO(\strucn{\struc} + \strucm{\struc})$ time, based on \Cref{proposition: BMLP term sparsity}.

Below, we give a more precise description of the algorithm.
Assume that the universe has been identified with $\range{\strucn{\struc}}$.
Each binary "relation" $E$ is represented as a duplicate-free lexicographically sorted list of ordered pairs.
Since each coordinate fits in one word, radix sort builds
this representation in $\asympO(\strucn{\struc}+\card{E})$ deterministic time in the "word-RAM model".

\proofcase{For $R$ and $\corid$}
The list for $R$ is obtained from the input as above, and the list for $\corid$ is $\tuple{1,1}, \dots, \tuple{\strucn{\struc},\strucn{\struc}}$, built in $\asympO(\strucn{\struc})$ time.

\proofcase{For unions ($\union$), intersections ($\intersection$), and differences ($\setminus$)}
These cases are implemented by linear merge scans over sorted lists.

\proofcase{For converse ($\cdot^{\smile}$)} We swap the two coordinates and radix-sort the resulting list once.

Thus, every term operation takes deterministic $\asympO(\strucn{\struc}+\strucm{\struc})$ time.
\end{proof}

\begin{proposition}\label{proposition: BMLP model checking for formula}
\gdef\propositionBMLPmodelcheckingforformula{%
For $\BMLP$-"formulas@@BMLP" $\mlfml$,
$\sem{\mlfml}^{\struc}$ can be computed in $\asympO(\len{\mlfml}) \cdot (\strucn{\struc} + \strucm{\struc})$ time.
}\propositionBMLPmodelcheckingforformula
\end{proposition}
\begin{proof}
We evaluate $\mlfml$ in a bottom-up manner.
Given the semantics of the relevant "subformulas" and "terms@@BMLP",
we can compute
$\sem{U}^{\struc}$, $\sem{\fmltrue}^{\struc}$, 
$\sem{\lnot \mlfml[1]}^{\struc}$, 
$\sem{\mlfml[1] \lor \mlfml[2]}^{\struc}$,
$\sem{\mlfml[1] \land \mlfml[2]}^{\struc}$, 
$\sem{\bo{\term}\mlfml[1]}^{\struc}$, and
$\sem{\windowbo{\term}\mlfml[1]}^{\struc}$ in $\asympO(\strucn{\struc} + \strucm{\struc})$ time.

Below, we give a more detailed description of the algorithm.
Each unary relation $E$ is represented as a $0$-$1$ (false-true) array of length $\strucn{\struc}$.

\proofcase{For $\sem{U}^{\struc}$ and $\sem{\fmltrue}^{\struc}$}
The array for $U$ is built from the input in $\asympO(\strucn{\struc} + \card{U^{\struc}})$ time, and the array for $\fmltrue$ is the all-true array.

\proofcase{For $\sem{\lnot \mlfml[1]}^{\struc}$, $\sem{\mlfml[1] \lor \mlfml[2]}^{\struc}$, $\sem{\mlfml[1] \land \mlfml[2]}^{\struc}$}
By applying the corresponding Boolean operations at each index of the arrays,
this takes $\asympO(\strucn{\struc})$ time.

\proofcase{For $\sem{\bo{\term} \mlfml[1]}^{\struc}$}
By $\card{\sem{\term}^{\struc}} \le \asympO(\strucn{\struc} + \strucm{\struc})$ (\Cref{proposition: BMLP term sparsity}),
we can compute $\sem{\bo{\term} \mlfml[1]}^{\struc}$ in $\asympO(\strucn{\struc} + \strucm{\struc})$ time, as follows:
\[\begin{array}{ll}
1: &\sem{\bo{\term} \mlfml[1]}^{\struc} \gets \text{``array of length $\strucn{\struc}$ initialized to all true''}; \\
2: &\textbf{foreach}\ \tuple{x, y} \in \sem{\term}^{\struc}: \\
3: &\qquad\textbf{if}\ y \not\in \sem{\mlfml[1]}^{\struc}: \quad \text{/* $\asympO(1)$ time check (as $\sem{\mlfml[1]}^{\struc}$ is an array) */} \\
4: &\qquad\quad  \sem{\bo{\term} \mlfml[1]}^{\struc}[x] \gets \text{false};
\end{array}\]

\proofcase{For $\sem{\windowbo{\term} \mlfml[1]}^{\struc}$}
We observe:
\[x \in \sem{\windowbo{\term}\mlfml[1]}^{\struc}
~\iff~
\card{\set{y : \tuple{x, y} \in \sem{\term}^{\struc} \text{ and } y \in \sem{\mlfml[1]}^{\struc}}}
=
\card{\sem{\mlfml[1]}^{\struc}}.
\]
Based on this equivalence with $\card{\sem{\term}^{\struc}} \le \asympO(\strucn{\struc} + \strucm{\struc})$ (\Cref{proposition: BMLP term sparsity}),
we can compute $\sem{\windowbo{\term} \mlfml[1]}^{\struc}$ in $\asympO(\strucn{\struc} + \strucm{\struc})$ time, as follows:
\[\begin{array}{ll}
1: & X \gets \text{``array of length $\strucn{\struc}$ initialized to all 0''}; \\
2: &\textbf{foreach}\ \tuple{x, y} \in \sem{\term}^{\struc}: \\
3: &\qquad\textbf{if}\ y \in \sem{\mlfml[1]}^{\struc}: \quad \text{/* $\asympO(1)$ time check (as $\sem{\mlfml[1]}^{\struc}$ is an array) */} \\
4: &\qquad\quad X[x] \gets X[x] + 1;\\
5: &\sem{\windowbo{\term} \mlfml[1]}^{\struc} \gets \text{``array of length $\strucn{\struc}$ such that $\sem{\windowbo{\term} \mlfml[1]}^{\struc}[x] = \text{true}$ "iff" $X[x] = \card{\sem{\mlfml[1]}^{\struc}}$''};
\end{array}\]
Hence, we have the desired algorithm.
\end{proof}
\Cref{proposition: BMLP model checking} follows from \Cref{proposition: BMLP model checking for formula} as follows.
\begin{proof}[Proof of \Cref{proposition: BMLP model checking}]
Similarly to \Cref{proposition: isolated},
we first keep at most two isolated "elements" (including $I(x)$ if it is isolated) and remove the other isolated "elements"; we still denote the resulting "structure" by $\struc$.
This is computable in $\asympO(\strucm{\struc})$ time.
Since every $\BMLP$-"formula@@BMLP" is expressible by an $\kFO{2}$-"formula" with one free "variable" (by the standard translation),
the same argument as in \Cref{proposition: isolated} ("cf" $2$-pebble games~\cite[\S 3.3.1]{ebbinghausFiniteModelTheory1995}) shows that this preserves the truth of $\mlfml(x)$ under $I$.
The resulting "structure" satisfies $\strucn{\struc} \le \asympO(\strucm{\struc})$.
By \Cref{proposition: BMLP model checking for formula},
we can compute $\sem{\mlfml}^{\struc}$ in $\asympO(\len{\mlfml} \cdot \strucm{\struc})$ time
and then test whether $I(x)\in\sem{\mlfml}^{\struc}$.
\end{proof}

\section{Supplement to Remark~\ref{remark: BMLP to CoR}: A Fragment of \texorpdfstring{$\CoR$}{CoR} Corresponding to \texorpdfstring{$\BMLP$}{ML*}}\label{section: BMLP to CoR}

For $\CoR$ "terms@@CoR" $\term[1]$ and $\term[2]$, we write
\begin{align*}
\term[1]\setminus\term[2]
&\defeq
\term[1]\intersection\term[2]^{-},
&
\term[1]^{\exists}
&\defeq
\term[1] \compo \cortop,\\
\bo{\term[1]}\term[2]
&\defeq
\corid\setminus
\bigl(\term[1]\setminus((\term[2] \intersection \corid)^{\exists})^{\smile}\bigr)^{\exists},
&
\windowbo{\term[1]}\term[2]
&\defeq
\corid\setminus
\bigl(((\term[2] \intersection \corid)^{\exists})^{\smile}\setminus\term[1]\bigr)^{\exists}.
\end{align*}
Here, $\setminus$ takes set difference and ${}^{\exists}$ forms the cylinder over the domain of a "term@@CoR".
These "terms@@CoR" encode the box and window modalities of $\BMLP$ on the diagonal.
Given a "structure" $\struc$, we have
\[
\sem{\term[1]\setminus\term[2]}^{\struc}
=
\sem{\term[1]}^{\struc}\setminus\sem{\term[2]}^{\struc},
\qquad
\sem{\term[1]^{\exists}}^{\struc}
=
\set{\tuple{a,c} \in\univ{\struc}^2 \mid \exists b\in\univ{\struc}.~\tuple{a,b}\in\sem{\term[1]}^{\struc}}.
\]
Thus, we have:
\begin{align*}
\sem{\bo{\term[1]}\term[2]}^{\struc}
&=
\set{\tuple{a,a}\mid
\forall b\in\univ{\struc}.~
\tuple{a,b}\in\sem{\term[1]}^{\struc}
\text{ implies }
\tuple{b,b}\in\sem{\term[2]}^{\struc}},\\
\sem{\windowbo{\term[1]}\term[2]}^{\struc}
&=
\set{\tuple{a,a}\mid
\forall b\in\univ{\struc}.~
\tuple{b,b}\in\sem{\term[2]}^{\struc}
\text{ implies }
\tuple{a,b}\in\sem{\term[1]}^{\struc}}.
\end{align*}
Define $\intro*\CoRtwostar$ as the set of $\CoR$ "terms@@CoR" generated by the following grammar, where $R$ ranges over binary "relation symbols":
\[
\term[1],\term[2] ~\Coloneqq~
R \mid \corid \mid \term[1]^{\smile} \mid
\term[1]\intersection\term[2] \mid \term[1]\union\term[2]
~\mid~ \term[1]\setminus\term[2]
~\mid~ \bo{\term[1]}\term[2] \mid \windowbo{\term[1]}\term[2].
\]
$\CoRtwostar$ is a fragment of $\intro*\CoRtwo$ \cite[\S 3.2]{nakamuraExpressivePowerSuccinctness2022}, which is the set of $\CoR$ "terms@@CoR" generated by the grammar
\[
\term[1],\term[2] ~\Coloneqq~
R \mid \corid \mid \term[1]^{\smile} \mid
\term[1]\intersection\term[2] \mid \term[1]\union\term[2]
~\mid~ \term[1]^{-}
~\mid~ \term[1]^{\exists}.
\]
Observe that the grammar of $\CoRtwo$ uses composition only inside ${}^{\exists}$.

\subsection{Translation from \texorpdfstring{$\BMLP$}{ML*} to \texorpdfstring{$\CoRtwostar$}{CoR2star}}\label{section: BMLP to CoR translation}

Assume that each unary "relation symbol" $U$ has a corresponding binary "relation symbol" $D_U$ interpreted as
\(D_U^{\struc}\defeq \set{\tuple{a,a}\mid a\in U^{\struc}}\).

Define the translation $\delta$ from $\BMLP$-"formulas@@BMLP" to $\CoRtwostar$ "terms@@CoR", as follows:
\begin{align*}
\delta(U) &\defeq D_U, \quad \delta(\fmltrue) \defeq \corid,&
\delta(\lnot\mlfml) &\defeq \corid\setminus\delta(\mlfml),&
\delta(\bo{\term}\mlfml) &\defeq \bo{\term}\delta(\mlfml),\\
\delta(\mlfml[1]\land\mlfml[2]) &\defeq \delta(\mlfml[1])\intersection\delta(\mlfml[2]),&
\delta(\mlfml[1]\lor\mlfml[2]) &\defeq \delta(\mlfml[1])\union\delta(\mlfml[2]),&
\delta(\windowbo{\term}\mlfml) &\defeq \windowbo{\term}\delta(\mlfml).
\end{align*}
\begin{proposition}\label{proposition: BMLP to CoR}
For every $\BMLP$-"formula@@BMLP" $\mlfml$,
one can construct, in $\asympO(\len{\mlfml})$ time,
a $\CoRtwostar$ "term@@CoR" $\delta(\mlfml)$ such that, for every "structure" $\struc$ satisfying the above condition on the symbols $D_U$,
we have \(\sem{\delta(\mlfml)}^{\struc} = \set{\tuple{a,a} \mid a\in\sem{\mlfml}^{\struc}}\).
\end{proposition}
\begin{proof}
By induction on $\mlfml$.
The atomic and Boolean cases are immediate from the definitions.
For the two modal cases, the induction hypothesis says that $\delta(\mlfml)$ stores exactly $\sem{\mlfml}^{\struc}$ on the diagonal, and the definitions above are exactly the corresponding $\BMLP$ semantics.
The construction adds only a constant number of symbols per constructor.
\end{proof}

\subsection{Translation from \texorpdfstring{$\CoRtwo$}{CoR2} to \texorpdfstring{$\BMLP$}{ML*} for unary queries}\label{section: CoR to BMLP translation}

Conversely, consider a $\CoRtwo$ (in particular, $\CoRtwostar$) "atomic FO-formula@@CoR" of the form $\term(x,x)$.
The translation from $\CoRtwo$ to $\kFO{2}$ in \cite[\S 3.2]{nakamuraExpressivePowerSuccinctness2022}
gives, in $\asympO(\len{\term})$ time,
an $\kFO{2}\sigres{(\ordinalomega,\ordinalomega)}$-"formula" $\alpha_{\term}(x)$ such that
\[
\struc,I \models \alpha_{\term}(x)
\quad\text{if and only if}\quad
\tuple{I(x),I(x)}\in\sem{\term}^{\struc}.
\]
By \Cref{proposition: fo2 to bml+c,proposition: bml+c to bml+},
the "formula" $\alpha_{\term}(x)$ can be translated into an $\BMLP$-"formula@@BMLP".
\begin{proposition}\label{proposition: CoR to BMLP}
For every $\CoRtwo$ "term@@CoR" $\term$ (in particular, every $\CoRtwostar$ "term@@CoR"),
one can construct, in $2^{\asympO(\len{\term})}$ time,
an $\BMLP$-"formula@@BMLP" $\eta(\term)$ such that
$\term(x,x)$ and $\eta(\term)(x)$ are "semantically equivalent".
\end{proposition}
\begin{proof}
Translate $\term(x,x)$ into the $\kFO{2}$-"formula" $\alpha_{\term}(x)$ above.
Then apply the translation from $\kFO{2}$ to $\BMLPwithC$ and the complement-elimination translation from $\BMLPwithC$ to $\BMLP$.
Semantic equivalence follows by composing these translations.
\end{proof}
Thus, $\BMLP$, $\CoRtwostar$, and $\CoRtwo$ have the same "expressive power" for "atomic FO-formulas@@CoR" of the form $\term(x,x)$ (unary queries).

\subsection{Expressive power comparison between \texorpdfstring{$\CoRtwostar$}{CoR2star}, \texorpdfstring{$\CoRtwo$}{CoR2}, and \texorpdfstring{$\CoR$}{CoR}}\label{section: expressive power comparison}
For unary queries, we view a $\CoR$ "term@@CoR" $\term$ as the "atomic FO-formula@@CoR" $\term(x,x)$.
After encoding unary "relation symbols" as diagonal binary "relation symbols", the relationships are:
\[
(\BMLP \mathrel{=)} \CoRtwostar = \CoRtwo = \kFO{2} \subsetneq \CoR = \kFO{3}.
\]
The equivalences between $\BMLP$, $\CoRtwostar$, and $\CoRtwo$ follow from \Cref{proposition: BMLP to CoR,proposition: CoR to BMLP} and the correspondence between $\CoRtwo$ and $\kFO{2}$ \cite[\S 3.2]{nakamuraExpressivePowerSuccinctness2022}.
The equivalence between $\CoR$ and $\kFO{3}$ follows from \cite{tarskiFormalizationSetTheory1987,nakamuraExpressivePowerSuccinctness2022}.

For the strict inclusion $\kFO{2} \subsetneq \CoR$, consider the following "sentence", which detects a directed $3$-cycle over irreflexive directed graphs:
\[
\exists x\,\exists y\,\exists z\, (E(x,y)\land E(y,z)\land E(z,x)).
\]
$\CoR$ can also express the "sentence" by
\[
\bigl(\cortop\compo((E\compo E)\intersection E^{\smile})\compo\cortop\bigr)(v,v),
\]
where $v$ is a dummy "variable".
In contrast, $\kFO{2}$ cannot express this "sentence":
it separates the disjoint union of the irreflexive directed cycles $C_3$ and $C_4$ from that of two copies of $C_4$,
whereas no $\kFO{2}$-"sentence" does, as easily shown by $2$-pebble games~\cite[\S 3.3.1]{ebbinghausFiniteModelTheory1995}.

For binary queries, the equality between $\CoRtwostar$ and $\CoRtwo$ above becomes a strict inclusion:
\[
\CoRtwostar \subsetneq \CoRtwo = \kFO{2} \subsetneq \CoR = \kFO{3}.
\]
The strictness of $\CoRtwostar \subsetneq \CoRtwo$ follows from the following sparsity property, which holds for every $\CoRtwostar$ "term@@CoR" $\term$ ("cf" \Cref{proposition: BMLP term sparsity}):
\[
\sem{\term}^{\struc}
\subseteq
\sem{\corid}^{\struc}
\cup
\bigcup_{R \in \sig : \arity{R}=2} \sem{R \union R^{\smile}}^{\struc}.
\]
This is shown by induction on $\term$.
Thus, for instance, $\CoRtwostar$ cannot express the off-diagonal relation $\corid^{-}$, which is expressible in $\CoRtwo$.

\clearpage
\section*{Appendices to \Cref{sec: dichotomy dense}}

\section{Dichotomy Theorem in the Dense Setting} \label{section: theorem: dichotomy dense}
We also establish the following dichotomy result for $\kPP{k}\sigres{\sig} \subseteq \fmlclass \subseteq \kFO{k}\sigres{\sig}$, in the dense setting.
Its hardness part relies on the following hypothesis.
\AP
\begin{hypothesis*}[$(k, r)$-\intro*\HyperCliqueHypothesis{} \cite{abboudMoreConsequencesFalsifying2018,lincolnTightHardnessShortest2018}]
Let $k > r \ge 3$.
For any $\eps > 0$,
the $(k, r)$-\HyperCliqueProblem{} on $r$-"uniform hypergraphs" with $n$ "vertices" cannot be solved in $\asympO(n^{k - \eps})$ time.
\end{hypothesis*}
\begin{theorem}[Dichotomy in the dense setting; see also \Cref{table: dense dichotomy}]\label{theorem: dichotomy dense}
Fix $k \ge 2$, a "relational signature" $\sig$, and a class $\fmlclass$ such that $\kPP{k}\sigres{\sig} \subseteq \fmlclass \subseteq \kFO{k}\sigres{\sig}$.
Then the following dichotomy holds:

\noindent\textbf{Tractable cases.}
$\MC{\fmlclass}$ is solvable in $2^{\asympO(\len{\fml})} \cdot n^{k - \eps}$ time for some $\eps > 0$ if one of the following holds:
\begin{itemize}
    \item $\sigarity{\sig} \le 1$, in which case the problem reduces to one-variable counting logic $\kC{1}$ (this case is solvable in $\asympO(\len{\fml})$ time if $\sigarity{\sig} = 0$ and in $2^{\asympO(\len{\fml})} \cdot n$ time if $\sigarity{\sig} = 1$);
    \item $k = 3$ and $\sigarity{\sig} = 2$, in which case the problem reduces to $\CoR$ (this case is solvable in $2^{\asympO(\len{\fml})} \cdot n^{\matrixexponent}$ time).
\end{itemize}

\noindent\textbf{Hard cases.}
In all remaining cases, there is a fixed $\kPP{k}\sigres{\sig}$-"sentence"
$\fml$ whose "model-checking problem" cannot be solved in $\asympO(n^{k-\eps})$ time for any $\eps > 0$, under the corresponding condition below:
\begin{itemize}
    \item $k \le \sigarity{\sig}$, unconditionally (see \Cref{proposition: trivial lower bound});
    \item $k \geq 4$ and $2 \le \sigarity{\sig} < k$, assuming the $(qk, q)$-\HyperCliqueHypothesis{} for some $3 \le q < k$.
\end{itemize}
Consequently, under the corresponding assumption above,
$\MC{\fmlclass}$ cannot be solved in $\asympO(n^{k - \eps})$ time for any $\eps > 0$,
and hence cannot be solved in $f(\len{\fml}) \cdot n^{k-\eps}$ time
for any function $f$ and any $\eps > 0$.
\end{theorem}

\begin{table}[h]
\centering
\caption{Dense dichotomy for $\kPP{k}\sigres{\sig} \subseteq \fmlclass \subseteq \kFO{k}\sigres{\sig}$
where $k \ge 2$ and $\sig$ is a \kl{relational signature}.
The cyan entries correspond to tractable cases
and the orange entries correspond to hard cases in \Cref{theorem: dichotomy dense}.}
\label{table: dense dichotomy}

\definecolor{myPink}{HTML}{FFEEDD}
\definecolor{myCyan}{HTML}{DDEEFF}

\resizebox{\textwidth}{!}{
\begin{tblr}{
  colspec = {Q[c,m,1.cm] | Q[c,m,1.4cm] | Q[c,m,1.4cm] | Q[c,m,1.5cm] | Q[c,m,5.5cm] | Q[c,1.4cm] },
  row{1} = {},
  column{1} = {},
  hlines, 
  vlines,
  hline{2} = {1.5pt},
  hline{3} = {5-6}{fg=myPink},
  hline{4} = {6-6}{fg=myPink},
  hline{5} = {5-5}{fg=myPink},
  vline{2} = {1.5pt},
  vline{3} = {2-5}{0pt},
  vline{5} = {2-2}{fg=myPink},
  vline{5} = {4-5}{fg=myPink},
  vline{6} = {5-5}{fg=myPink},
  cell{2}{4} = {}{bg=myPink},       %
  cell{2}{5} = {r=2, c=2}{bg=myPink, fg = gray},       %
  cell{3}{5} = {c=2}{bg=myPink},       %
  cell{4}{6} = {}{bg=myPink}, %
  cell{2}{2} = {r=4}{bg=myCyan},       %
  cell{2}{3} = {r=4}{bg=myCyan},       %
  cell{3}{4} = {c=1}{bg=myCyan},       %
  cell{4}{4} = {r=2}{bg=myPink}, %
  cell{4}{5} = {r=2}{bg=myPink}, %
  cell{5}{6} = {}{bg=myPink}, %
}
  \diagbox[width=1.2cm, height = 1.\line]{\raisebox{-.5ex}{$k$}}{\raisebox{.5ex}{$\sigarity{\sig}$}} & 0 & 1 & 2 & 3 & $\sigarity{\sig} \ge 4$ \\

  2 & {$\asympO(\len{\fml})$} \par via $\kC{1}$ \par (Thm.~\ref{theorem: monadic}) & {$2^{\asympO(\len{\fml})} \cdot n$} \par via $\kC{1}$ \par (Thm.~\ref{theorem: monadic}) 
    & & Naive $\asympO(n^k)$ is optimal \par by the direct $r$-ary relation lower bound \par (Prop.~\ref{proposition: trivial lower bound}) & \\

  3 & & & {$2^{\asympO(\len{\fml})} \cdot n^{\matrixexponent}$} \par via CoR (Thm.~\ref{theorem: fo3}) & \\

  4 & & & & Naive $\asympO(n^k)$ is optimal under \par $(qk, q)$-\HyperCliqueHypothesis{} \par 
  \text{for some $3 \le q < k$} (Cor.~\ref{corollary: dense hardness}) & \\

  $k \ge 5$ & & & & & $\ddots$ 

\end{tblr}
}
 \end{table}

\Cref{theorem: dichotomy dense} is a direct consequence of combining the results in \Cref{section: dichotomy sparse} and \Cref{sec: dichotomy dense}.
See also \Cref{section: trivial lower bound} for the direct $r$-ary relation lower bound.

\begin{proof}[Proof of \Cref{theorem: dichotomy dense}]
\proofcase{For Tractable Cases}
If $\sigarity{\sig}\le 1$, the claim follows from \Cref{theorem: monadic}.
If $k=3$ and $\sigarity{\sig}=2$, it follows from \Cref{theorem: fo3}.
These give $n^{k-\eps}$ bounds with $\eps = k-1$ and $\eps = 3-\matrixexponent$, respectively.
\proofcase{For Hard Cases}
If $k \le \sigarity{\sig}$, the claim follows from \Cref{proposition: trivial lower bound} applied to some "relation symbol" $R \in \sig$ of arity $r \ge k$.
If $k\ge4$ and $2 \le \sigarity{\sig} < k$,
then this claim follows from \Cref{corollary: dense hardness}
(and \Cref{proposition: arity monotone} for $\sigarity{\sig} > 2$).
\end{proof}

\subsection{On the direct \texorpdfstring{$r$-ary}{r-ary} relation lower bound}\label{section: trivial lower bound}

In this subsection, we show that one \(r\)-ary "relation symbol" with \(r\ge k\) already gives a fixed \(\PNFkPP{k}\)-"sentence" ruling out \(\asympO(n^{k-\eps})\) time, both for adjacency tensors and for lists of tuples.
\begin{proposition}\label{proposition: trivial lower bound}%
\gdef\trivialLowerbound{
Fix integers $k \geq 2$ and $r \geq k$, and let $R$ be an $r$-ary "relation symbol".
Consider the following fixed $\PNFkPP{k}$-"sentence":
\[
\fml_{k,r} \:\defeq\: \exists x_1. \dots \exists x_k.\,
R(x_1,x_2,x_3,\ldots,x_k,\underbracket{x_k,\ldots,x_k}_{r-k\text{ further copies}})
\land
R(x_2,x_1,x_3,\ldots,x_k,\underbracket{x_k,\ldots,x_k}_{r-k\text{ further copies}}),
\]
where \(x_3,\ldots,x_k\) are absent if \(k=2\), and the final repeated block is omitted if \(r=k\).
When \(R^{\struc}\) is given either by its full \(n^r\)-entry adjacency tensor
or as a list of tuples,
$\MC{\set{\fml_{k,r}}}$ cannot be solved in $\asympO(n^{k - \eps})$ time for any $\eps > 0$.
}
\trivialLowerbound
\end{proposition}
\begin{proof}
Let \(\struc[1]\) be the "structure" with
"universe" \(\univ{\struc[1]}=\range{n}\) and \(R^{\struc[1]}\) defined as follows:
\[
R^{\struc[1]} \defeq
\set{\tuple{a_2,a_1,a_3,\ldots,a_k,a_k,\ldots,a_k}
\mid a_1,\ldots,a_k\in \univ{\struc[1]},\ a_1<a_2}.
\]
Then \(\struc[1] \not\models \fml_{k, r}\)
but there are \(\asympTheta(n^k)\) tuples of the form
\(\tuple{a_1,a_2,a_3,\ldots,a_k,a_k,\ldots,a_k}\) with \(a_1<a_2\)
that do not belong to \(R^{\struc[1]}\).

\proofcase{Adjacency-Tensor Case}
Let $w=\asympTheta(\log n)$ be the word size.
The corresponding entries occupy
\(\asympOmega(n^k/w)\) words.
If a deterministic algorithm in the "word-RAM model" reads \(\asympo(n^k/w)\) words on the encoding of \(\struc[1]\),
then it leaves unread some word containing one of these entries.
Changing that entry from \(0\) to \(1\) gives another input "structure" satisfying \(\fml_{k,r}\) while leaving the execution unchanged,
which is a contradiction.
Thus the worst-case running time is $\asympOmega(n^k/w)=\asympOmega(n^k/\log n)$,
which is not $\asympO(n^{k-\eps})$ for any fixed $\eps>0$.

\proofcase{Tuple-List Case}
When \(R^{\struc[1]}\) is represented as a list of its
\(\asympTheta(n^k)\) tuples,
an \(\asympo(n^k)\)-time deterministic algorithm in the "word-RAM model" reads only \(\asympo(n^k)\) tuples from the list,
so some tuple remains unread.
Replacing that unread tuple by \(\tuple{a,\ldots,a}\), for any fixed \(a\in \univ{\struc[1]}\),
makes the resulting input "structure" satisfy \(\fml_{k,r}\), again without changing the execution.
Hence this representation also requires \(\asympOmega(n^k)\) time in the worst case.
\end{proof}

\begin{remark}
\Cref{theorem: dichotomy dense} excludes $k = 1$.
For adjacency tensors,
"model checking" of the "sentence" $\exists x.\ U(x)$ cannot be solved in $\asympO(n^{1-\eps})$ time for any $\eps > 0$.
If the input is instead given by lists of tuples,
the case of $k = 1$, $\sigarity{\sig} = 1$, and $\card{\sig} = 1$ is exceptional, for the same reason as in \Cref{section: footnote:k=1}.
\end{remark}

\clearpage

\section*{Appendices to \Cref{sec: fixed sentence}}

\section{Proofs of \texorpdfstring{\Cref{theorem:sparse-k4 fo,theorem:sparse-k3 fo}}{Theorem~\ref{theorem:sparse-k4 fo} and Theorem~\ref{theorem:sparse-k3 fo}}}\label{section: proof of hardness by H-Clique fo}

\begin{lemma}\label{lemma: hardness by OVHClique with QW}
Fix $k > r \ge 2$.
There exists a fixed $\kPFO{k}_{\exists^{k}\forall^{1}}\sigres{(0, 1)}$-"sentence" $\fml$ with the following property:
for every instance $\mathcal{V}$ of the $(k, r)$-\OVHyperCliqueProblem{n, d},
we can construct a "structure" $\struc_{\mathcal{V}}$ in $\asympO(n (d+1))$ time such that
\[
\text{$\OVgraph{\mathcal{V}, r}$ contains a $(k, r)$-"hyperclique"}
\quad\iff\quad 
\struc_{\mathcal{V}} \models \fml.\]
Moreover, $\len{\fml} = \asympO(r \binom{k}{r})$,
$\strucn{\struc_{\mathcal{V}}} = n + d$, and
$\strucm{\struc_{\mathcal{V}}} \le n (d+1)$.
\end{lemma}
\begin{proof}
Let $\fml$ be the $\kPFO{k}_{\exists^{k}\forall^{1}}\sigres{(0, 1)}$-"sentence" defined as follows:
\begin{align*}
\fml
&~\defeq~
\exists x_1. \dots \exists x_{k}.~ \left(\left(\bigwedge_{i = 1}^{k} E(x_i, x_i)\right) \land
\bigwedge_{S \in \binom{\range{k}}{r}} \hyperedgefml(S)\right),
\end{align*}
where
\[
\intro*\hyperedgefml(S) ~\defeq~ \forall x_{\mu(S)}.\, E(x_{\mu(S)}, x_{\mu(S)}) \lor \bigvee_{j \in S} E(x_{\mu(S)}, x_j)
\]
and $\mu(S) \defeq \min(\range{k} \setminus S)$, which exists since $k > r$.
By construction, we only use the "variables" $x_1, \dots, x_k$, and hence $\fml \in \kPFO{k}_{\exists^{k}\forall^{1}}\sigres{(0, 1)}$.
Moreover, $\len{\fml} = \asympO(r \binom{k}{r})$.
Let $\mathcal{V}$ be an instance of the $(k, r)$-\OVHyperCliqueProblem{n, d}.
We then build the following "structure" $\struc_{\mathcal{V}}$:
\begin{align*}
    \univ{\struc_{\mathcal{V}}} &\defeq \mathcal{V} \sqcup \range{d},&
    E^{\struc_{\mathcal{V}}} &\defeq \set{\tuple{v, v} \mid v \in \mathcal{V}} \sqcup \set{\tuple{\ell, v} \in \range{d} \times \mathcal{V} \mid v\ith{\ell} = 0}.
\end{align*}
For all $v_1, \dots, v_k \in \mathcal{V}$
and $S \in \binom{\range{k}}{r}$,
we have:
\begin{align*}
\forall \ell \in \range{d}.\, \bigvee_{j \in S} (v_j \ith{\ell} = 0)
&\iff \forall \ell \in \range{d}.~ \struc_{\mathcal{V}}, \set{x_j \mapsto v_j \mid j \in S} \cup \set{x_{\mu(S)} \mapsto \ell} \models \bigvee_{j \in S} E(x_{\mu(S)}, x_j)\\
&\iff \struc_{\mathcal{V}}, \set{x_j \mapsto v_j \mid j \in S} \models  \hyperedgefml(S).
\end{align*}
We show that
$\OVgraph{\mathcal{V}, r}$ has a $(k, r)$-"hyperclique"
"iff"
$\struc_{\mathcal{V}} \models \fml$.
\proofcase{($\Rightarrow$)} Straightforward.
\proofcase{($\Leftarrow$)}
Suppose $\struc_{\mathcal{V}} \models \fml$,
witnessed by $x_1, \dots, x_k \mapsto a_1,\ldots,a_k$.
By $\tuple{a_1, a_1}, \dots, \tuple{a_k, a_k} \in E^{\struc_{\mathcal{V}}}$,
$a_1, \dots, a_k \in \mathcal{V}$.
By the equivalence above, for each $r$-element subset $S \in \binom{\range{k}}{r}$,
$\GIP{a_{S\ith{1}}, \ldots, a_{S\ith{r}}} = 0$.
Because the input $\mathcal{V}$ is "well-formed",
$a_{S\ith{1}}, \ldots, a_{S\ith{r}}$ are pairwise distinct for each $S \in \binom{\range{k}}{r}$, and hence $a_1, \dots, a_k$ are pairwise distinct.
Hence, $\set{a_1,\ldots,a_k}$ is a $(k,r)$-"hyperclique" in $\OVgraph{\mathcal{V}, r}$.
\end{proof}

\begin{theorem*}[Restatement of {\Cref{theorem:sparse-k4 fo}}]
\sparseFOhardnesstheoremFO
\end{theorem*}
\begin{proof}
Toward a contradiction, assume that
$\MC{\set{\fml}}$ is solvable in $\asympO(m^{k - \eps})$ time for some $\eps > 0$,
where $\fml$ is the "sentence" from \Cref{lemma: hardness by OVHClique with QW} for $r = k - 1$.
By \Cref{lemma: hardness by OVHClique with QW},
for every $c > 0$, the $(k, k-1)$-\OVHyperCliqueProblem{n, \le c \log n} can be solved in $\asympO(n \cdot c \log n) + \asympO((n (c \log n + 1))^{k - \eps}) \le \asympO(n^{k - \eps/2})$ time, and hence the $(k, k-1)$-\OVHyperCliqueHypothesis{} fails.
By \Cref{proposition: alman VVW}, the \SparseMAXxSATHypothesis{(k-1)} then fails, which contradicts the assumption.
\end{proof}

\begin{theorem*}[Restatement of {\Cref{theorem:sparse-k3 fo}}]
\sparsehardnessbyovtriangleFO
\end{theorem*}
\begin{proof}
The proof is analogous to that of \Cref{theorem:sparse-k4 fo}, with $k = 3$ and $r = 2$, the bound $\asympO(m^{\matrixexponent - \eps})$ in place of $\asympO(m^{k - \eps})$, and the \OVTriangleHypothesis{} and the \SparseMAXxSATHypothesis{2} in place of the $(k, k-1)$-\OVHyperCliqueHypothesis{} and the \SparseMAXxSATHypothesis{(k-1)}.
\end{proof}

\clearpage
\section*{Appendices to \Cref{sec: target language hardness}}
\section{Proofs of \texorpdfstring{\Cref{corollary: CoR fixed term hardness,corollary: CoR uniform fixed term hardness}}{Corollaries~\ref{corollary: CoR fixed term hardness} and~\ref{corollary: CoR uniform fixed term hardness}}}\label{section: CoR hardness}

\begin{proof}
The translation of \cite{nakamuraExpressivePowerSuccinctness2022} sends each $\kPP{3}\sigres{(0,1)}$ (respectively, $\kFO{3}\sigres{(0,1)}$) "sentence" to a "semantically equivalent" $\PPCoR$ (respectively, $\CoR$) "atomic FO-formula@@CoR" $\term(x,y)$ with distinct "variables" $x, y$ over the same "signature".
Apply it to the "sentences" $\fml_{\eps}$ of \Cref{theorem:sparse-k3} and to the single fixed "sentence" of \Cref{theorem:sparse-k3 fo}, respectively.
Each resulting term denotes the full "relation" when the original "sentence" is true and the empty "relation" otherwise.
Thus both membership and nonemptiness decide the original "sentence" on the same input "structure", preserving $m$ and yielding the claimed lower bounds.
\end{proof}
 
\clearpage
\section*{Appendices to \Cref{section: succinctness}}

We first recall the $k$-\OVHypothesis{} and the low-""dimensional"" $k$-\OVHypothesis{}. \AP
Given $n$-element sets $S_1, \dots, S_k \subseteq \set{0, 1}^d$ (with $d \ge 1$),
the problem $k$-\intro*\OVProblem{n, d} asks whether there exist $u_1 \in S_1, \dots, u_k \in S_k$ such that for every $\ell \in \range{d}$, $u_i\ith{\ell} = 0$ for some $i \in \range{k}$;
we write $k$-\OVProblem{n, \le c \log n} for its restriction to $d \le c \log n$.

\begin{hypothesis*}[The $k$-\intro*\OVHypothesis{} {\cite{gaoCompletenessFirstorderProperties2018}, \cite[Hypothesis~4]{williamsFinegrainedQuestionsAlgorithms2018}}]
    Fix $k \geq 2$.
    For any $\eps > 0$,
    $k$-\OVProblem{n, d} is not in $\poly(d) n^{k-\eps}$ time.
\end{hypothesis*}

\begin{hypothesis*}[The Low-Dimensional $k$-\reintro*\OVHypothesis{} {\cite[Hypothesis~24]{bringmannStructuralInvestigationApproximability2022arXiv}}]
Fix $k \geq 2$.
For any $\eps > 0$,
there exists a constant $c > 0$ such that $k$-\OVProblem{n, \le c \log n} is not in $\asympO(n^{k-\eps})$ time.
\end{hypothesis*}

Both the $k$-\OVHypothesis{} and the low-"dimensional" $k$-\OVHypothesis{} follow from the Strong Exponential Time Hypothesis (""SETH"") \cite{impagliazzoComplexity$k$SAT2001,impagliazzoWhichProblemsHave2001}, via the reduction of Williams~\cite{williamsNewAlgorithmOptimal2005}.

\section{Proof of \texorpdfstring{\Cref{theorem: FO model checking hardness}}{Theorem~\ref{theorem: FO model checking hardness}}: \texorpdfstring{$k$-\OVProblem{}}{k-OV} to \texorpdfstring{$\MC{\PNFkEP{k}\sigres
{(\ordinalomega)}}$}{PNF-EPk[(∞)]-MC}} \label{section: theorem: FO model checking hardness}
To prove \Cref{theorem: FO model checking hardness},
it suffices to give the following reduction.
\begin{lemma}\label{lemma: reduction from kOV to PNF-FO}
Let $S_1, S_2, \ldots, S_k$ be an instance of $k$-\OVProblem{n, d}, "ie",
each $S_i$ is an $n$-element subset of $\set{0, 1}^d$.
Then we can construct, in $\asympO(k d n)$ time, a "structure" $\struc$ and a $\PNFkEP{k}\sigres{(k + d)}$-"sentence" $\fml$,
such that
the instance is true ("ie", $\exists u_1 \in S_1.\,\exists u_2 \in S_2.\ \ldots \exists u_k \in S_k.\ \forall \ell \in \range{d}. \exists i \in \range{k}.\, u_i\ith{\ell} = 0$) "iff" $\struc \models \fml$.
Moreover, $\len{\fml} = \asympO(k d)$, $\strucn{\struc} = k n$, and $\strucm{\struc} = \asympO(k d n)$.
\end{lemma}
\begin{proof}
For $i \in \range{k}$,
let $S_i = \set{v_{i,1}, \dots, v_{i,n}}$.
Let $\hat{S}_1, \ldots, \hat{S}_k$ and $Z_1, \ldots, Z_d$ be unary "relation symbols".
We define the following $\PNFkEP{k}\sigres{(k+d)}$-"sentence" $\fml$
and the following "structure" $\struc$:
\begin{align*}
\fml
&\defeq
\exists x_1. \dots \exists x_k. \bigwedge_{i = 1}^{k} \hat{S}_i(x_i) \land \bigwedge_{\ell = 1}^{d} \bigvee_{i = 1}^{k} Z_{\ell}(x_i), \span \span\span \span\span\\
\univ{\struc} & \defeq \range{k} \times \range{n}, &
\hat{S}_i^{\struc} & \defeq \set{i} \times \range{n} \text{ for $i \in \range{k}$, } & 
Z_{\ell}^{\struc} & \defeq \set{\tuple{i, j} \in \univ{\struc} \mid v_{i,j}\ith{\ell} = 0} \text{ for $\ell \in \range{d}$. }
\end{align*}
Note that $\len{\fml} = \asympO(k d)$, $\strucn{\struc} = k n$, and $\strucm{\struc} = \asympO(k d n)$.
Also, $\fml$ and $\struc$ are constructed in $\asympO(k d n)$ time.

Intuitively, each $\tuple{i, j} \in \univ{\struc}$ represents the vector $v_{i, j}$,
$\hat{S}_i(x_i)$ expresses $u_i \in S_i$, and
$Z_{\ell}(x_i)$ expresses $u_i\ith{\ell} = 0$.
Therefore, the desired equivalence holds, which completes the proof.
\end{proof}

\begin{theorem*}[Restatement of \Cref{theorem: FO model checking hardness}]\leavevmode\\
\theoremFOmodelcheckinghardness
\end{theorem*}
\begin{proof}
Suppose that for some fixed $\eps > 0$,
$\MC{\PNFkEP{k}\sigres{(\ordinalomega)}}$ is in $\poly(\len{\fml}) \cdot m^{k - \eps}$ time.
Without loss of generality, assume that $\eps$ is sufficiently small, say $\eps < 1$.
By the reduction of \Cref{lemma: reduction from kOV to PNF-FO},
we can solve $k$-\OVProblem{n, d} in time
$\asympO(k d n) + \poly(\asympO(k d)) \cdot \asympO(k d n)^{k - \eps}
= \poly(k d) \cdot n^{k - \eps}$ (by $k - \eps > 1$),
and hence in $\poly(d) \cdot n^{k - \eps}$ time as $k$ is fixed,
which contradicts the $k$-\OVHypothesis{}.
\end{proof}
Additionally, the low-"dimensional" $k$-\OVHypothesis{}
implies the following stronger result.
\begin{theorem}\label{theorem: FO model checking hardness low-dimensional}
    Fix $k \ge 2$.
    Assuming the low-"dimensional" $k$-\OVHypothesis{},
    $\MC{\PNFkEP{k}\sigres{(\ordinalomega)}}$
    cannot be solved in $2^{\asympo(\len{\fml})} \cdot m^{k - \eps}$ time
    for any $\eps > 0$.
\end{theorem}
\begin{proof}
Suppose that for some fixed $\eps > 0$,
$\MC{\PNFkEP{k}\sigres{(\ordinalomega)}}$ is in $2^{\asympo(\len{\fml})} \cdot m^{k - \eps}$ time.
Without loss of generality, assume that $\eps$ is sufficiently small, say $\eps < 1$.
By the reduction of \Cref{lemma: reduction from kOV to PNF-FO},
we can solve $k$-\OVProblem{n, d} in time
$\asympO(k d n) + 2^{\asympo(\asympO(k d))} \cdot \asympO(k d n)^{k - \eps}
= 2^{\asympo(k d)} \cdot n^{k - \eps}$ (by $k - \eps > 1$),
and hence in $2^{\asympo(d)} \cdot n^{k - \eps}$ time as $k$ is fixed.
Hence, for every $c > 0$,
$k$-\OVProblem{n, \le c \log n} is in time $2^{\asympo(\log n)} \cdot n^{k - \eps} = n^{k - \eps + \asympo(1)} \le \asympO(n^{k - \eps / 2})$,
contradicting the low-"dimensional" $k$-\OVHypothesis{}.
Note that the "sentence" $\fml$ constructed in \Cref{lemma: reduction from kOV to PNF-FO} is already "tree-represented" with $\len{\fml} = \asympO(k d)$;
hence the same lower bound holds even when restricted to "tree-represented" "sentences".
\end{proof}

\section{Proof of \texorpdfstring{\Cref{theorem: FO model checking hardness one binary}}{Theorem~\ref{theorem: FO model checking hardness one binary}}:
\texorpdfstring{$k$-\OVProblem{}}{k-OV} to \texorpdfstring{$\MC{\kEP{k}\sigres
{(0,1)}}$}{EPk[(0,1)]-MC}}\label{section: theorem: FO model checking hardness one binary}

\begin{theorem*}[Restatement of \Cref{theorem: FO model checking hardness one binary}]\leavevmode\\
\theoremFOmodelcheckinghardnessonebinary
\end{theorem*}
\begin{proof}
    Suppose that for some fixed $\eps > 0$,
$\MC{\kEP{k}\sigres{(0, 1)}}$ is in $\poly(\len{\fml}) \cdot m^{k - \eps}$ time.
By the reduction of \Cref{lemma: reduction monadic and binary predicates to single binary predicate} (with $C = \asympO(\len{\fml})$, after removing isolated "elements" by \Cref{proposition: isolated}; for "tree-represented" "sentences", we use \Cref{lemma: reduction monadic and binary predicates to single binary predicate for tree} instead),
$\MC{\kEP{k}\sigres{(\ordinalomega)}}$ is in time $\poly(\len{\fml}) \cdot (\len{\fml} \cdot m)^{k - \eps}
= \poly(\len{\fml}) \cdot m^{k - \eps}$,
which contradicts \Cref{theorem: FO model checking hardness} under the $k$-\OVHypothesis{}.
\end{proof}
Additionally, the low-"dimensional" $k$-\OVHypothesis{}
implies the following stronger result.
\begin{theorem}\label{theorem: FO model checking hardness low-dimensional binary low-dimensional}
    Fix $k \ge 3$.
    Assuming the low-"dimensional" $k$-\OVHypothesis{},
    $\MC{\kEP{k}\sigres{(0, 1)}}$
    cannot be solved in $2^{\asympo(\len{\fml})} \cdot m^{k - \eps}$ time for any $\eps > 0$.
    For "tree-represented" "sentences",
    assuming the low-"dimensional" $k$-\OVHypothesis{},
    $\MC{\kEP{k}\sigres{(0, 1)}}$
    cannot be solved in $2^{\asympo(\len{\fml}/\log \len{\fml})} \cdot m^{k - \eps}$ time
    for any $\eps > 0$.
\end{theorem}
\begin{proof}
For "DAG-represented" inputs,
suppose that for some fixed $\eps > 0$,
$\MC{\kEP{k}\sigres{(0, 1)}}$ is in $2^{\asympo(\len{\fml})} \cdot m^{k - \eps}$ time.
By the reduction of \Cref{lemma: reduction monadic and binary predicates to single binary predicate} (as in the proof of \Cref{theorem: FO model checking hardness one binary}),
$\MC{\kEP{k}\sigres{(\ordinalomega)}}$ on "DAG-represented" inputs is in time $2^{\asympo(\len{\fml})} \cdot (\len{\fml} \cdot m)^{k - \eps}
= 2^{\asympo(\len{\fml})} \cdot m^{k - \eps}$,
which contradicts \Cref{theorem: FO model checking hardness low-dimensional} under the low-"dimensional" $k$-\OVHypothesis{}.

For "tree-represented" inputs, 
suppose that for some fixed $\eps > 0$,
$\MC{\kEP{k}\sigres{(0, 1)}}$ is in $2^{\asympo(\len{\fml} / \log \len{\fml})} \cdot m^{k - \eps}$ time.
By the reduction of \Cref{lemma: reduction monadic and binary predicates to single binary predicate for tree} (with $L = \asympO(\log \len{\fml})$),
$\MC{\kEP{k}\sigres{(\ordinalomega)}}$ on "tree-represented" inputs is in time $2^{\asympo(\len{\fml} \log \len{\fml} / \log (\len{\fml} \log \len{\fml}))} \cdot (\log \len{\fml} \cdot m)^{k - \eps}
= 2^{\asympo(\len{\fml})} \cdot m^{k - \eps}$,
which contradicts \Cref{theorem: FO model checking hardness low-dimensional} under the low-"dimensional" $k$-\OVHypothesis{}.
\end{proof}
\section{Unconditional Lower Bounds for \texorpdfstring{$\kFO{2}\sigres{(\ordinalomega,\ordinalomega)}$}{FO2[(∞,∞)]} to \texorpdfstring{$\BMLP$}{ML*} and \texorpdfstring{$\kFO{k}\sigres{(\ordinalomega)}$}{FOk[(∞)]} to \texorpdfstring{$\kC{1}$}{C1}}\label{section: succinctness unconditional}
In this section, we note that there is no polynomial-time translation from $\kFO{2}\sigres{(\ordinalomega,\ordinalomega)}$ to $\BMLP$
and from $\kFO{k}\sigres{(\ordinalomega)}$ ($k \ge 2$) to $\kC{1}$, without any assumption.

For $\kFO{2}\sigres{(\ordinalomega,\ordinalomega)}$ to $\BMLP$, and moreover for $\kFO{2}\sigres{(\ordinalomega)}$ to $\BMLP$,
there is no polynomial-size translation (and hence, no polynomial-time translation),
because there is a trivial polynomial-time translation from $\BMLP$ to $\MLLSW$
and there is no polynomial-size translation from $\kFO{2}\sigres{(\ordinalomega)}$ to $\MLLSW$ \cite[Theorem 2]{lutzModalLogicTwoVariable2001} (also for "DAG-represented" formulas).

For $\kFO{k}\sigres{(\ordinalomega)}$ to $\kC{1}$ where $k \ge 2$,
there is no polynomial-time translation without any assumption,
because the satisfiability problem for $\kFO{k}\sigres{(\ordinalomega)}$ is "NEXPTIME"-hard \cite[Theorem 4.1]{lewisComplexityResultsClasses1980} but that for $\kC{1}$ is in "NP" (even if the numbers occurring in counting quantifiers are encoded in binary) \cite{kuncakEfficientSatisfiabilityChecking2007}, \cite[Theorem 1]{pratt-hartmannComputationalComplexityNumerically2008} (also for "DAG-represented" formulas),
with "NP" $\subsetneq$ "NEXPTIME" \cite{cookHierarchyNondeterministicTime1973}.
 
\ifdefined\submission\else
\clearpage
\clearpage

\fi
\fi

\end{document}